\documentclass{lmcs}
\pdfoutput=1

\usepackage{lastpage}
\lmcsdoi{19}{4}{6}
\lmcsheading{}{\pageref{LastPage}}{}{}%
{Mar.~17,~2022}{Oct.~25,~2023}{}


\usepackage[T1]{fontenc}
\usepackage[utf8]{inputenc}
\usepackage{numprint}
\usepackage{graphicx}
\usepackage{caption}
\usepackage{subcaption}
\usepackage{amsmath}
\usepackage{stackrel}
\usepackage{amsthm}
\usepackage{mathrsfs}
\usepackage{amssymb}
\usepackage{amsfonts}
\usepackage{wasysym}
\usepackage{stmaryrd}
\usepackage{tikz}
\usepackage[ruled, vlined, linesnumbered]{algorithm2e}
\usepackage{stmaryrd}
\usepackage{enumitem}
\usetikzlibrary{automata}
\usetikzlibrary{shapes}
\usepackage{hyperref}

\newcommand{\<}{\langle}
\renewcommand{\>}{\rangle}

\newcommand{\bsigma}{\bar{\sigma}}
\newcommand{\btau}{\bar{\tau}}

\renewcommand{\|}{\upharpoonright}
\newcommand{\hr}{\hat{r}}
\newcommand{\hmu}{\hat{\mu}}

\newcommand{\N}{\mathbb{N}}
\newcommand{\Z}{\mathbb{Z}}
\newcommand{\R}{\mathbb{R}}
\newcommand{\Q}{\mathbb{Q}}
\newcommand{\bR}{\overline{\mathbb{R}}}

\newcommand{\playcircle}{{\scriptsize{\Circle}}}

\renewcommand{\l}{\ell}
\renewcommand{\epsilon}{\varepsilon}
\renewcommand{\phi}{\varphi}

\newcommand{\opt}{\mathsf{opt}}

\newcommand{\SC}{\mathsf{SC}}
\newcommand{\Plays}{\mathsf{Plays}}
\newcommand{\Hist}{\mathsf{Hist}}

\newcommand{\Conv}{\mathsf{Conv}}

\newcommand{\nego}{\mathsf{nego}}
\newcommand{\lCons}{\lambda\mathsf{Cons}}
\newcommand{\lRat}{\lambda\mathsf{Rat}}

\newcommand{\Prop}{\mathsf{Prop}}
\newcommand{\Acc}{\mathsf{Acc}}
\newcommand{\Dev}{\mathsf{Dev}}
\newcommand{\ML}{\mathsf{ML}}

\newcommand{\Rat}{\mathsf{Rat}}
\newcommand{\Abs}{\mathsf{Abs}}
\newcommand{\Conc}{\mathsf{Conc}}

\newcommand{\val}{\mathsf{val}}

\newcommand{\last}{\mathsf{last}}
\newcommand{\first}{\mathsf{first}}
\newcommand{\Occ}{\mathsf{Occ}}
\newcommand{\Inf}{\mathsf{Inf}}
\newcommand{\MP}{\mathsf{MP}}

\newcommand{\SConn}{\mathsf{SConn}}

\newcommand{\Aff}{\mathsf{Aff}}
\newcommand{\Sum}{\mathsf{sum}}

\newcommand{\dseal}{\!\,^\llcorner}

\renewcommand{\P}{\mathbb{P}}
\newcommand{\C}{\mathbb{C}}
\renewcommand{\S}{\mathbb{S}}

\newcommand{\Mem}{\mathsf{Mem}}

\newcommand{\NP}{\mathbf{NP}}
\newcommand{\ExpTime}{\mathbf{ExpTime}}

\newcommand{\bx}{\bar{x}}
\renewcommand{\by}{\bar{y}}
\newcommand{\bz}{\bar{z}}

\newcommand{\balpha}{\bar{\alpha}}
\newcommand{\bbeta}{\bar{\beta}}

\newcommand{\ba}{\bar{a}}

\newcommand{\blambda}{\lambda \mspace{-10mu} \bar{\phantom{v}}}

\newcommand{\bzero}{\bar{0}}

\newcommand{\dH}{\dot{H}}

\newcommand{\dc}{\dot{c}}
\newcommand{\dpi}{\dot{\pi}}
\newcommand{\dchi}{\dot{\chi}}
\newcommand{\dxi}{\dot{\xi}}

\begin{document}

\title{Subgame-perfect Equilibria in Mean-payoff Games}
\titlecomment{This paper is an enhanced version of \cite{Concur}.
All the proofs are now included in the main body of the paper, and some of them have been rewritten to be more readable, especially those of Theorems~\ref{thm_piecewise_affine} and \ref{thm_decidable}.
The proof of Theorem~\ref{thm_spe} has been modified in order to relax its hypotheses, \emph{via} a change in the definition of \emph{steady negotiation} (Definition~\ref{def_nego}).
We have also taken advantage of the new available space to add examples in the main body of this paper, especially Example~\ref{ex_propagation}.}
\bibliographystyle{alphaurl}

\author[L.~Brice]{Léonard Brice\lmcsorcid{0000-0001-7748-7716}}[a]
\author[J.~Raskin]{Jean-François Raskin\lmcsorcid{0000-0002-3673-1097}}[a]
\author[M.~van~den~Bogaard]{Marie van den Bogaard\lmcsorcid{0009-0007-2070-1196}}[b]

\address{Université Libre de Bruxelles, Faculté des Sciences
Campus de la Plaine - CP 212
Boulevard du Triomphe, ACC.2
1050 Bruxelles,
Belgique}
\email{leonard.brice@ulb.be, jraskin@ulb.be}

\address{Université Gustave Eiffel,
bâtiment Copernic,
5 boulevard Descartes
77420 Champs-sur-Marne,
France}
\email{marie.van-den-bogaard@univ-eiffel.fr}

\begin{abstract}
\noindent
In this paper, we provide an effective characterization of all the subgame-perfect equilibria in infinite duration games played on finite graphs with mean-payoff objectives.
To this end, we introduce the notion of requirement, and the notion of negotiation function.
We establish that the plays that are supported by SPEs are exactly those that are consistent with a fixed point of the negotiation function.
Finally, we use that characterization to prove that the SPE threshold problem, whose status was left open in the literature, is decidable.
\end{abstract}

\maketitle

\section{Introduction}
The notion of Nash equilibrium (NE) is one of the most important and most studied solution concepts in game theory.
A profile of strategies is an NE when no rational player has an incentive to change their strategy unilaterally, i.e. while the other players keep their strategies.
Thus an NE models a stable situation. 
Unfortunately, it is well known that, in sequential games, NEs suffer from the problem of {\em non-credible threats}, see e.g.~\cite{Osborne}. In those games, some NEs only exist when some players do {\em not} play rationally in subgames and so use non-credible threats to force the NE.
This is why, in sequential games, the stronger notion of {\em subgame-perfect equilibrium} is used instead: a profile of strategies is a subgame-perfect equilibrium (SPE) if it is an NE in all the subgames of the sequential game.
Thus SPEs impose rationality even after a deviation has occured.

In this paper, we study sequential games that are infinite-duration games played on graphs with mean-payoff objectives, and focus on SPEs.
While NEs are guaranteed to exist in infinite duration games played on graphs with mean-payoff objectives, it is known that it is not the case for SPEs, see e.g.~\cite{solan2003deterministic,DBLP:conf/csl/BrihayeBMR15}.
We provide in this paper a constructive characterization of the entire set of SPEs, which allows us to decide, among others, the SPE threshold problem.
This problem was left open in previous contributions on the subject.
More precisely, our contributions are described in the next paragraphs.

\subsection{Contributions}

First, we introduce two important new notions that allow us to capture NEs, and more importantly SPEs, in infinite duration games played on graphs with mean-payoff objectives\footnote{A large part of our results apply to the larger class of games with prefix-independent objectives. For the sake of readability of this introduction, we focus here on mean-payoff games but the technical results in the paper usually cover broader classes of games.}: the notion of {\em requirement} and the notion of {\em negotiation function}. 

A requirement $\lambda$ is a function that assigns to each vertex $v \in V$ of a game graph a value in $\mathbb{R} \cup \{ -\infty, +\infty\}$.
The value $\lambda(v)$ represents a requirement on any play $\rho = \rho_0 \rho_1 \dots \rho_n \dots$ that traverses this vertex: if one wants the player who controls the vertex $v$ to follow $\rho$ and to give up deviating from $\rho$, then the play must offer to that player a payoff that is at least $\lambda(v)$.
An infinite play $\rho$ is \emph{$\lambda$-consistent} if, for each player~$i$, the payoff of $\rho$ for player~$i$ is larger than or equal to the largest value of $\lambda$ on vertices occurring along $\rho$ and controlled by player~$i$.

We first use these notions to rephrase a classical result about NEs: if $\lambda$ maps each vertex $v$ to the largest value that the player who controls $v$ can secure against a fully adversarial coalition of the other players, i.e. if $\lambda(v)$ is the zero-sum worst-case value, then the set of plays that are $\lambda$-consistent is exactly the set of plays that are supported by an NE (Theorem~\ref{thm_ne}). 

As SPEs are forcing players to play rationally in all subgames, we cannot rely on the zero-sum worst-case value to characterize them.
Indeed, when considering the worst-case value, we allow adversaries to play fully adversarially after a deviation and so potentially in an irrational way w.r.t. their own objective. 
In fact, in an SPE, a player is refrained to deviate when opposed by a coalition of {\em rational adversaries}. 
To characterize this relaxation of the notion of worst-case value, we rely on our notion of \emph{negotiation function}.

The negotiation function $\nego$ operates from the set of requirements into itself.
To understand the purpose of the negotiation function, let us consider its application on the requirement $\lambda$ that maps every vertex $v$ to the worst-case value as above.
Now, we can naturally formulate the following question: given $v$ and $\lambda$, can the player who controls $v$ improve the value that they can ensure against all the other players, if only plays that are consistent with $\lambda$ are proposed by the other players?
In other words, can this player enforce a better value when playing against the other players if those players are not willing to give away their own worst-case value?
Clearly, securing this worst-case value can be seen as a minimal goal for any {\em rational} adversary.
So $\nego(\lambda)(v)$ returns this value; and this reasoning can be iterated.
One of the contributions of this paper is to show that the least fixed point $\lambda^*$ of the negotiation function is exactly characterizing the set of plays supported by SPEs (Theorem~\ref{thm_spe}).

To turn this fixed point characterization of SPEs into algorithms, we additionally draw links between the negotiation function and two classes of zero-sum games, that are called {\em abstract} and {\em concrete} (see Theorem~\ref{thm_concrete_game}) negotiation games.
The abstract negotiation game is conceptually simple but is played on an uncountably infinite graph, and therefore cannot be turned into an effective algorithm.
However, it captures the intuition behind the concrete negotiation game, which is played on a finite graph.
We show that in the concrete negotiation game, one of the players has a memoryless optimal strategy (Lemma~\ref{lm_concrete_memoryless}), which can be used to solve it effectively.
Thus, the negotiation function is computable.
However, that is not sufficient to compute its least fixed point, because the sequence of Kleene-Tarski's iterations may require a transfinite number of steps to reach a fixed point (Theorem~\ref{thm_not_stationary}).

Nevertheless, we prove that the concrete negotiation game can be used to construct a geometrical representation of the fixed points of the negotiation function, from which one can extract its least fixed point (Theorem~\ref{thm_piecewise_affine}).
Thus, the SPE threshold problem is decidable (Theorem~\ref{thm_decidable}).

All the previous results do also apply to $\epsilon$-SPEs, a classical quantitative relaxation of SPEs --- see for example~\cite{DBLP:journals/ijgt/FleschP16}.
In particular, Theorem~\ref{thm_decidable} does also apply to the $\epsilon$-SPE threshold problem.

\subsection{Related works}

Non-zero sum infinite duration games have attracted a large attention in recent years, with applications targeting reactive synthesis problems. We refer the interested reader to the following survey papers~\cite{DBLP:conf/lata/BrenguierCHPRRS16,DBLP:conf/dlt/Bruyere17} and their references for the relevant literature. We detail below contributions more closely related to the work presented here.

In~\cite{DBLP:conf/lfcs/BrihayePS13}, Brihaye et al. offer a characterization of NEs in quantitative games for cost-prefix-linear reward functions based on the worst-case value. The mean-payoff is cost-prefix-linear. In their paper, the authors do not consider the stronger notion of SPE, which is the central solution concept studied in our paper.
In \cite{DBLP:conf/csl/BruyereMR14}, Bruy{\`{e}}re et al. study secure equilibria that are a refinement of NEs. Secure equilibria are not subgame-perfect and are, as classical NEs, subject to non-credible threats in sequential games.

In~\cite{DBLP:conf/fsttcs/Ummels06}, Ummels proves that there always exists an SPE in games with $\omega$-regular objectives and defines algorithms based on tree automata to decide constrained SPE problems.
Strategy logics, see e.g.~\cite{DBLP:journals/iandc/ChatterjeeHP10}, can be used to encode the concept of SPE in the case of $\omega$-regular objectives with application to the rational synthesis problem~\cite{DBLP:journals/amai/KupfermanPV16} for instance.  
In~\cite{DBLP:journals/mor/FleschKMSSV10}, Flesch et al. show that the existence of $\epsilon$-SPEs is guaranteed when the reward function is {\em lower-semicontinuous}. The mean-payoff reward function is neither $\omega$-regular, nor lower-semicontinuous, and so the techniques defined in the papers cited above cannot be used in our setting. Furthermore, as already recalled above, see e.g.~\cite{vieille:hal-00464953,DBLP:conf/csl/BrihayeBMR15}, contrary to the $\omega$-regular case, SPEs in games with mean-payoff objectives may fail to exist.

In~\cite{DBLP:conf/csl/BrihayeBMR15}, Brihaye et al. introduce and study the notion of weak subgame-perfect equilibria, which is a weakening of the classical notion of SPE. This weakening is equivalent to the original SPE concept on reward functions that are {\em continuous}.
This is the case for example for the quantitative reachability reward function, on which Brihaye et al. solve the SPE threshold problem in \cite{DBLP:conf/concur/BrihayeBGRB19}.
On the contrary, the mean-payoff cost function is not continuous and the techniques used in~\cite{DBLP:conf/csl/BrihayeBMR15}, and generalized in~\cite{DBLP:conf/fossacs/Bruyere0PR17}, cannot be used to characterize SPEs for the mean-payoff reward function.

In~\cite{thesis_noemie}, Meunier develops a method based on Prover-Challenger games to solve the problem of the existence of SPEs on games with a finite number of possible payoffs. This method is not applicable to the mean-payoff reward function, as the number of payoffs in this case is uncountably infinite.

In~\cite{DBLP:journals/mor/FleschP17}, Flesch and Predtetchinski present another characterization of SPEs on games with finitely many possible payoffs, based on a game structure that we will present here under the name of \emph{abstract negotiation game}.
Our contributions differ from this paper in two fundamental aspects. First, it lifts the restriction to finitely many possible payoffs. This is crucial as mean-payoff games violate this restriction. Instead, we identify a class of games, that we call {\em with steady negotiation}, that encompasses mean-payoff games and for which some of the conceptual tools introduced in that paper can be generalized. Second, the procedure developed by Flesch and Predtetchinski is not an algorithm in computer science acceptation: it needs to solve infinitely many games that are not represented effectively, and furthermore it needs a transfinite number of iterations. On the contrary, our procedure is effective and leads to a complete algorithm in the classical sense: with guarantee of termination in finite time and applied on effective representations of games.

In \cite{DBLP:conf/concur/ChatterjeeDEHR10}, Chaterjee et al. study mean-payoff \emph{automata}, and give a result that can be translated into an expression of all the possible payoff vectors in a mean-payoff game.
In \cite{DBLP:conf/cav/BrenguierR15}, Brenguier and Raskin give an algorithm to build the Pareto curve of a multi-dimensional two-player zero-sum mean-payoff game.
Techniques defined in these papers are used in several technical steps of our algorithm.

\subsection{Structure of the paper}

In Section~\ref{sec_background}, we introduce the necessary background.
Section~\ref{sec_negotiation} defines the notion of requirement and the negotiation function.
Section~\ref{sec_link_nego_spe} shows that the set of plays that are supported by an SPE are those that are $\lambda$-consistent, where $\lambda$ is a fixed point of the negotiation function.
Section~\ref{sec_abstract_game} draws a link between the negotiation function and the abstract negotiation game.
Section~\ref{sec_concrete_game} shows that the abstract negotiation game can be transformed into a game on a finite graph, the concrete negotiation game, which can be solved to compute effectively the negotiation function.
Section~\ref{sec_analysis} uses the concrete negotiation game to prove that the negotiation function is a piecewise affine function, of which one can compute an effective representation.
Finally, Section~\ref{sec_conclusion} applies these results to prove that the SPE threshold problem in mean-payoff games is decidable, $2\ExpTime$-easy and $\NP$-hard.

	\section{Background} \label{sec_background}

        \subsection{Games, strategies, equilibria}

In all what follows, we will use the word \emph{game} for the infinite duration turn-based quantitative games on finite graphs with complete information.

\begin{defi}[Non-initialized game]\label{defi_game}
	A \emph{non-initialized game} --- or \emph{game} for short --- is a tuple
	$G = \left(\Pi, V, (V_i)_{i \in \Pi}, E, \mu\right)$, where:
	
	\begin{itemize}
		\item $\Pi$ is a finite set of \emph{players};
		
		\item $(V, E)$ is a directed graph, called the \emph{underlying graph} of $G$, whose vertices are sometimes called \emph{states} and whose edges are sometimes called \emph{transitions}, and in which every state has at least one outgoing transition.
		For the simplicity of writing, a transition $(v, w) \in E$ will often be written $vw$.
		
		\item $(V_i)_{i \in \Pi}$ is a partition of $V$, in which $V_i$ is the set of states \emph{controlled} by player $i$;
		
		\item $\mu: V^\omega \to \R^\Pi$ is an \emph{payoff function}, that maps each infinite word $\rho$ to the tuple $\mu(\rho) = (\mu_i(\rho))_{i \in \Pi}$ of the players' \emph{payoffs}.
	\end{itemize}
\end{defi}

\begin{defi}[Initialized game]
	An \emph{initialized game} --- or game for short --- is a tuple $(G, v_0)$, often written $G_{\|v_0}$, where $G$ is a non-initialized game and $v_0 \in V$ is a state called \emph{initial state}.
\end{defi}

When the context is clear, we often use the word \emph{game} for both initialized and non-initialized games.

\begin{defi}[Play, history]
	A \emph{play} (resp. history) in the game $G$ is an infinite (resp. finite) path in the graph $(V, E)$.
	It is also a play (resp. history) in the initialized game $G_{\|v_0}$, where $v_0$ is its first vertex.
	The set of plays (resp. histories) in the game $G$ (resp. the initialized game $G_{\|v_0}$) is denoted by $\Plays G$ (resp. $\Plays G_{\|v_0}, \Hist G, \Hist G_{\|v_0}$).
	We write $\Hist_i G$ (resp. $\Hist_i G_{\|v_0}$) for the set of histories in $G$ (resp. $G_{\|v_0}$) of the form $hv$, where $v$ is a vertex controlled by player $i$.
\end{defi}

Given a play $\rho$ (resp. a history $h$), we write $\Occ(\rho)$ (resp. $\Occ(h)$) the set of vertices that appear in $\rho$ (resp. $h$), and $\Inf(\rho)$ the set of vertices that appear infinitely often in $\rho$.
For a given index $k$, we write $\rho_{\leq k}$ (resp. $h_{\leq k}$), or $\rho_{< k+1}$ (resp. $h_{<k+1}$), the finite prefix $\rho_0 \dots \rho_k$ (resp. $h_0 \dots h_k$), and $\rho_{\geq k}$ (resp. $h_{\geq k}$), or $\rho_{> k-1}$ (resp. $h_{>k-1}$), the infinite (resp. finite) suffix $\rho_k \rho_{k+1} \dots$ (resp. $h_k h_{k+1} \dots h_{|h|-1}$).
Finally, we write $\first(\rho)$ (resp. $\first(h)$) the first vertex of $\rho$, and $\last(h)$ the last vertex of $h$.

\begin{defi}[Strategy, strategy profile]
	A \emph{strategy} for player $i$ in the initialized game $G_{\|v_0}$ is a function $\sigma_i: \Hist_i G_{\|v_0} \to V$, such that $v\sigma_i(hv)$ is an edge of $(V, E)$ for every $hv$.
	A history $h$ is \emph{compatible} with a strategy $\sigma_i$ if and only if $h_{k+1} = \sigma_i(h_0 \dots h_k)$ for all $k$ such that $h_k \in V_i$. A play $\rho$ is compatible with $\sigma_i$ if all its prefixes are.

	A \emph{strategy profile} for $P \subseteq \Pi$ is a tuple $\bsigma_P = (\sigma_i)_{i \in P}$, where each $\sigma_i$ is a strategy for player $i$ in $G_{\|v_0}$.
    A play or a history is \emph{compatible} with $\bsigma_P$ if it is compatible with every $\sigma_i$ for $i \in P$.
    A \emph{complete} strategy profile, usually written $\bsigma$, is a strategy profile for $\Pi$.
    Exactly one play is compatible with the strategy profile $\bsigma$: we call it its \emph{outcome} and write it $\< \bsigma \>$.
    
    When $i$ is a player and when the context is clear, we will often write $-i$ for the set $\Pi \setminus \{i\}$.
    We will often refer to $\Pi \setminus \{i\}$ as the \emph{environment} against player $i$.
    When $\btau_P$ and $\btau'_Q$ are two strategy profiles with $P \cap Q = \emptyset$, $(\btau_P, \btau'_Q)$ denotes the strategy profile $\bsigma_{P \cup Q}$ such that $\sigma_i = \tau_i$ for $i \in P$, and $\sigma_i = \tau'_i$ for $i \in Q$.
\end{defi}

In a strategy profile $\bsigma_P$, the $\sigma_i$'s domains are pairwise disjoint. Therefore, we can consider $\bsigma_P$ as one function: for $hv \in \Hist G_{\|v_0}$ such that $v \in \bigcup_{i \in P} V_i$, we liberally write $\bsigma_P(hv)$ for $\sigma_i(hv)$ with $i$ such that $v \in V_i$.

Before moving on to SPEs, let us recall the notion of Nash equilibrium.

\begin{defi}[Nash equilibrium]
	Let $G_{\|v_0}$ be a game. A strategy profile $\bsigma$ is a \emph{Nash equilibrium} --- or \emph{NE} for short --- in $G_{\|v_0}$ if and only if for each player $i$ and for every strategy $\sigma'_i$, called \emph{deviation of $\sigma_i$}, we have the inequality $\mu_i\left(\< \sigma'_i, \bsigma_{-i} \>\right) \leq \mu_i\left(\< \bsigma \>\right)$.
\end{defi}

To define SPEs, we need the notion of subgame.

\begin{defi}[Subgame, substrategy]
	Let $hv$ be a history in the game $G$. The \emph{subgame} of $G$ after $hv$ is the game $\left(\Pi, V, (V_i)_i, E, \mu_{\|hv}\right)_{\|v}$, where $\mu_{\|hv}$ maps each play to its payoff in $G$, assuming that the history $hv$ has already been played: formally, for every $\rho \in \Plays G_{\|hv}$, we have $\mu_{\|hv}(\rho) = \mu(h\rho)$.
    
    If $\sigma_i$ is a strategy in $G_{\|v_0}$, its \emph{substrategy} after $hv$ is the strategy $\sigma_{i\|hv}$ in $G_{\|hv}$, defined by $\sigma_{i\|hv}(h') = \sigma_i(hh')$ for every $h' \in \Hist_i G_{\|hv}$.
\end{defi}

\begin{rem}
    The initialized game $G_{\|v_0}$ is also the subgame of $G$ after the one-state history $v_0$.
\end{rem}

\begin{defi}[Subgame-perfect equilibrium]
	Let $G_{\|v_0}$ be a game.
	The strategy profile $\bsigma$ is a \emph{subgame-perfect equilibrium} --- or \emph{SPE} for short --- in $G_{\|v_0}$ if and only if for every history $h$ in $G_{\|v_0}$, the strategy profile $\bsigma_{\|h}$ is a Nash equilibrium in the subgame $G_{\|h}$.
\end{defi}

The notion of subgame-perfect equilibrium can be seen as a refinement of Nash equilibrium: it is a stronger equilibrium which excludes players resorting to non-credible threats.

\begin{exa}
In the game represented in Figure~\ref{fig_ne_spe}, where the square state is controlled by player $\Box$ and the round states by player $\Circle$, if both players get the payoff $1$ by reaching the state $d$ and $0$ in the other cases, there are actually two NEs: one, in blue, where $\Box$ goes to the state $b$ and then player $\Circle$ goes to $d$, and both win, and one, in red, where player $\Box$ goes to $c$ because player $\Circle$ was planning to go to $e$. However, only the blue one is an SPE, as moving from $b$ to $e$ is irrational for player $\Circle$ in the subgame $G_{\|ab}$.
\end{exa}

An $\epsilon$-SPE is a strategy profile which is \emph{almost} an SPE: if a player deviates after some history, they will not be able to improve their payoff by more than a quantity $\epsilon \geq 0$.

\begin{defi}[$\epsilon$-SPE]
	Let $G_{\|v_0}$ be a game, and $\epsilon \geq 0$. A strategy profile $\bsigma$ from $v_0$ is an $\epsilon$-SPE if and only if for every history $hv$, for every player $i$ and every strategy $\sigma'_i$, we have $\mu_i(\< \bsigma_{-i\|hv}, \sigma'_{i\|hv} \>) \leq \mu_i(\< \bsigma_{\|hv} \>) + \epsilon$.
\end{defi}

\begin{rem} \label{rem_0SPE}
    Note that a $0$-SPE is an SPE, and conversely.
\end{rem}

Hereafter, we focus on \emph{prefix-independent} games, and in particular \emph{mean-payoff} games.

        \subsection{Mean-payoff games}
	
	\begin{defi}[Mean-payoff, mean-payoff game]
	    In a graph $(V, E)$, we associate to each \emph{reward function} $r: E \to \Q$ the \emph{mean-payoff function}:
	    $$\MP_r: h_0 \dots h_n \mapsto \frac{1}{n} \underset{k=0}{\overset{n-1}{\sum}} r_i\left(h_k h_{k+1}\right).$$
		A game $G = \left(\Pi, V, (V_i)_i, E, \mu \right)$ is a \emph{mean-payoff game} if its underlying graph is finite, and if there exists a tuple $(r_i)_{i \in \Pi}$ of reward functions, such that for each player $i$ and every play $\rho$:
		$$\mu_i(\rho) = \liminf_{n \to \infty} \MP_{r_i}(\rho_{\leq n}).$$
	\end{defi}
	
In a mean-payoff game, the quantities given by the function $r_i$ represent the immediate reward that each action gives to player $i$.
The final payoff of player $i$ is their average payoff along the play, classically defined as the limit inferior over $n$ (since the limit may not be defined) of the average payoff after $n$ steps.
When the context is clear, we liberally write $\MP_i(h)$ for $\MP_{r_i}(h)$, and $\MP(h)$ for $(\MP_i(h))_i$, as well as $r(uv)$ for $(r_i(uv))_i$.

\begin{defi}[Prefix-independent game]
	A game $G$ is \emph{prefix-independent} if, for every history $h$ and for every play $\rho$, we have $\mu(h\rho) = \mu(\rho)$.
	We also say, in that case, that the payoff function $\mu$ is prefix-independent.
\end{defi}

\begin{figure}
    \begin{subfigure}[b]{0.35\textwidth}
        \centering
    	\begin{tikzpicture}[->,>=latex,shorten >=1pt, initial text={}, scale=0.75, every node/.style={scale=0.75}]
    	\node[initial above, state, rectangle] (a) at (0, 0) {$a$};
    	\node[state] (b) at (-2, -2) {$b$};
    	\node[state] (c) at (2, -2) {$c$};
    	\node[state, double] (d) at (-3, -4) {$d$};
    	\node[state] (e) at (-1, -4) {$e$};
    	\node[state] (f) at (1, -4) {$f$};
    	\node[state] (g) at (3, -4) {$g$};
    	\path[->, blue] (a) edge (b);
    	\path[->, red] (a) edge (c);
    	\path[->, red] (b) edge (e);
    	\path[->, blue] (b) edge (d);
    	\path[->, blue] (c) edge (f);
    	\path[->, red] (c) edge (g);
    	\path (d) edge [loop below] (d);
    	\path (e) edge [loop below] (e);
    	\path (f) edge [loop below] (f);
    	\path (g) edge [loop below] (g);
    	\end{tikzpicture}
    	\caption{Two NEs and one SPE}
    	\label{fig_ne_spe}
    \end{subfigure}
    \begin{subfigure}[b]{0.5\textwidth}
        \centering
		\begin{tikzpicture}[->,>=latex,shorten >=1pt, initial text={}, scale=0.8, every node/.style={scale=0.8}]
		\node[state] (a) at (0, 0) {$a$};
		\node[state] (c) at (-2, 0) {$c$};
		\node[state, rectangle] (b) at (2, 0) {$b$};
		\node[state, rectangle] (d) at (4, 0) {$d$};
		\path[->] (a) edge (c);
		\path[<->] (a) edge node[above] {$\stackrel{\playcircle}{0} \stackrel{\Box}{3}$} (b);
		\path[->] (b) edge (d);
		\path (d) edge [loop right] node {$\stackrel{\playcircle}{2} \stackrel{\Box}{2}$} (d);
		\path (c) edge [loop left] node {$\stackrel{\playcircle}{1} \stackrel{\Box}{1}$} (c);
		
		\node[red] (l0) at (-3, -0.7) {$(\lambda_0)$};
		\node[red] (l0a) at (0, -0.7) {$-\infty$};
		\node[red] (l0b) at (2, -0.7) {$-\infty$};
		\node[red] (l0c) at (-2, -0.7) {$-\infty$};
		\node[red] (l0d) at (4, -0.7) {$-\infty$};
		
		\node[red] (l1) at (-3, -1.4) {$(\lambda_1)$};
		\node[red] (l1a) at (0, -1.4) {$1$};
		\node[red] (l1b) at (2, -1.4) {$2$};
		\node[red] (l1c) at (-2, -1.4) {$1$};
		\node[red] (l1d) at (4, -1.4) {$2$};
		
		\node[red] (l2) at (-3, -2.1) {$(\lambda_2)$};
		\node[red] (l2a) at (0, -2.1) {$2$};
		\node[red] (l2b) at (2, -2.1) {$2$};
		\node[red] (l2c) at (-2, -2.1) {$1$};
		\node[red] (l2d) at (4, -2.1) {$2$};
		
		\node[red] (l3) at (-3, -2.8) {$(\lambda_3)$};
		\node[red] (l3a) at (0, -2.8) {$2$};
		\node[red] (l3b) at (2, -2.8) {$3$};
		\node[red] (l3c) at (-2, -2.8) {$1$};
		\node[red] (l3d) at (4, -2.8) {$2$};
		
		\node[red] (l4) at (-3, -3.5) {$(\lambda_4)$};
		\node[red] (l4a) at (0, -3.5) {$+\infty$};
		\node[red] (l4b) at (2, -3.5) {$+\infty$};
		\node[red] (l4c) at (-2, -3.5) {$1$};
		\node[red] (l4d) at (4, -3.5) {$2$};
		\end{tikzpicture}
		\caption{A game without SPE} 
		\label{fig_sans_spe}
    \end{subfigure}
    \caption{Two examples of games}
\end{figure}
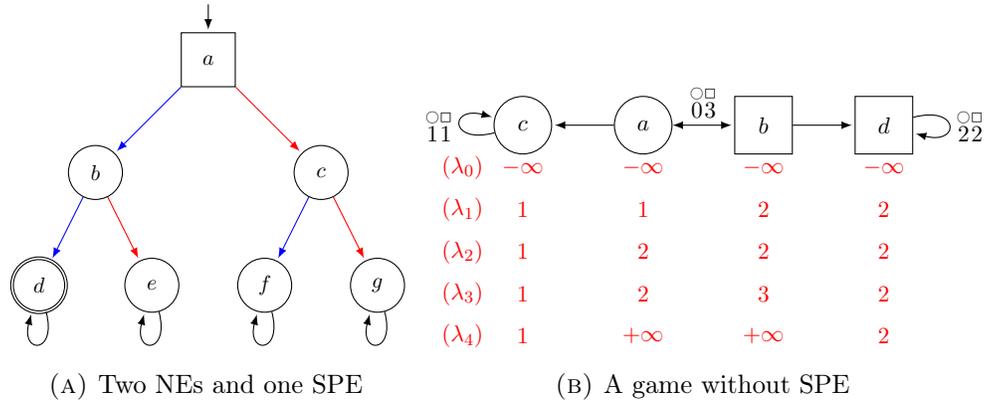

Mean-payoff games are prefix-independent.
A first important result that we need is the characterization of the set of possible payoffs in a mean-payoff game, which has been introduced in \cite{DBLP:conf/concur/ChatterjeeDEHR10}.
Given a graph $(V, E)$, we write $\SC(V, E)$ the set of simple cycles it contains.
Given a finite set $D$ of dimensions and a set $X \subseteq \R^D$, we write $\Conv X$ the convex hull of $X$.
We will often use the subscript notation $\Conv_{x \in X} f(x)$ for the set $\Conv f(X)$.

\begin{defi}[Downward sealing] \label{def_dseal}
    Given a set $Y \subseteq \R^D$, the \emph{downward sealing} of $Y$ is the set $\dseal Y = \left\{\left. \left( \min_{\bz \in Z} z_d \right)_{d \in D} ~\right|~ Z \mathrm{~is~a~finite~subset~of~} Y \right\}.$
\end{defi}

\begin{exa}
    In $\R^2$, if $Y$ is the blue area in Figure~\ref{fig_inf_spe_payoffs}, then $\dseal Y$ is obtained by adding the gray area.
\end{exa}

\begin{lemC}[\cite{DBLP:conf/concur/ChatterjeeDEHR10}]\label{lm_dseal}
    Let $G$ be a mean-payoff game, whose underlying graph is strongly connected.
    The set of the payoffs $\mu(\rho)$, where $\rho$ is a play in $G$, is exactly the set:
    $$\dseal\left( \underset{c \in \SC(V, E)}{\Conv} \MP(c) \right).$$
\end{lemC}

        \subsection{The \texorpdfstring{$\epsilon$}{ε}-SPE threshold problem}

In the sequel, we prove the decidability of the \emph{$\epsilon$-SPE threshold problem}, which is a generalization of the \emph{SPE threshold problem} (since SPEs are $0$-SPEs and conversely, by Remark~\ref{rem_0SPE}), defined as follows.

\begin{defi}[$\epsilon$-SPE threshold problem]
    Given a rational number $\epsilon \geq 0$, a mean-payoff game $G_{\|v_0}$ and two thresholds $\bx, \by \in \Q^\Pi$, does there exist an $\epsilon$-SPE $\bsigma$ in $G_{\|v_0}$ such that $\bx \leq \mu(\< \bsigma \>) \leq \by$?
\end{defi}

That problem is illustrated by the two following examples.

\begin{exa}[A game without SPEs]\label{ex_sans_spe}
	Let $G$ be the mean-payoff game of Figure~\ref{fig_sans_spe}, where each edge is labelled by the rewards $r_\playcircle$ and $r_\Box$. No reward is given for the edges $ac$ and $bd$ since they can be used only once, and therefore do not influence the final payoff.
	For now, the reader should not pay attention to the red labels below the states.
	As shown in~\cite{DBLP:journals/corr/Bruyere0PR16}, this game does not have any SPE, neither from the state $a$ nor from the state $b$.
	
	Indeed, the only NE outcomes from the state $b$ are the plays where player $\Box$ eventually leaves the cycle $ab$ and goes to $d$: if he stays in the cycle $ab$, then player $\Circle$ would be better off leaving it, and if she does, player $\Box$ would be better off leaving it before.
    From the state $a$, if player $\Circle$ knows that player $\Box$ will leave, she has no incentive to do it before: there is no NE where $\Circle$ leaves the cycle and $\Box$ plans to do it if ever she does not. Therefore, there is no SPE where $\Circle$ leaves the cycle.
    But then, after a history that terminates in $b$, player $\Box$ has actually no incentive to leave if player $\Circle$ never plans to do it afterwards: contradiction.
\end{exa}

\begin{figure}
	\begin{subfigure}[b]{0.5\textwidth}
	    \centering
		\begin{tikzpicture}[->,>=latex,shorten >=1pt, initial text={}, scale=0.8, every node/.style={scale=0.8}]
		\node[state] (a) at (0, 0) {$a$};
		\node[state, rectangle] (b) at (3, 0) {$b$};
		\path[<->] (a) edge node[above, rectangle] {$\stackrel{\playcircle}{2}\stackrel{\Box}{2}$} (b);
		\path (a) edge [loop left] node {$\stackrel{\playcircle}{0}\stackrel{\Box}{1}$} (a);
		\path (b) edge [loop right] node {$\stackrel{\playcircle}{1}\stackrel{\Box}{0}$} (b);
		\end{tikzpicture}
		\caption{The game $G$}
		\label{fig_inf_spe}
	\end{subfigure}
	\begin{subfigure}[b]{0.3\textwidth}
	    \centering
		\begin{tikzpicture}[scale=0.9]
		\draw [->] (0,0) -- (2.3,0);
		\draw (2.3,0) node[right] {${\Circle}$};
		\draw [->] (0,0) -- (0,2.3);
		\draw (0,2.3) node[above] {${\Box}$};
		\fill [blue] (0, 1) -- (1, 0) -- (2, 2);
		\fill [gray!50] (0, 1) -- (1, 0) -- (0, 0);
		\draw (0, 0) grid (2, 2);
		\foreach \x in {1,2} \draw(0,\x)node[left]{\x};
		\foreach \x in {0,1,2} \draw(\x,0)node[below]{\x};
		
		\draw [red, very thick] (1, 1) -- (1, 1.5) -- (2, 2) -- (1.5, 1) -- (1, 1);
		\end{tikzpicture}
	    \caption{The payoffs of plays and SPE outcomes in $G$} \label{fig_inf_spe_payoffs}
	\end{subfigure}
	
	\caption{A game with an infinity of SPEs}
\end{figure}
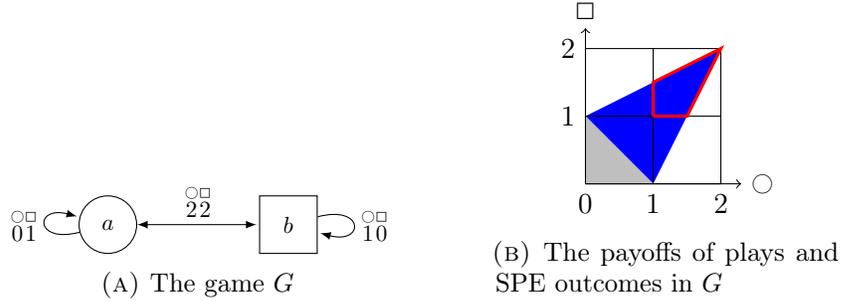

\begin{exa}[A game where SPEs require infinite memory] \label{ex_inf_spe}
    Let us now study the game of Figure~\ref{fig_inf_spe}.
    By Lemma~\ref{lm_dseal}, the payoffs of possible plays in that game correspond to the gray and blue areas in Figure~\ref{fig_inf_spe_payoffs}.
	Indeed, following exclusively one of the three simple cycles $a$, $ab$ and $b$ of the game graph during a play yields the payoffs $01, 10$ and $22$, respectively.
	By combining those cycles with well chosen frequencies, one can obtain any payoff in the convex hull of those three points.
	Now, it is also possible to obtain the point $00$ by using the properties of the limit inferior: it is for instance the payoff of the play $a^2 b^4 a^{16} b^{256} \dots a^{2^{2^n}} b^{2^{2^{n+1}}} \dots$.
	In fact, one can construct a play that yields any payoff in the convex hull of the four points $00, 10, 01$, and $22$.

	We claim that the payoffs of SPEs plays correspond to the red-circled area in Figure~\ref{fig_inf_spe_payoffs}: there exists an SPE $\bsigma$ in $G_{\|a}$ with $\< \bsigma \> = \rho$ if and only if $\mu_{\Box}(\rho), \mu_\playcircle(\rho) \geq 1$.
	That statement will be a direct consequence of the results we show in the remaining sections, but let us give a first intuition: a play with such a payoff necessarily uses infinitely often both states. It is an NE outcome because none of the players can get a better payoff by looping forever on their state, and they can both force each other to follow that play, by threatening them to loop for ever on their state whenever they can. But such a strategy profile is clearly not an SPE.
	
	It can be transformed into an SPE as follows: when a player deviates, say player $\Box$, then player $\Circle$ can punish him by looping on $a$, not forever, but a large number of times, until player $\Box$'s mean-payoff gets very close to $1$. 
	Afterwards, both players follow again the play that was initially planned. 
	Since that threat is temporary, it does not affect player $\Circle$'s payoff on the long term, but it really punishes player $\Box$ if that one tries to deviate infinitely often.
	
	Not that such an SPE requires infinite memory.
\end{exa}

        \subsection{Two-player zero-sum games}

The concept of SPEs has been designed for non-zero-sum games with arbitrarily many players, but the methods we will present in the sequel will bring us back to the more classical framework of two-player zero-sum games, with more complex payoff functions.
We will therefore need the following notions and results.

\begin{defi}[Zero-sum game]
	A game $G$, with $\Pi = \{1, 2\}$, is \emph{zero-sum} if $\mu_2 = -\mu_1$.
\end{defi}
	
\begin{defi}[Borel game]
	A game $G$ is \emph{Borel} if the function $\mu$, from the set $V^\omega$ equipped with the product topology to the Euclidian space $\R^\Pi$, is Borel, i.e. if, for every Borel set $B \subseteq \R^\Pi$, the set $\mu^{-1}(B)$ is Borel.
\end{defi}

\begin{rem}
    Mean-payoff games are Borel.
\end{rem}
	
\begin{lemC}[\cite{BorelDeterminacy}] \label{lm_borel_determinacy}
	Let $G_{\|v_0}$ be a zero-sum Borel game, with $\Pi = \{1, 2\}$. Then, we have the following equality:
	$$\sup_{\sigma_1} ~ \inf_{\sigma_2} ~ \mu_1(\< \bsigma \>) = \inf_{\sigma_2} ~ \sup_{\sigma_1} ~ \mu_1(\< \bsigma \>).$$
\end{lemC}

That quantity is called \emph{value} of $G_{\|v_0}$, denoted by $\val_1(G_{\|v_0})$.
\emph{Solving} a game $G_{\|v_0}$ means computing its value.

\begin{defi}[Optimal strategy]
    Let $G_{\|v_0}$ be a zero-sum Borel game, with $\Pi = \{1, 2\}$.
    The strategy $\sigma_1$ is \emph{optimal} in $G_{\|v_0}$ if $\sup_{\sigma_2} \mu_1(\< \bsigma \>) = \val_1(G_{\|v_0})$.
\end{defi}

Now, let us define memoryless strategies, and state a condition under which they can be optimal.

\begin{defi}[Memoryless strategy]
	A strategy $\sigma_i$ in a game $G_{\|v_0}$ is \emph{memoryless} if for all vertices $v \in V_i$ and for all histories $h$ and $h'$, we have $\sigma_i(hv) = \sigma_i(h'v)$.
\end{defi}

For every game $G_{\|v_0}$ and each player $i$, we write $\ML_i\left(G_{\|v_0}\right)$, or $\ML\left(G_{\|v_0}\right)$ when the context is clear, for the set of memoryless strategies for player $i$ in $G_{\|v_0}$.

\begin{defi}[Shuffling]
    Let $\rho, \eta$ and $\theta$ be three plays in a game $G$.
    The play $\theta$ is a \emph{shuffling} of $\rho$ and $\eta$ if there exist two sequences of indices $k_0 < k_1 < \dots$ and $\l_0 < \l_1 < \dots$ such that $\eta_0 = \rho_{k_0} = \eta_{\l_0} = \rho_{k_1} = \dots$, and:
    $$\theta = \rho_0 \dots \rho_{k_0-1} \eta_0 \dots \eta_{\l_0-1} \rho_{k_0} \dots \rho_{k_1-1} \eta_{\l_0} \dots \eta_{\l_1-1} \dots.$$
\end{defi}

\begin{defi}[Convexity, concavity]
    A payoff function $\mu_i: \Plays G \to \R$ is \emph{convex} if for every shuffling $\theta$ of two plays $\rho$ and $\eta$, we have $\mu_i(\theta) \geq \min\{\mu_i(\rho), \mu_i(\eta)\}$.
    It is \emph{concave} if $-\mu_i$ is convex.
\end{defi}

\begin{rem}
    Mean-payoff functions are convex.
\end{rem}

\begin{lem} \label{lm_memoryless}
    In a two-player zero-sum game played on a finite graph, every player whose payoff function is concave has an optimal strategy that is memoryless.
\end{lem}

\begin{proof}
    According to \cite{DBLP:conf/icalp/Kopczynski06}, this result is true for qualitative objectives, i.e. when $\mu$ can only take the values $0$ and $1$.
    It follows that for every $\alpha \in \R$, if a player $i$, whose payoff function is concave, has a strategy that ensures $\mu_i(\rho) \geq \alpha$ (understood as a qualitative objective), then they have a memoryless one.
    Hence the equality:
    $$\val_1(G_{\|v_0}) = \sup_{\sigma_1 \in \ML(G_{\|v_0})} ~ \inf_{\sigma_2} ~ \mu_1(\< \bsigma \>).$$
    Since the underlying graph $(V, E)$ is assumed to be finite, there exists a finite number of memoryless strategies, hence the infimum above is realized by a memoryless strategy $\sigma_1$ that is, therefore, finite.
\end{proof}

	\section{Requirements and negotiation} \label{sec_negotiation}

We will now see that SPEs are strategy profiles that respect some \emph{requirements} about the payoffs, depending on the states they traverse.
In this part, we develop the notions of \emph{requirement} and \emph{negotiation}.

		\subsection{Requirement}

In the method we will develop further, we will need to analyze the players' behaviours when they have some \emph{requirement} to satisfy.
Intuitively, one can see requirements as \emph{rationality constraints} for the players, that is, a threshold payoff value under which a player will not accept to follow a play.
In all what follows, $\bR$ denotes the set $\R \cup \{\pm \infty\}$.

\begin{defi}[Requirement]
	A \emph{requirement} on the game $G$ is a mapping $\lambda: V \to \bR$.
\end{defi}

For a given state $v$, the quantity $\lambda(v)$ represents the minimal payoff that the player controlling $v$ will require in a play beginning in $v$.

\begin{defi}[$\lambda$-consistency]
	Let $\lambda$ be a requirement on a game $G$. A play $\rho$ in $G$ is \emph{$\lambda$-consistent} if and only if, for all $i \in \Pi$ and $n \in \N$ with $\rho_n \in V_i$, we have $\mu_i(\rho_{\geq n})~\geq~\lambda(\rho_n)$.
	The set of the $\lambda$-consistent plays from a state $v$ is denoted by $\lCons(v)$.
\end{defi}

\begin{defi}[$\lambda$-rationality]
	Let $\lambda$ be a requirement on a mean-payoff game $G$. Let $i \in \Pi$. A strategy profile $\bsigma_{-i}$ is \emph{$\lambda$-rational} if and only if there exists a strategy $\sigma_i$ such that, for every history $hv$ compatible with $\bsigma_{-i}$, the play $\< \bsigma_{\|hv} \>$ is $\lambda$-consistent.
	We then say that the strategy profile $\bsigma_{-i}$ is $\lambda$-rational \emph{assuming} $\sigma_i$.
	The set of $\lambda$-rational strategy profiles in $G_{\|v}$ is denoted by $\lRat(v)$.	
\end{defi}
	
Note that $\lambda$-rationality is a property of a strategy profile for all the players but one, player $i$. Intuitively, their rationality is justified by the fact that they collectively assume that player $i$ will, eventually, play according to the strategy $\sigma_i$: if it is the case, then everyone gets their payoff satisfied.

Finally, let us define a particular requirement: the \emph{vacuous requirement}, that requires nothing, and with which every play is consistent.

\begin{defi}[Vacuous requirement]
	In any game, the \emph{vacuous requirement}, denoted by $\lambda_0$, is the requirement constantly equal to $-\infty$.
\end{defi}

		\subsection{Negotiation} \label{ss_def_nego}

We will show that SPEs in prefix-independent games are characterized by the fixed points of a function on requirements. That function captures a \emph{negotiation} process: when a player has a requirement to satisfy, another player can hope a better payoff than what they can secure in general, and therefore update their own requirement.
Note that we always use the convention $\inf \emptyset = +\infty$.

\begin{defi}[Negotiation function, steady negotiation] \label{def_nego}
	Let $G$ be a game.
	The \emph{negotiation function} is the function that transforms any requirement $\lambda$ on $G$ into a requirement $\nego(\lambda)$ on $G$, such that for each $i \in \Pi$ and $v \in V_i$, we have:
	$$\nego(\lambda)(v) = \inf_{\bsigma_{-i} \in \lRat(v)} \sup_{\sigma_i} \mu_i(\< \bsigma\>).$$
	If that infimum is realized for every $\lambda$, $i$ and $v \in V_i$ such that $\lRat(v) \neq \emptyset$, then the game $G$ is called a game \emph{with steady negotiation}\footnote{The reader having read \cite{Concur} may note that the definition of steady negotiation has been relaxed here.
	The proof of Theorem~\ref{thm_spe} has been modified to fit with this new definition.}.
\end{defi}

\begin{rem}
    The negotiation function satisfies the following properties.
    
    \begin{itemize}
        \item It is monotone: if $\lambda \leq \lambda'$ (for the pointwise order, i.e. if for each $v$, $\lambda(v) \leq \lambda'(v)$), then $\nego(\lambda) \leq \nego(\lambda')$.
        
        \item It is also non-decreasing: for every $\lambda$, we have $\lambda \leq \nego(\lambda)$.
        
        \item There exists a $\lambda$-rational strategy profile from $v$ against the player controlling $v$ if and only if $\nego(\lambda)(v) \neq +\infty$.
    \end{itemize}
\end{rem}

In the general case, the quantity $\nego(\lambda)(v)$ represents the worst case value that the player controlling $v$ can ensure, assuming that the other players play $\lambda$-rationally.

\begin{exa}
    Let us consider the game of Example~\ref{ex_sans_spe}: in Figure~\ref{fig_sans_spe}, on the two first lines below the states, we present the requirements $\lambda_0$ and $\lambda_1 = \nego(\lambda_0)$, which is easy to compute since any strategy profile is $\lambda_0$-rational: for each $v$, $\lambda_1(v)$ is the classical \emph{worst-case value} or \emph{antagonistic value} of $v$, i.e. the best value the player controlling $v$ can enforce against a fully hostile environment. Let us now compute the requirement $\lambda_2 = \nego(\lambda_1)$.
	
	From $c$, there exists exactly one $\lambda_1$-rational strategy profile $\bsigma_{-\playcircle} = \sigma_\Box$, which is the empty strategy since player $\Box$ has never to choose anything. Against that strategy, the best and the only payoff player $\Circle$ can get is $1$, hence $\lambda_2(c) = 1$.
	For the same reasons, $\lambda_2(d) = 2$.
	
	From $b$, player $\Circle$ can force $\Box$ to get the payoff $2$ or less, with the strategy profile $\sigma_\playcircle: h \mapsto c$. Such a strategy is $\lambda_1$-rational, assuming the strategy $\sigma_\Box: h \mapsto d$. Therefore, we have $\lambda_2(b) = 2$.
		
	Finally, from $a$, player $\Box$ can force $\Circle$ to get the payoff $2$ or less, with the strategy profile $\sigma_\Box: h \mapsto d$. Such a strategy is $\lambda_1$-rational, assuming the strategy $\sigma_\playcircle: h \mapsto c$. But, he cannot force her to get less than the payoff $2$, because she can force the access to the state $b$, and the only $\lambda_1$-consistent plays from $b$ are the plays with the form $(ba)^k b d^\omega$. Therefore, $\lambda_2(a) = 2$.
\end{exa}

It will be proved in Section~\ref{sec_concrete_game} that mean-payoff games are with steady negotiation.

		\subsection{Link with Nash equilibria}

Requirements and the negotiation function are able to capture Nash equilibria.
Indeed, if $\lambda_0$ is the vacuous requirement, then $\nego(\lambda_0)$ characterizes the NE outcomes, in the following formal sense:

\begin{thm}\label{thm_ne}
	Let $G$ be a game with steady negotiation.
	Then, a play $\rho$ in $G$ is an NE outcome if and only if $\rho$ is $\nego(\lambda_0)$-consistent.
\end{thm}

\begin{proof} \hfill
\begin{itemize}
	\item Let $\bsigma$ be a Nash equilibrium in $G_{\|v_0}$, for some state $v_0$, and let $\rho = \< \bsigma \>$ : let us prove that the play $\rho$ is $\nego(\lambda_0)$-consistent.
	Let $k \in \N$, let $i \in \Pi$ be such that $\rho_k \in V_i$, and let us prove that $\mu_i\left(\rho_{\geq k}\right) \geq \nego(\lambda_0)(\rho_k)$.
	For any deviation $\sigma'_i$ of $\sigma_{i\|\rho_{\leq k}}$, by definition of NEs, we have $\mu_i\left(\< \bsigma_{-i\|\rho_{\leq k}}, \sigma'_i \>\right) \leq \mu_i(\rho)$.
	Therefore, we have $\mu_i(\rho) \geq \sup_{\sigma'_i} \mu_i\left(\< \bsigma_{-i\|\rho_{\leq k}}, \sigma'_i \>\right)$, hence $\mu_i(\rho) \geq \inf_{\btau_{-i}} ~\sup_{\tau_i} ~ \mu_i\left(\< \btau_{-i\|\rho_{\leq k}}, \tau_i \>\right)$, i.e.	$\mu_i(\rho) \geq \nego(\lambda_0)(\rho_k)$.

	\item Conversely, let $\rho$ be a $\nego(\lambda_0)$-consistent play from a state $v_0$. Let us define a strategy profile $\bsigma$ such that $\< \bsigma \> = \rho$, by:
	
	\begin{itemize}
		\item $\< \bsigma \> = \rho$;
		
		\item for each history of the form $\rho_0 \dots \rho_k v$ with $v \neq \rho_{k+1}$, let $i$ be the player controlling $\rho_k$.
		Since the game $G$ is with steady negotiation, the infimum:
		$$\inf_{\btau_{-i} \in \lambda_0\Rat(\rho_k)}~ \sup_{\tau_i}~ \mu_i(\< \btau \>) = \nego(\lambda_0)(v) \neq +\infty$$
		is a minimum.
		Let $\btau^k_{-i}$ be $\lambda_0$-rational strategy profile from $\rho_k$ realizing that minimum, and let $\tau^k_i$ be some strategy from $\rho_k$ such that $\tau^k_i(\rho_k) = v$. Then, we define:
		$$\< \bsigma_{\|\rho_0 \dots \rho_k v} \> = \< \btau^k_{\rho_k v} \>;$$
		
		\item for every other history $h$, the state $\bsigma(h)$ is defined arbitrarily.
	\end{itemize}

	Let us prove that $\bsigma$ is an NE: let $\sigma'_i$ be a deviation of $\sigma_i$, let $\rho' = \< \bsigma_{-i}, \sigma'_i \>$ and let $\rho_0 \dots \rho_k$ be the longest common prefix of $\rho$ and $\rho'$. Let $v = \rho'_{k+1}$.
	Then, we have:
	$$\mu_i(\rho') \leq \sup_{\tau_i^k} ~\mu_i\left(\< \btau^k \>\right) = \nego(\lambda_0)(\rho_k),$$
	and since $\rho$ is $\lambda_0$-consistent, we have $\nego(\lambda_0)(\rho_k) \leq \mu_i(\rho)$, hence $\mu_i(\rho') \leq \mu_i(\rho)$. \qedhere
\end{itemize}
\end{proof}

\begin{exa}
    Let us consider again the game of Example~\ref{ex_sans_spe}, with the requirement $\lambda_1$ given in Figure~\ref{fig_sans_spe}. The only $\lambda_1$-consistent plays in this game, starting from the state $a$, are $ac^\omega$, and $(ab)^k d^\omega$ with $k \geq 1$. One can check that those plays are exactly the NE outcomes in that game.
\end{exa}

In the following section, we will prove that as well as $\nego(\lambda_0)$ characterizes the NEs, the requirement that is the least fixed point of the negotiation function characterizes the SPEs.

	\section{Link between negotiation and SPEs} \label{sec_link_nego_spe}

The notion of negotiation will enable us to find the SPEs, but also more generally the $\epsilon$-SPEs, in a game. For that purpose, we need the notion of $\epsilon$-fixed points of a function.

\begin{defi}[$\epsilon$-fixed point]
	Let $\epsilon \geq 0$, let $D$ be a finite set and let $f: \bR^D \to \bR^D$ be a mapping. A tuple $\bx \in \R^D$ is a \emph{$\epsilon$-fixed point} of $f$ if for each $d \in D$, for $\by = f(\bx)$, we have $y_d \in [x_d - \epsilon, x_d + \epsilon]$.
\end{defi}

\begin{rem}
    A $0$-fixed point is a fixed point, and conversely.
\end{rem}

By Tarski's fixed point theorem, the negotiation function, which is a monotone function from a complete lattice to itself, has a least fixed point.
That result can be generalized to $\epsilon$-fixed points.

\begin{lem} \label{lm_least_fixed point}
	Let $\epsilon \geq 0$.
	On every game, the function $\nego$ has a least $\epsilon$-fixed point.
\end{lem}

\begin{proof}
    The following proof is a generalization of a classical proof of Tarski's fixed point theorem.
    Let $\Lambda$ be the set of the $\epsilon$-fixed points of the negotiation function. The set $\Lambda$ is not empty, since it contains at least the requirement $v \mapsto +\infty$. Let $\lambda^*$ be the requirement defined by:
    $$\lambda^*: v \mapsto \inf_{\lambda \in \Lambda} \lambda(v).$$
    
    For every $\epsilon$-fixed point $\lambda$ of the negotiation function, we have then $\lambda^*(v) \leq \lambda(v)$ for each $v$, and then $\nego(\lambda^*)(v) \leq \nego(\lambda)(v)$ since $\nego$ is monotone; and therefore, we have $\nego(\lambda^*)(v) \leq \lambda(v) + \epsilon$.
    As a consequence, we have:
    $$\nego(\lambda^*)(v) \leq \inf_{\lambda \in \Lambda} \lambda(v) + \epsilon = \lambda^*(v) + \epsilon.$$
    The requirement $\lambda^*$ is an $\epsilon$-fixed point of the negotiation function, and is therefore the least of them.
\end{proof}

In all what follows, for a given game $G$ and a given $\epsilon > 0$, we will write $\lambda^*$ for the least $\epsilon$-fixed point of the negotiation function.
Intuitively, the requirement $\lambda^*$ is such that, from every vertex $v$, the player $i$ controlling $v$ cannot enforce a payoff greater than $\lambda^*(v) + \epsilon$ against a $\lambda^*$-rational behaviour.
Therefore, the $\lambda^*$-consistent plays are such that if one player tries to deviate, it is possible for the other players to prevent them improving their payoff by more than $\epsilon$, while still playing rationally --- which defines $\epsilon$-SPE outcomes.
Formally:

\begin{thm} \label{thm_spe}
	Let $G_{\|v_0}$ be a prefix-independent game played on a finite graph, and let $\epsilon \geq 0$.
	Let $\theta$ be a play starting in $v_0$.
	If there exists an $\epsilon$-SPE $\bsigma$ such that $\< \bsigma \> = \theta$, then $\theta$ is $\lambda^*$-consistent.
	If $G$ is also a game with steady negotiation, then conversely, if $\theta$ is $\lambda^*$-consistent, then it is an $\epsilon$-SPE outcome.
\end{thm}

\begin{proof} \hfill
    \begin{itemize} 
        \item \emph{If $\bsigma$ is an $\epsilon$-SPE, then the play $\theta = \< \bsigma \>$ is $\lambda^*$-consistent.}

        Let us define a requirement $\lambda$ by, for each $i \in \Pi$ and $v \in V_i$:
        $$\lambda(v) = \inf_{hv \in \Hist G_{\|v_0}} \mu_i(\< \bsigma_{\|hv} \>).$$
        Then, for every history $hv$ starting in $v_0$, the play $\< \bsigma_{\|hv} \>$ is $\lambda$-consistent.
        In particular, the play $\theta$ is.
        Let us now prove that $\lambda$ is an $\epsilon$-fixed point of $\nego$.
        We will then have $\lambda \geq \lambda^*$, which implies that the play $\theta$ is $\lambda^*$-consistent.
        
        Let $i \in \Pi$, let $v \in V_i$, and let us assume toward contradiction (since the negotiation function is non-decreasing) that $\nego(\lambda)(v) > \lambda(v) + \epsilon$, that is to say:
        $$\inf_{\btau_{-i} \in \lRat(v)} ~ \sup_{\tau_i} ~ \mu_i(\< \btau \>) > \lambda(v) + \epsilon = \inf_{hv \in \Hist G_{\|v_0}} \mu_i(\< \bsigma_{\|hv} \>) + \epsilon.$$
        
        Then, since all the plays generated by the strategy profile $\bsigma$ are $\lambda$-consistent, and therefore since any strategy profile of the form $\bsigma_{-i\|hv}$ is $\lambda$-rational, we have:
        $$\inf_{hv} ~ \sup_{\tau_i} ~ \mu_i(\< \bsigma_{-i\|hv}, \tau_i \>) > \inf_{hv} ~ \mu_i(\< \bsigma_{\|hv} \>) + \epsilon.$$
        
        Therefore, there exists a history $hv$ such that:
        $$\sup_{\tau_i} ~ \mu_i(\< \bsigma_{-i\|hv}, \tau_i \>) > \mu_i(\< \bsigma_{\|hv} \>) + \epsilon,$$
        which is impossible if the strategy profile $\bsigma$ is an $\epsilon$-SPE.
        Therefore, there is no such $v$, and the requirement $\lambda$ is an $\epsilon$-fixed point of the negotiation function.

        \item \emph{If $G$ is a game with steady negotiation and $\theta$ is $\lambda^*$-consistent, then $\theta$ is an $\epsilon$-SPE outcome.}
    
\begin{itemize}
	\item \emph{A particular case: if there exists $v$ accessible from $v_0$ such that $\lambda^*(v) = +\infty$.}
	
	In that case, for each $u$ such that $uv \in E$, if the player controlling $u$ chooses to go to $v$, no $\lambda^*$-consistent play can be proposed to them from there, hence there is no $\lambda^*$-rational strategy profile against that player from $u$, and $\nego(\lambda^*)(u) = +\infty$.
	Since $\epsilon$ is finite and since $\lambda^*$ is an $\epsilon$-fixed point of the negotiation function, it follows that $\lambda^*(u) = +\infty$. Since $v$ is accessible from $v_0$, we can repeat this argument and show that $\lambda^*(v_0) = +\infty$; in that case, there is no $\lambda^*$-consistent play $\theta$ from $u$, and then the proof is done.
	
	Therefore, for the rest of the proof, we assume that for all $v$, we have $\lambda^*(v) \neq +\infty$.
	As a consequence, since $\lambda^*$ is an $\epsilon$-fixed point of the function $\nego$, for each $v$ accessible from $v_0$, we have $\nego(\lambda^*)(v) \neq +\infty$; which implies that for each such $v$, there exists a $\lambda^*$-consistent strategy profile against the player controlling $v$, starting from $v$.

    The rest of the proof constructs the strategy profile $\bsigma$ and proves that it is an SPE.
    That construction is illustrated by Figure~\ref{fig_construction_bsigma}.
    
    \begin{figure}
        \centering
        \begin{center}
		\begin{tikzpicture}[scale=0.7]
		\draw (0, 0) node {$\bullet$};
		\draw (0, 0) node[below left] {$v_0$};
		\draw (0, 0) -- (4, 0);
		\draw [dotted] (4, 0) -- (5, 0);
		\draw (4, 0) node[above right] {$\theta$};
		
		\draw[red] (2, 0) node {$\bullet$};
		\draw[red] (2, 0) node[above] {$u \in i$};
		\draw[red, ->] (2, 0) -- (4, -2);
		
		\draw (4, -2) -- (7, -2);
		\draw[dotted] (7, -2) -- (8, -2);
		
		\draw[red] (5, -2) node {$\bullet$};
		\draw[red, ->] (5, -2) -- (6, -1);
		\draw (6, -1) -- (8, -1);
		\draw[dotted] (8, -1) -- (9, -1);
		\draw[red] (6, -2) node {$\bullet$};
		\draw[red, ->] (6, -2) -- (7, -3);
		\draw (7, -3) -- (9, -3);
		\draw[dotted] (9, -3) -- (10, -3);
		
		\draw[red] (14, -2) node {$\left. \begin{matrix} \\ \\ \end{matrix} \right\}$ \small{$\btau^u$}};
		
		\draw[blue] (8, -3) node {$\bullet$};
		\draw[blue] (8, -3) node[below left] {$w \in V_j$};
		\draw[blue, ->] (8, -3) -- (9, -4);
		\draw[dotted] (9, -4) -- (10, -4);
		\draw[blue] (14, -4) node {$\left. \right\}$ \small{$\btau^w$}};
		
		\draw[red] (7, -1) node {$\bullet$};
		\draw[red] (7, -1) node[above left] {$v \in V_i$};
		\draw[red, ->] (7, -1) -- (8, 0);
		\draw[dotted] (8, 0) -- (9, 0);
		
		\draw[red] (14, 0) node {$\left. \right\}$ \small{$\btau^v$ (if reset)}};

		\end{tikzpicture}
	\end{center}
        \caption{The construction of $\bsigma$}
        \label{fig_construction_bsigma}
    \end{figure}
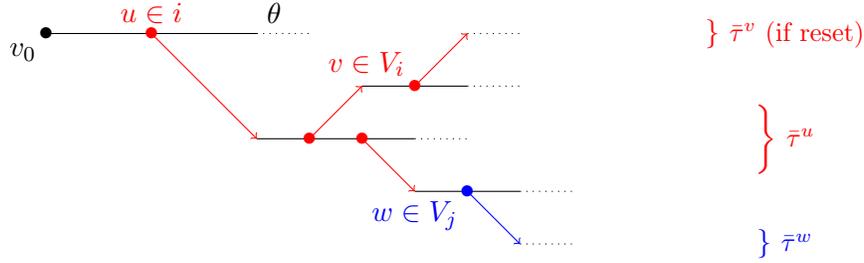

	\item \emph{Spare parts: the strategy profiles $\btau^{v*}$.}
	
	Recall that since $G$ is a game with steady negotiation, for every requirement $\lambda^*$, for every player $i$ and for every state $v \in V_i$, since by the previous point we assume $\lambda^*\Rat(v) \neq \emptyset$, we know that there exists a strategy profile $\btau^v_{-i}$ from $v$ that is $\lambda^*$-rational assuming a strategy $\tau^v_i$ and that satisfies the inequality:
	$$\sup_{\tau_i} ~ \mu_i(\< \btau^v_{-i}, \tau_i \>) = \inf_{\btau_{-i} \in \lRat(v)} ~\sup_{\tau_i} ~\mu_i(\< \btau \>) = \nego(\lambda^*)(v),$$
	i.e. there exists a worst $\lambda^*$-rational strategy profile against player $i$ from the state $v$, with regards to player $i$'s payoff.
	Our goal in this part of the proof is to construct a strategy profile $\btau^{v*}_{-i}$, that is $\lambda^*$-rational assuming a strategy $\tau^{v*}_i$, and that will be used to punish player $i$ when they deviate from $\bsigma$ until another player deviates.
 
	The strategy profile $\btau^v_{-i}$ and the strategy $\tau^v_i$ are not sufficient for that purpose, because if some history $h$ compatible with $\btau^v_{-i}$ is such that $\mu_i(\< \btau^v_{\|h} \>) < \mu_i(\< \btau^v \>)$, then in the corresponding subgame, it may be possible for player $i$ to deviate and get a payoff that would be smaller than or equal to $\mu_i(\< \btau^v \>)$, but greater than $\mu_i(\< \btau^v_{\|h} \>)$.
	On the other hand, the construction of $\btau^{v*}_{-i}$ will ensure that each time player $i$ deviates, the other players punish them at least as harshly as they were planning to do before the deviation.
	
	Let us construct inductively the strategy profile $\btau^{v*}$.
	We define it only on histories that are compatible with $\btau^{v*}_{-i}$, since it can be defined arbitrarily on other histories.
	We proceed by assembling the strategy profiles of the form $\btau^w$ for various $w \in V_i$, and the histories after which we follow a new $\btau^w$ will be called the \emph{resets} of $\btau^{v*}$: they will be histories of the form $hw'$, where $h$ is empty or $\last(h) = w$.
	
	First, we set $\< \btau^{v*} \> = \< \btau^v \>$: the one-state history $v$ is then the first reset of $\btau^{v*}_{-i}$.
	
    Then, for every history $hww'$ from $v$ such that $h$ is compatible with $\btau^{v*}_{-i}$, that $w \in V_i$, and that $w' \neq \tau^{v*}_i(hw)$: let us decompose $hww' = h_1h_2$, so that the history $h_1 \first(h_2)$ is the longest reset of $\btau^{v*}_{-i}$ among the prefixes of $hw$.
    Or, in other words, so that the strategy profile $\btau^{v*}_{\|h_1\first(h_2)}$ has been defined as equal to $\btau^u$ over the prefixes of $h_2$ until $w$, where $u = v$ if $h_1$ is empty, or $u = \last(h_1)$ otherwise.
	By prefix-independence of $G$ and by definition of $\btau^u$ and $\btau^w$, we have:
	$$\inf_{\btau_{-i} \in \lRat(w')} \sup_{\tau_i} \mu_i(\< \btau \>) \leq \sup_{\tau_i} \mu_i(\< \btau^w_{-i}, \tau_i \>) = \nego(\lambda^*)(w).$$
	Let us now separate two cases.
		
		\begin{itemize}
			\item Suppose first that there is equality:
			$$\inf_{\btau_{-i} \in \lRat(w')} \sup_{\tau_i} \mu_i(\< \btau \>) = \nego(\lambda^*)(w).$$
			Then, we choose $\< \btau^{v*}_{\|hww'} \> = \< \btau^u_{\|uh_2} \>$: the coalition of players against player $i$ keeps following the same strategy profile.
			
			\item Suppose now that the inequality is strict:
			$$\inf_{\btau_{-i} \in \lRat(w')} \sup_{\tau_i} \mu_i(\< \btau \>) < \nego(\lambda^*)(w).$$
			Then, we choose $\< \btau^{v^*}_{\|hww'} \> = \< \btau^w_{\|ww'} \>$: player $i$ has done something that lowers the payoff they can ensure, and therefore the other players have to update their strategy profile in order to punish them more.
			The history $hw$ is a reset of $\btau^{v*}_{-i}$.
		\end{itemize}
	
	Since there are finitely many histories of each length, this process completely defines $\btau^{v*}$.
    Moreover, all the plays constructed are $\lambda^*$-consistent, hence the strategy profile $\btau^{v*}_{-i}$ is $\lambda^*$-rational assuming $\tau^{v*}_i$, as desired.

	\item \emph{Construction of $\bsigma$.}

	Let us now construct inductively the strategy profile $\bsigma$: we will prove in the next part of the proof that it is an $\epsilon$-SPE.
    We proceed inductively, by defining all the plays $\< \bsigma_{\|hv} \>$, for $hv \in \Hist(G_{v_0})$ with $v \neq \bsigma(h)$.
    We maintain the induction hypothesis that such a play is always $\lambda^*$-consistent.

	\begin{itemize}
		\item First, we choose $\< \bsigma \> = \theta$, which satisfies the induction hypothesis.

		\item Let now $huv$ be a history such that the strategy profile $\bsigma$ has been defined on all the prefixes of $hu$, which we now assume to be nonempty, but not on $huv$ itself, and such that $v \neq \bsigma(hu)$.
		Let $i$ be the player controlling the state $u$.
		
		Then, we define $\< \bsigma_{\|huv} \> = \< \btau^{u*}_{\|uv} \>$, and inductively, for every history $h'w$ starting from $v$ and compatible with $\bsigma_{-i\|huv}$, we define $\< \bsigma_{\|huh'w} \> = \< \btau^{u*}_{\|uh'w} \>$.
		The strategy profile $\bsigma_{\|huv}$ is then equal to $\btau^{v*}_{\|uv}$ on any history compatible with $\btau^{v*}_{-i}$.
	\end{itemize}

	Since there are finitely many histories of each length, this process completely defines $\bsigma$.

	\item \emph{Such $\bsigma$ is an $\epsilon$-SPE.}
	
	Consider a history $h_0 w \in \Hist G_{\|v_0}$, a player $i \in \Pi$, and a deviation $\sigma'_i$ of $\sigma_i$.
    Let $\rho = h_0 \< \bsigma_{\|h_0w} \>$, and let $\rho' = h_0 \< \bsigma_{-i\|h_0w}, \sigma'_{i\|h_0w} \>$.
	We wish to prove that $\mu_i(\rho') \leq \mu_i(\rho) + \epsilon$.
	
	First, if the play $\rho'$ is compatible with $\sigma_i$, then $\rho' = \rho$ and the proof is immediate.
	Now, if it is not, we let $\rho'_{\leq n}$ denote the shortest prefix of $\rho'$ such that $\rho'_{n-1} \in V_i$ and $\rho'_n \neq \sigma_i(\rho'_{< n})$, and such that $\rho'_{\geq n}$ is compatible with $\bsigma_{-i\|\rho'_{\leq n}}$. 
	Thus, the transition $\rho'_{n-1}\rho'_n$ marks the time when player $i$ begins to deviate unilaterally from $\sigma_i$.
 However, note that $\rho'_{\leq n}$ can be both longer or shorter than $h_0 w$: player $i$ may have already deviated in $h_0 w$, or may wait afterwards to effectively deviate.
	
	Be that as it may, the history $\rho'_{< n}$ is a common prefix of the plays $\rho$ and $\rho'$, and the substrategy profile $\bsigma_{\|\rho'_{\leq n}}$ has been defined during the construction of $\bsigma$ as equal to $\btau^{v*}_{\| \rho'_{n-1}\rho'_n}$, where $v = \rho'_{n-1}$, on any history compatible with $\bsigma_{-i\|\rho'_{\leq n}}$. 
	
	By construction of $\btau^{v*}$, the sequence $\left(\nego(\rho'_k)\right)_{k \geq n-1, \rho'_k \in V_i}$ is non-increasing.
	It is therefore stationary (or finite), because it can take only a finite number of values.
	Consequently, there is a finite number of resets along the play $\rho'_{\geq n-1}$.
	Let $\rho'_{n-1} \dots \rho'_m$ be the last (longest) one.
	Afterwards, the play $\rho'_{\geq m}$ is compatible with the strategy profile $\btau^{\rho'_{m-1}}_{-i}$.
    By definition of that strategy profile, we have the inequality $\mu_i(\rho') \leq \nego(\lambda^*)(\rho'_{m-1})$.
    We need now to prove $\nego(\lambda^*)(\rho'_{m-1}) \leq \mu_i(\rho) + \epsilon$.
	
	Let $\rho_{\leq p} = \rho'_{\leq p}$ denote the longest common prefix of $\rho$ and $\rho'$ such that $\rho_p \in V_i$.
	Since player $i$ does not control any vertex between $\rho_p$ and $\rho_{n-1}$, and therefore cannot deviate, we have $\rho_{\geq p} = \< \bsigma_{\|\rho_{\leq p}} \>$, which is $\lambda^*$-consistent.
	As a consequence, we have $\mu_i(\rho) \geq \lambda^*(\rho_p)$.
 
 Finally, since the sequence of the quantities $\nego(\rho'_k)$ with $\rho'_k \in V_i$ is non-increasing for $k \geq n-1$, we also have $\nego(\lambda^*)(\rho'_{m-1}) \leq \nego(\lambda^*)(\rho_p)$.
 Consequently, we have:
 $$\mu_i(\rho') \leq \nego(\lambda^*)(\rho'_{m-1}) \leq \nego(\lambda^*)(\rho_p) \leq \lambda^*(\rho_p) + \epsilon \leq \mu_i(\rho) + \epsilon.$$
	
	The strategy profile $\bsigma$ is an $\epsilon$-SPE. \qedhere
\end{itemize}
\end{itemize}
\end{proof}

	\section{A first way to handle negotiation: the abstract negotiation game} \label{sec_abstract_game}

        \subsection{Informal definition}

We have now proved that SPEs are characterized by the requirements that are fixed points of the negotiation function; but we need to know how to compute, in practice, the quantity $\nego(\lambda)$ for a given requirement $\lambda$.
We first define an \emph{abstract negotiation game}, that is conceptually simple but not directly usable for an algorithmic purpose, because it is defined on an uncountably infinite state space.

A similar definition was given in \cite{DBLP:journals/mor/FleschP17}, as a tool in a general method to compute SPE outcomes in games whose payoff functions have finite range, which is not the case of mean-payoff games.
Here, linking that game with our concepts of requirements, negotiation function and steady negotiation enables us to present an effective algorithm in the case of mean-payoff games, by constructing a finite version of the abstract negotiation game, the \emph{concrete negotiation game}.

The abstract negotiation game from a state $v_0 \in V_i$, with regards to a requirement $\lambda$, is denoted by $\Abs_{\lambda i}(G)_{\|v_0}$ and opposes two players, \emph{Prover} and \emph{Challenger}, with the following rules:
	
	\begin{itemize}
		\item first, Prover proposes a $\lambda$-consistent play $\rho$ from $v_0$ (or loses, if she has no play to propose).
		
		\item Then, either Challenger accepts the play and the game terminates; or, he chooses an edge $\rho_k \rho_{k+1}$, with $\rho_k \in V_i$, from which he can make player $i$ deviate, using another edge $\rho_k v$ with $v \neq \rho_{k+1}$: then, the game starts again from $v$ instead of $v_0$.
		
		\item In the resulting play (either eventually accepted by Challenger, or constructed by an infinity of deviations), Prover wants player $i$'s payoff to be low, and Challenger wants it to be high.
	\end{itemize}

	That game gives us the basis of a method to compute $\nego(\lambda)$ from $\lambda$: the maximal payoff that Challenger --- or $\C$ for short --- can ensure in $\Abs_{\lambda i}(G)_{\|[v_0]}$, with $v_0 \in V_i$, is also the maximal payoff that player $i$ can ensure in $G_{\|v_0}$, against a $\lambda$-rational environment; hence the equality $\val_\C\left(\Abs_{\lambda i}(G)_{\|[v_0]}\right) = \nego(\lambda)(v_0).$
	A proof of that statement, with a complete formalization of the abstract negotiation game, is presented in Appendix~\ref{app_abstract}.

\begin{exa} \label{ex_abstract_game}
Let us consider again the game of Example~\ref{ex_sans_spe}: the requirement $\lambda_2 = \nego(\lambda_1)$, computed in Section~\ref{ss_def_nego}, is also presented on the third line below the states in Figure~\ref{fig_sans_spe}.
Let us use the abstract negotiation game to compute the requirement $\lambda_3 = \nego(\lambda_2)$.

From $a$, Prover can propose the play $abd^\omega$, and the only deviation Challenger can do is going to $c$; he has of course no incentive to do it. Therefore, $\lambda_3(a) = 2$.
From $b$, whatever Prover proposes at first, Challenger can deviate and go to $a$. Then, from $a$, Prover cannot propose the play $ac^\omega$, which is not $\lambda_2$-consistent: she has to propose a play beginning by $ab$, and to let Challenger deviate once more. He can then deviate infinitely often that way, and generate the play $(ba)^\omega$: therefore, $\lambda_3(b) = 3$.
The other states keep the same values.
Note that there exists no $\lambda_3$-consistent play from $a$ or $b$, hence $\nego(\lambda_3)(a) = \nego(\lambda_3)(b) = +\infty$.
This proves that there is no SPE in that game.
\end{exa}

        \subsection{An imperfect method: the negotiation sequence}

A classical way to compute the least fixed point of a function is, as in the example above, to compute its iterations on the least element of the set we are considering until reaching a fixed point --- which is, then, the least one.
We call this sequence the \emph{negotiation sequence}, and write it $(\lambda_n)_{n \in \N} = (\nego^n(\lambda_0))_n$.
In many simple examples, in practice, computing the negotiation sequence, using the abstract negotiation game, is the way we will find the least fixed point of the negotiation function and solve SPE problems.

\begin{exa} \label{ex_propagation}
    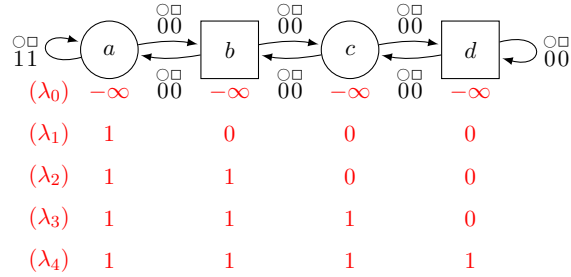
\begin{figure}
    \centering
		\begin{tikzpicture}[->,>=latex,shorten >=1pt, initial text={}, scale=0.8, every node/.style={scale=0.8}]
		\node[state] (a) at (0, 0) {$a$};
		\node[state, rectangle] (b) at (2, 0) {$b$};
		\node[state] (c) at (4, 0) {$c$};
		\node[state, rectangle] (d) at (6, 0) {$d$};
		\path (a) edge [loop left] node {$\stackrel{\playcircle}{1} \stackrel{\Box}{1}$} (a);
		\path[bend left=10] (a) edge node[above] {$\stackrel{\playcircle}{0} \stackrel{\Box}{0}$} (b);
		\path[bend left=10] (b) edge node[below] {$\stackrel{\playcircle}{0} \stackrel{\Box}{0}$} (a);
		\path[bend left=10] (b) edge node[above] {$\stackrel{\playcircle}{0} \stackrel{\Box}{0}$} (c);
		\path[bend left=10] (c) edge node[below] {$\stackrel{\playcircle}{0} \stackrel{\Box}{0}$} (b);
		\path[bend left=10] (c) edge node[above] {$\stackrel{\playcircle}{0} \stackrel{\Box}{0}$} (d);
		\path[bend left=10] (d) edge node[below] {$\stackrel{\playcircle}{0} \stackrel{\Box}{0}$} (c);
		\path (d) edge [loop right] node {$\stackrel{\playcircle}{0} \stackrel{\Box}{0}$} (d);

		\node[red] (l0) at (-1, -0.7) {$(\lambda_0)$};
		\node[red] (l0a) at (0, -0.7) {$-\infty$};
		\node[red] (l0b) at (2, -0.7) {$-\infty$};
		\node[red] (l0c) at (4, -0.7) {$-\infty$};
		\node[red] (l0d) at (6, -0.7) {$-\infty$};
		
		\node[red] (l1) at (-1, -1.4) {$(\lambda_1)$};
		\node[red] (l1a) at (0, -1.4) {$1$};
		\node[red] (l1b) at (2, -1.4) {$0$};
		\node[red] (l1c) at (4, -1.4) {$0$};
		\node[red] (l1d) at (6, -1.4) {$0$};
		
		\node[red] (l2) at (-1, -2.1) {$(\lambda_2)$};
		\node[red] (l2a) at (0, -2.1) {$1$};
		\node[red] (l2b) at (2, -2.1) {$1$};
		\node[red] (l2c) at (4, -2.1) {$0$};
		\node[red] (l2d) at (6, -2.1) {$0$};
		
		\node[red] (l3) at (-1, -2.8) {$(\lambda_3)$};
		\node[red] (l3a) at (0, -2.8) {$1$};
		\node[red] (l3b) at (2, -2.8) {$1$};
		\node[red] (l3c) at (4, -2.8) {$1$};
		\node[red] (l3d) at (6, -2.8) {$0$};
		
		\node[red] (l4) at (-1, -3.5) {$(\lambda_4)$};
		\node[red] (l4a) at (0, -3.5) {$1$};
		\node[red] (l4b) at (2, -3.5) {$1$};
		\node[red] (l4c) at (4, -3.5) {$1$};
		\node[red] (l4d) at (6, -3.5) {$1$};
		\end{tikzpicture}
    \caption{Iterations of the negotiation function}
    \label{fig_ex_nego}
\end{figure}

    Let $G$ be the game of Figure~\ref{fig_ex_nego}, where each edge is labelled by the rewards $r_\playcircle$ and $r_\Box$.
    Below the states, we present the requirements $\lambda_0: v \mapsto -\infty$, $\lambda_1 = \nego(\lambda_0)$, $\lambda_2 = \nego(\lambda_1)$, $\lambda_3 = \nego(\lambda_2)$, and $\lambda_4 = \nego(\lambda_3)$.
    Let us explicate those computations, using the abstract negotiation game.
	From $\lambda_0$ to $\lambda_1$: since every play is $\lambda_0$-consistent, Prover can always propose whatever she wants.
	From the state $a$, whatever she (trying to minimize player $\Circle$'s payoff) proposes, Challenger can always make player $\Circle$ deviate in order to loop on the state $a$.
	Then, in the game $G$, player $\Circle$ gets the payoff $1$, hence $\lambda_1(a) = 1$.
	From the state $b$, Prover (trying to minimize player $\Box$'s payoff) can propose the play $(bc)^\omega$.
	If Challenger makes player $\Box$ deviate to go to the state $a$, then Prover can propose the play $a(bc)^\omega$.
	Even if Challenger makes player $\Box$ deviate infinitely often, he cannot give him more than the payoff $0$, hence $\lambda_1(b) = 0$.
	Similar situations happen from the states $c$ and $d$, hence $\lambda_1(c) = \lambda_1(d) = 0$.
	From $\lambda_1$ to $\lambda_2$: now, from the state $b$, whatever Prover proposes at first, Challenger can make player $\Box$ deviate and go to the state $a$.
	From there, since we have $\lambda_1(a) = 1$, Prover has to propose a play in which player $\Circle$ gets the payoff $1$.
	The only such plays do also give the payoff $1$ to player $\Box$, hence $\lambda_2(b) = 1$.
	Similar situations explain $\lambda_3(c) = 1$, and $\lambda_4(c) = 1$.
	Finally, plays ending with the loop $a^\omega$ are all $\lambda_4$-consistent, hence Prover can always propose them, hence the requirement $\lambda_4$ is a fixed point of the negotiation function --- and therefore the least.
\end{exa}

The interested reader will find other such examples in Appendix~\ref{app_ex}.
However, this cannot be turned into an effective algorithm: the negotiation sequence is not always stationary.

\begin{thm} \label{thm_not_stationary}
    There exists a mean-payoff game on which the negotiation sequence is not stationary.
\end{thm}

\begin{proof}
    Let $G$ be the game of Figure~\ref{fig_not_stationary}.
    Since all the $\Diamond$ rewards are equal to $0$, for all $n > 0$, we have $\lambda_n(c) = \lambda_n(d) = \lambda_n(e) = \lambda_n(f) = 0$.
    Moreover, by symmetry of the game, we always have $\lambda_n(a) = \lambda_n(b)$. Therefore, to compute the negotiation sequence, it suffices to compute $\lambda_{n+1}(a)$ as a function of $\lambda_n(b)$, knowing that $\lambda_1(a) = \lambda_1(b) = 1$, and therefore that for all $n > 0$,  $\lambda_n(a) = \lambda_n(b) \geq 1$.
    
    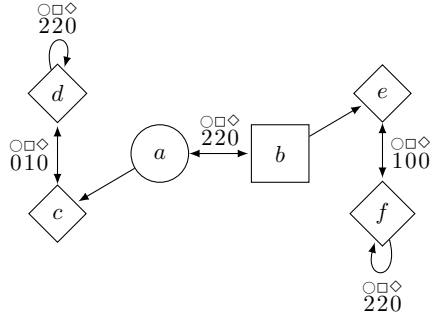
\begin{figure} 
    \begin{center}
    	\begin{tikzpicture}[<->,>=latex,shorten >=1pt, , scale=0.8, every node/.style={scale=0.8}]
    	\node[state, diamond] (c) at (-1.7, -1) {$c$};
    	\node[state, diamond] (d) at (-1.7, 1) {$d$};
    	\node[state] (a) at (0, 0) {$a$};
    	\node[state, rectangle] (b) at (2, 0) {$b$};
    	\node[state, diamond] (e) at (3.7, 1) {$e$};
    	\node[state, diamond] (f) at (3.7, -1) {$f$};
    	
    	\path[<->] (a) edge node[above] {$\stackrel{\playcircle}{2} \stackrel{\Box}{2} \stackrel{\Diamond}{0}$} (b);
    	\path[->] (a) edge (c);
    	\path[<->] (c) edge node[left] {$\stackrel{\playcircle}{0} \stackrel{\Box}{1} \stackrel{\Diamond}{0}$} (d);
    	\path (d) edge [loop above] node {$\stackrel{\playcircle}{2} \stackrel{\Box}{2} \stackrel{\Diamond}{0}$} (d);
    	\path[->] (b) edge (e);
    	\path[<->] (e) edge node[right] {$\stackrel{\playcircle}{1} \stackrel{\Box}{0} \stackrel{\Diamond}{0}$} (f);
    	\path (f) edge [loop below] node {$\stackrel{\playcircle}{2} \stackrel{\Box}{2} \stackrel{\Diamond}{0}$} (f);
    	\end{tikzpicture}
    \end{center}
    \caption{A game where the negotiation sequence is not stationary} \label{fig_not_stationary}
    \end{figure}

    From $a$, the worst play that player $\Box$ could propose to player $\Circle$ would be a combination of the cycles $cd$ and $d$ giving her exactly $1$.
    But then, player $\Circle$ will deviate to go to $b$, from which if player $\Box$ proposes plays in the strongly connected component containing $c$ and $d$, then player $\Circle$ will always deviate and generate the play $(ab)^\omega$, and then get the payoff $2$.
    
    Then, in order to give her a payoff lower than $2$, player $\Box$ has to go to the state $e$.
    Since player $\Circle$ does not control any state in that strongly connected component, the play he will propose will be accepted: he will, then, propose the worst possible combination of the cycles $ef$ and $f$ for player $\Circle$, such that he gets at least his requirement $\lambda_n(b)$.
    The payoff $\lambda_{n+1}(a)$ is then the minimal solution of the system:
    $$\left\{\begin{matrix}
    \lambda_{n+1}(a) = x + 2(1-x) \\
    2(1-x) \geq \lambda_n(b) \\
    0 \leq x \leq 1
    \end{matrix}\right.$$
    that is to say $\lambda_{n+1}(a) = 1 + \frac{\lambda_n(b)}{2} = 1 + \frac{\lambda_n(a)}{2}$, and by induction, for all $n > 0$:
    $$\lambda_n(a) = \lambda_n(b) = 2 - \frac{1}{2^{n-1}}$$
    which converges to $2$ but does never reach it.
\end{proof}

	\section{A tool to compute negotiation: the concrete negotiation game} \label{sec_concrete_game}
	
	    \subsection{Definition}

In the abstract negotiation game, Prover has to propose complete plays, on which we can make the hypothesis that they are $\lambda$-consistent.
In practice, there will often be an infinity of such plays, and therefore it cannot be used directly for an algorithmic purpose.
Instead, those plays can be given edge by edge, in a finite state game.
Its definition is more technical, but it can be shown that it is equivalent to the abstract one.

\begin{defi}[Concrete negotiation game]
	Let $G$ be a prefix-independent game played on a finite graph, let $i \in \Pi$ and $v_0 \in V_i$, and let $\lambda$ be a requirement on $G$.
	The \emph{concrete negotiation game} of $G_{\|v_0}$ is the two-player zero-sum game $\Conc_{\lambda i}(G)_{\|s_0}~=~\left( \{\P, \C\}, S, (S_{\P}, S_{\C}), \Delta, \nu\right)_{\|s_0}$, defined as follows:

	\begin{itemize}
	    \item player $\P$ is called \emph{Prover}, and player $\C$ is called \emph{Challenger}.
	
		\item The set of states controlled by Prover is $S_{\P} = V \times 2^V$, where the state $s = (v, M)$ contains the information of the current state $v$ on which Prover has to define the strategy profile, and the \emph{memory} $M$ of the states that have been traversed so far since the last deviation, and that define the requirements Prover has to satisfy.
		The initial state is $s_0 = (v_0, \{v_0\})$.
		
		\item The set of states controlled by Challenger is $S_{\C} = E \times 2^V$, where in the state $s = (uv, M)$, the edge $uv$ is the edge proposed by Prover.
		
		\item The set $\Delta$ contains three types of transitions: \emph{proposals}, \emph{acceptations} and \emph{deviations}.
		
		\begin{itemize}
			\item The proposals are transitions in which Prover proposes an edge of the game $G$:
			$$\Prop = \left\{ (v, M) (vw, M) ~\left|~
			vw \in E, M \in 2^V \right.\right\};$$
			
			\item the acceptations are transitions in which Challenger accepts to follow the edge proposed by Prover (it is in particular his only possibility when that edge begins on a state that is not controlled by player $i$) --- note that the memory is updated:
			$$\Acc = \left\{ (vw, M)\left(w, M \cup \{w\}\right) ~\left|~
			j \in \Pi, w \in V_j
			\right.\right\};$$
			
			\item the deviations are transitions in which Challenger refuses to follow the edge proposed by Prover, as he can if that edge begins in a state controlled by player $i$ --- the memory is erased, and only the new state the deviating edge leads to is memorized:
			$$\Dev = \left\{ (uv, M) (w, \{w\})	~\left|~ 
			u \in V_i, w \neq v, uw \in E
			\right.\right\}.$$
		\end{itemize}

		\item Let $H = (h_0, M_0) (h_0h'_0, M_0) \dots (h_nh'_n, M_n)$ be a history in $\Conc_{\lambda i}(G)$: the \emph{projection} of the history $H$ is the history $\dH = h_0 \dots h_n$ in the game $G$.
	    That definition is naturally extended to plays.

		\item The payoff function $\nu_\C = -\nu_\P$ measures player $i$'s payoff, with a winning condition if the constructed strategy profile is not $\lambda$-rational, that is to say if after finitely many player $i$'s deviations, it can generate a play which is not $\lambda$-consistent:
		
		\begin{itemize}
		    \item $\nu_\C(\pi) = +\infty$ if after some index $n \in \N$, the play $\pi_{\geq 2n}$ contains no deviation, and if the play $\dpi_{\geq n}$ is not $\lambda$-consistent\footnote{When we combine the notations $\dpi$ and $\pi_{\geq n}$, the notation $\dpi$ is applied first; that is, the play $\dpi_{\geq n}$ is the projection of the play $\pi_{\geq 2n}$, not $\pi_{\geq n}$.};
		    
		    \item $\nu_\C(\pi) = \mu_i(\dpi)$ otherwise.
		\end{itemize}
	\end{itemize}
\end{defi}

Like in the abstract negotiation game, the goal of Challenger is to find a $\lambda$-rational strategy profile that forces the worst possible payoff for player $i$, and the goal of Prover is to find a possibly deviating strategy for player $i$ that gives them the highest possible payoff.

\begin{rem}
    The concrete negotiation game has the following properties.
    \begin{itemize}
        \item If $G$ is Borel, then $\Conc_{\lambda i}(G)$ is Borel.
    
        \item When $\pi_{\geq 2n}$ contains no deviation, the memory of its states is increasing, and therefore eventually equal to the memory $M = \Occ(\dpi_{\geq n})$.
        If it is the longest such suffix of $\pi$, it means that the play $\dpi_{\geq n}$ is $\lambda$-consistent if and only if for each player $j$ and each vertex $v \in V_j$, we have $\mu_j(\dpi) \geq \lambda(v)$.
        
        \item When $G$ is a mean-payoff game and when $\lambda$ has finite values, the concrete negotiation game can be seen as a multidimensional two-player zero-sum mean-payoff game, with one dimension for each player, meant to control that each player gets the payoff they require, plus a special dimension $\star$, meant to measure player $i$'s actual payoff.
        The rewards of the proposals are all equal to $0$, and the rewards of acceptations and deviations are $\hr_\star((uv, M)(v', N)) = 2r_i(uv')$, and $\hr_j((uv, M)(v', N)) = 2r_j(uv') - 2\max\{\lambda(w) ~|~ w \in M \cap V_j\}$.
        The payoff $\nu_\C(\pi)$ equals then $+\infty$ for every $\pi$ that contains finitely many deviations and such that for some $j \in \Pi$, the mean-payoff $\hmu_j(\pi)$ is negative, and $\nu_\C(\pi) = \hmu_\star(\pi)$ otherwise.
    \end{itemize}
\end{rem}

        \subsection{Link with the negotiation function}

The concrete negotiation game is equivalent to the abstract one: the only differences are that the plays proposed by Prover are proposed edge by edge, and that their $\lambda$-consistency is not written in the rules of the game but in its payoff function.

\begin{thm} \label{thm_concrete_game}
	Let $G$ be a Borel prefix-independent game played on a finite graph.
	Let $\lambda$ be a requirement, let $i$ be a player and let $v_0 \in V_i$.
	Then, we have $\val_\C\left(\Conc_{\lambda i}(G)_{\|s_0}\right) = \nego(\lambda)(v_0)$.
	Moreover, if for each player $i$ and every state $v_0 \in V_i$, Prover has an optimal strategy in $\Conc_{\lambda i}(G)_{\|(v_0, \{v_0\})}$, then $G$ is a game with steady negotiation.
\end{thm}

\begin{proof} \hfill
\begin{itemize}
	\item \emph{First direction: $\nego(\lambda)(v_0) \leq \val_\C\left(\Conc_{\lambda i}(G)_{\|s_0}\right)$.}
	
	Let $\tau_\P$ be a strategy such that $\sup_{\tau_\C} \nu_\C(\< \btau \>) \neq +\infty$, and let $\bsigma$ be the strategy profile defined by $\bsigma(\dH) = w$ for every history $H$ compatible with $\tau_\P$ (by induction, the projection is injective on the histories compatible with $\tau_\P$) with $\tau_\P(H) = (vw, \cdot)$, and arbitrarily defined on any other histories.
	We prove that the strategy profile $\bsigma_{-i}$ is $\lambda$-rational assuming the strategy $\sigma_i$, and that $\sup_{\sigma'_i} \mu_i(\< \bsigma_{-i}, \sigma'_i \>) \leq \sup_{\tau_\C} \nu_\C(\< \btau \>)$.
	
	\begin{itemize}
		\item The strategy profile $\bsigma_{-i}$ is $\lambda$-rational, assuming the strategy $\sigma_i$. Indeed, let us assume it is not.
		Then, there exists a history $h = h_0 \dots h_n$ in $G_{\|v_0}$ compatible with $\bsigma_{-i}$ such that the play $\< \bsigma_{\|h} \>$ is not $\lambda$-consistent.
		Then, let:
		$$Hs = \left(h_0, M_0\right) \left(h_0\bsigma(h_0), M_0\right) \dots \left(h_n, M_n\right)$$
		be the only history in $\Conc_{\lambda i}(G)_{\|s_0}$ compatible with $\tau_\P$ such that $\dH = h$.
		Let $\tau_\C$ be a strategy constructing the history $h$, defined by:
		$$\tau_\C\left(H_0 \dots H_{2k-1}\right) = H_{2k}$$
		for every $k$, and:
		$$\tau_\C\left(H' (vw, M)\right) = (w, M \cup \{w\})$$
		for any other history $H' (vw, M)$.
		Then, the play $\pi = \< \btau \>$ contains finitely many deviations (Challenger stops the deviations after having constructed the history $h$), and the play $\dpi_{\geq n}$ is not $\lambda$-consistent. Therefore, we have $\nu_\C(\pi) = +\infty$, which is false by hypothesis.

		\item Now, let us prove the inequality $\sup_{\sigma'_i} \mu_i(\< \bsigma_{-i}, \sigma'_i \>) \leq \sup_{\tau_\C} \nu_\C(\< \btau \>)$.
		Let $\sigma'_i$ be a strategy for player $i$, and let $\eta = \< \bsigma_{-i}, \sigma'_i \>$.
		Let $\tau_\C$ be a strategy such that for every $k$:
		$$\tau_\C\left((\eta_0, \cdot)(\eta_0 \cdot, \cdot) \dots (\eta_k \cdot, \cdot)\right) = (\eta_{k+1}, \cdot),$$
		i.e. a strategy forcing $\eta$ against $\tau_\P$.
		Then, since $\nu_\C(\< \btau \>) \neq +\infty$ by hypothesis on $\tau_\P$, we have $\mu_i(\eta) = \nu_\C(\< \btau \>)$, hence $\mu_i(\< \bsigma_{-i}, \sigma'_i \>) \leq \sup_{\tau_\C} \nu_\C(\< \btau \>)$, hence the desired inequality.
	\end{itemize}
	
	Moreover, if $\tau_\P$ is optimal, then the $\lambda$-rational strategy profile $\bsigma_{-i}$ realizes the infimum:
	$$\inf_{\bsigma_{-i} \in \lRat(v_0)} \sup_{\sigma'_i} \mu_i(\< \bsigma_{-i}, \sigma'_i \>),$$
	hence if there exists such an optimal strategy for every vertex $v_0$, then the game $G$ is with steady negotiation.

	\item \emph{Second direction: $\val_\C\left(\Conc_{\lambda i}(G)_{\|s_0}\right) \leq \nego(\lambda)(v_0)$.}
	
	Let $\bsigma_{-i}$ be a $\lambda$-rational strategy profile from $v_0$, assuming the strategy $\sigma_i$; let us define a strategy $\tau_\P$, by $\tau_\P(H(v, \cdot)) = \left(v\bsigma(\dH v), \cdot\right)$ for every history $H$ and for every $v \in V$.
	Let us prove the inequality $\sup_{\tau_\C} ~\nu_\C(\< \btau \>) \leq \sup_{\sigma'_i} ~\mu_i(\< \bsigma_{-i}, \sigma'_i \>).$
	
	Let $\tau_\C$ be a strategy for Challenger, and let $\pi = \< \btau \>$.
	If $\nu_\C(\pi) = +\infty$, then there exists $n$ such that the play $\pi_{\geq 2n}$ contains no deviation, i.e. $\dpi_{\geq n} = \< \bsigma_{\|\dpi_{\leq n}} \>$, and that play is not $\lambda$-consistent, which is impossible. 
	Therefore, we have $\nu_\C(\pi) \neq + \infty$, and as a consequence $\nu_\C(\pi) = \mu_i(\dpi) = \mu_i(\< \bsigma_{-i}, \sigma'_i \>)$ for some strategy $\sigma'_i$, hence $\nu_\C(\pi) \leq \sup_{\sigma'_i} \mu_i(\< \bsigma_{-i}, \sigma'_i \>)$, hence the desired inequality. \qedhere
\end{itemize}
\end{proof}

    \subsection{Example}

Let us consider again the game from Example~\ref{ex_sans_spe}.
Figure~\ref{fig_concrete} represents the game $\Conc_{\lambda_1 \playcircle}(G)$ (with $\lambda_1(a) = 1$ and $\lambda_1(b) = 2$), where the dashed states are controlled by Challenger, and the other ones by Prover.

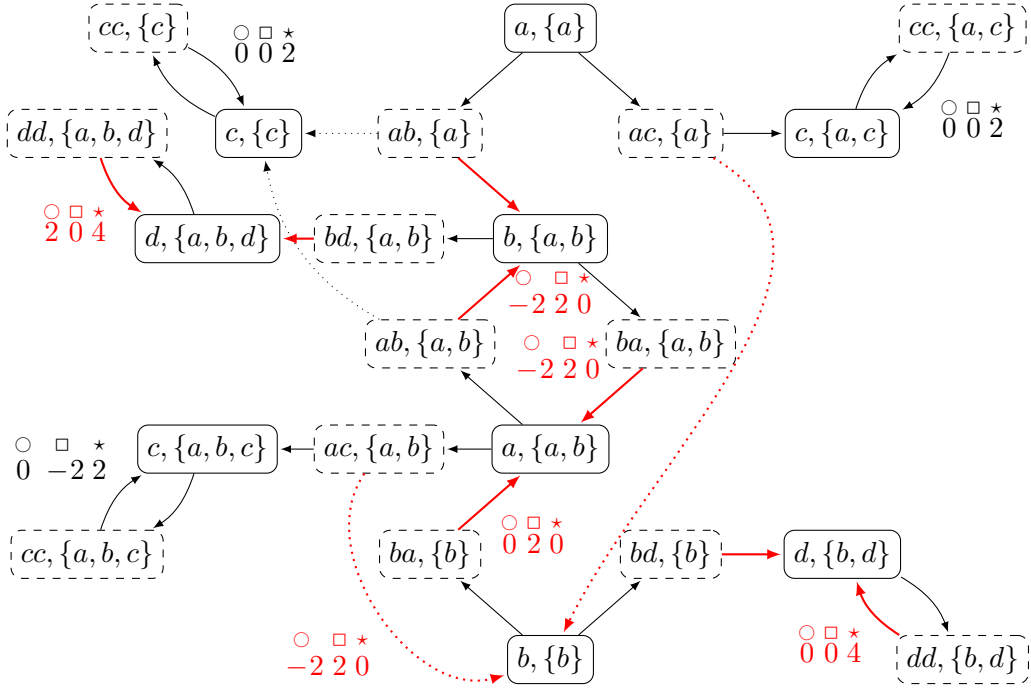
\begin{figure} 
\begin{center}
	\begin{tikzpicture}[->, >=latex,shorten >=1pt, initial text={}, scale=0.4, every node/.style={scale=1}]
	\newcommand{\deltaX}{4}
	\newcommand{\deltaY}{3.5}
	\newcommand{\DeltaX}{5.7}
	\newcommand{\DeltaY}{4.5}
	
	\node[draw, rounded corners] (b-ab) at (0, \deltaY) {$b, \{a, b\}$};
	\node[draw, rounded corners, dashed] (ba-ab) at (\deltaX, 0) {$ba, \{a, b\}$};
	\node[draw, rounded corners] (a-ab) at (0, -\deltaY) {$a, \{a, b\}$};
	\node[draw, rounded corners, dashed] (ab-ab) at (-\deltaX, 0) {$ab, \{a, b\}$};
	
	\node[draw, rounded corners] (a-a) at (0, 3*\deltaY) {$a, \{a\}$};
	\node[draw, rounded corners, dashed] (ab-a) at (-\deltaX, 2*\deltaY) {$ab, \{a\}$};
	\node[draw, rounded corners] (c-c) at (-\deltaX - \DeltaX, 2*\deltaY) {$c, \{c\}$};
	\node[draw, rounded corners, dashed] (cc-c) at (-2*\deltaX - \DeltaX, 3*\deltaY) {$cc, \{c\}$};
	\node[draw, rounded corners, dashed] (ac-a) at (\deltaX, 2*\deltaY) {$ac, \{a\}$};
	\node[draw, rounded corners] (c-ac) at (\deltaX + \DeltaX, 2*\deltaY) {$c, \{a, c\}$};
	\node[draw, rounded corners, dashed] (cc-ac) at (2*\deltaX + \DeltaX, 3*\deltaY) {$cc, \{a, c\}$};
	
	\node[draw, rounded corners] (b-b) at (0, -3*\deltaY) {$b, \{b\}$};
	\node[draw, rounded corners, dashed] (ba-b) at (-\deltaX, -2*\deltaY) {$ba, \{b\}$};
	\node[draw, rounded corners, dashed] (bd-b) at (\deltaX, -2*\deltaY) {$bd, \{b\}$};
	\node[draw, rounded corners] (d-bd) at (\deltaX + \DeltaX, -2*\deltaY) {$d, \{b, d\}$};
	\node[draw, rounded corners, dashed] (dd-bd) at (2*\deltaX + \DeltaX, -3*\deltaY) {$dd, \{b, d\}$};
	
	\node[draw, rounded corners, dashed] (ac-ab) at (-\DeltaX, -\deltaY) {$ac, \{a, b\}$};
	\node[draw, rounded corners] (c-abc) at (-2*\DeltaX, -\deltaY) {$c, \{a, b, c\}$};
	\node[draw, rounded corners, dashed] (cc-abc) at (-2*\DeltaX - \deltaX, -2*\deltaY) {$cc, \{a, b, c\}$};
	
	\node[draw, rounded corners, dashed] (bd-ab) at (-\DeltaX, \deltaY) {$bd, \{a, b\}$};
	\node[draw, rounded corners] (d-abd) at (-2*\DeltaX, \deltaY) {$d, \{a, b, d\}$};
	\node[draw, rounded corners, dashed] (dd-abd) at (-2*\DeltaX - \deltaX, 2*\deltaY) {$dd, \{a, b, d\}$};

	\path (b-ab) edge (ba-ab);
	\path[thick, red] (ba-ab) edge node[above left] {$\scriptsize{\stackrel{\strut\playcircle}{\strut-\!2} ~ \stackrel{\strut\Box}{\strut2} ~ \stackrel{\strut\star}{\strut0}}$} (a-ab);
	\path (a-ab) edge (ab-ab);
	\path (ab-ab) edge[thick, red] node[right] {$\scriptsize{\stackrel{\strut\playcircle}{\strut-\!2}~\stackrel{\strut\Box}{\strut2} ~ \stackrel{\strut\star}{\strut0}}$} (b-ab);
	
	\path (a-a) edge (ab-a);
	\path[dotted] (ab-a) edge (c-c);
	\path[bend left=20] (c-c) edge (cc-c);
	\path[bend left=20] (cc-c) edge node[above right] {$\scriptsize{\stackrel{\playcircle}{0}~\stackrel{\Box}{0} ~ \stackrel{\star}{2}}$} (c-c);
	\path (a-a) edge (ac-a);
	\path (ac-a) edge (c-ac);
	\path[bend left=20] (c-ac) edge (cc-ac);
	\path[bend left=20] (cc-ac) edge node[below right] {$\scriptsize{\stackrel{\playcircle}{0}~\stackrel{\Box}{0} ~ \stackrel{\star}{2}}$} (c-ac);
	
	\path (b-b) edge (ba-b);
	\path (b-b) edge (bd-b);
	\path[thick, red] (bd-b) edge (d-bd);
	\path[bend left=20] (d-bd) edge (dd-bd);
	\path[bend left=20, thick, red] (dd-bd) edge node[below left] {$\scriptsize{
			\stackrel{\playcircle}{0}~\stackrel{\Box}{0} ~ \stackrel{\star}{4}}$} (d-bd);
	
	\path (a-ab) edge (ac-ab);
	\path (ac-ab) edge (c-abc);
	\path[bend left=20] (c-abc) edge (cc-abc);
	\path[bend left=20] (cc-abc) edge node[above left] {$\scriptsize{
			\stackrel{\strut\playcircle}{\strut0}~\stackrel{\strut\Box}{\strut-\!2} ~ \stackrel{\strut\star}{\strut2}
			}$} (c-abc);
	
	\path (b-ab) edge (bd-ab);
	\path[thick, red] (bd-ab) edge (d-abd);
	\path[bend right=20] (d-abd) edge (dd-abd);
	\path[bend right=20, thick, red] (dd-abd) edge node[below left] {$\scriptsize{
			\stackrel{\playcircle}{2}~\stackrel{\Box}{0} ~ \stackrel{\star}{4}}$} (d-abd);
	
	\path[out=-30, in=64, dotted, thick, red] (ac-a) edge (b-b);
	\path[dotted, bend left=25] (ab-ab) edge (c-c);
	\path[thick, red] (ab-a) edge (b-ab);
	\path[thick, red] (ba-b) edge node[below right] {$\scriptsize{
			\stackrel{\playcircle}{0}~\stackrel{\Box}{2} ~ \stackrel{\star}{0}
		}$} (a-ab);
	\path[dotted, bend right=70, thick, red] (ac-ab) edge node[below left] {$\scriptsize{
			\stackrel{\strut\playcircle}{\strut-\!2}~\stackrel{\strut\Box}{\strut2}~\stackrel{\strut\star}{\strut0}
		}$} (b-b);
	\end{tikzpicture}
\end{center}
\caption{A concrete negotiation game} \label{fig_concrete}
\end{figure}

The dotted arrows indicate the deviations, and the transitions have been labelled with the immediate rewards defined as in the remark above.
The transitions that are not labelled are either zero for the three coordinates, or meaningless since they cannot be used more than once.
The red arrows indicate a (memoryless) optimal strategy for Challenger.
Against that strategy, the lowest payoff Prover can ensure is $2$.
Therefore, we have $\nego(\lambda_1)(v_0) = 2$, in line with the abstract game in Example~\ref{ex_abstract_game}.

		\subsection{Resolution}

We now know that $\nego(\lambda)(v)$, for a given requirement $\lambda$, a given player $i$ and a given state $v \in V_i$, is the value of the concrete negotiation game $\Conc_{\lambda i}(G)_{\|(v, \{v\})}$.
But we still do not know how to compute that value.
We present here an important result for that purpose.

For any game $G_{\|v_0}$ and any memoryless strategy $\sigma_i$, we write $G_{\|v_0}[\sigma_i]$ the graph \emph{induced} by $\sigma_i$, defined as the underlying graph of $G$ where all the transitions that are not compatible with $\sigma_i$, and all the vertices that are then no longer accessible from $v_0$, have been omitted.

\begin{lem} \label{lm_concrete_memoryless}
	Let $G_{\|v_0}$ be a mean-payoff game, let $i$ be a player, let $\lambda$ be a requirement and let $\Conc_{\lambda i}(G)_{\|s_0}$ be the corresponding concrete negotiation game.
	There exists a memoryless strategy $\tau_\C$ that is optimal for Challenger.
\end{lem}

\begin{proof}
    The structure of that proof is inspired from the proof of Lemma~14 in \cite{DBLP:journals/iandc/VelnerC0HRR15}.
    
    Let $\nu'_\C$ be the payoff function defined by:
    \begin{itemize}
        \item $\nu'_\C(\pi) = +\infty$ if there exists $n$ such that $\pi_{\geq 2n}$ contains no deviation, and such that the play $\dpi_{\geq n}$ is not $\lambda$-consistent.
        
        \item $\nu'_\C(\pi) = \limsup_n \MP_i(\dpi_{\leq n})$ otherwise.
    \end{itemize}
    
    The payoff function $\nu'_\C$ is then defined as $\nu_\C$, but with a limit superior instead of inferior.
    The payoff function $\nu'_\C$ is concave.
    Indeed, let $\pi$ and $\chi$ be two plays in $\Conc_{\lambda i}(G)_{\|v_0}$, and let $\xi$ be a shuffling of them.
    Let us check that $\nu'_\C(\xi) \leq \max\{\nu'_\C(\pi), \nu'_\C(\chi)\}$.
    
    If either $\nu'_\C(\pi) = +\infty$ or $\nu'_\C(\chi) = +\infty$, it is immediate.
    Otherwise, we also have $\nu'_\C(\xi) \neq +\infty$: if either $\pi$ or $\chi$ contains infinitely many deviations, then so does $\xi$.
    If both contain finitely many deviations, then so does $\xi$: the states of $\xi$ have therefore eventually the same memory $M$, which is also the memory of, eventually, the states of both $\pi$ and $\chi$.
    Now, since $\nu'_\C(\pi), \nu'_\C(\chi) \neq +\infty$, we have $\mu_j(\dpi), \mu_j(\dchi) \geq \lambda(v)$ for each player $j$ and every $v \in M \cap V_j$.
    Since mean-payoff functions are convex, it is also the case for the play $\dxi$, which is a shuffling of $\dpi$ and $\dchi$.
    Hence $\nu_\C(\xi) \neq +\infty$.
    
    Therefore, we have $\nu_\C(\xi) = \limsup_n \MP_i(\dxi_{\leq n})$, as well as $\nu_\C(\pi) = \limsup_n \MP_i(\dpi_{\leq n})$ and $\nu_\C(\dchi) = \limsup_n \MP_i(\dchi_{\leq n})$.
    Since, as shown in \cite{DBLP:journals/iandc/VelnerC0HRR15}, mean-payoff functions defined with a limit superior are concave, it implies $\nu_\C(\xi) \leq \max\{ \nu_\C(\pi), \nu_\C(\chi)\}$: the payoff function $\nu_\C$ is concave.
    
    Therefore, by Lemma~\ref{lm_memoryless} Challenger has a memoryless strategy that is optimal with regards to the payoff function $\nu'_\C$: let us write it $\tau_\C$.
    Now, we want to prove that the memoryless strategy $\tau_\C$ is also optimal with regards to $\nu_\C$.
    Note that for every play $\pi$, we have $\nu_\C(\pi) \leq \nu'_\C(\pi)$, and therefore $\val_\C\left( \Conc_{\lambda i}(G)_{\|s_0} \right) \leq \alpha$, where $\alpha$ is the value of the game $\Conc_{\lambda i}(G)_{\|s_0}$ with the payoff function $\nu'_\C$ instead of $\nu_\C$.
    Therefore, we have proven that $\tau_\C$ is optimal with regards to $\nu_\C$ if we prove that $\inf_{\tau_\P} \nu_\C(\< \btau \>) \geq \alpha$.
    
    Let $\pi$ be a play compatible with $\tau_\C$, i.e. an infinite path from $s_0$ in the graph $\Conc_{\lambda i}(G)_{\|s_0}[\tau_\C]$.
    If $\nu_\C(\pi) = +\infty$, then clearly $\nu_\C(\pi) \geq \alpha$.
    Otherwise, we have $\nu_\C(\pi) = \mu_i(\dpi)$, and by Lemma~\ref{lm_dseal}, we have:
    $$\mu_i(\dpi) \geq \min_{c \in C} \MP_i(\dc),$$
    where $C$ is the set of the simple cycles of the graph $\Conc_{\lambda i}(G)_{\|s_0}[\tau_\C]$.
    Now, for each such cycle, there exists a history $H$ such that the play $H c^\omega$ is compatible with the strategy $\tau_\C$, and therefore satisfies $\nu'_\C(Hc^\omega) \geq \alpha$, and consequently $\MP_i(\dc) \geq \alpha$.
    Therefore, we have $\mu_i(\dpi) \geq \alpha$, and the strategy $\tau_\C$ is optimal with regards to the payoff function $\nu_\C$.
\end{proof}

Using this lemma, computing $\nego(\lambda)$ for any given $\lambda$ amounts to looking for an optimal path for Prover in each graph $\Conc_{\lambda i}(G)_{\|s_0}[\tau_\C]$.
When $(V, E)$ is a graph, we write $\SConn(V, E)$ the set of its strongly connected components accessible from the vertex $v$.

Let $K$ be a strongly connected subgraph of a concrete negotiation game.
If $K$ contains no deviation, then the states of $K$ share all the same memory: let us us write it $\Mem(K)$.
If $K$ contains at least one deviation, we define $\Mem(K) = \emptyset$.
Then, we write:
$$\opt(K) = \inf \left\{ x_i ~\left|~ \begin{matrix}
        \bx \in \dseal \left( \underset{c \in \SC(K)}{\Conv} \MP_i(\dc) \right), \\
        \forall i \in \Pi, \forall v \in \Mem(K), x_i \geq \lambda(v)
    \end{matrix} \right. \right\}.$$

The set in which the variable $\bx$ evolves is the set of payoffs of plays that Prover can construct against Challenger, when she chooses to go in the strongly connected component $K$, and observes the requirements she has to observe --- those stored in the common memory of the states of $K$ if $K$ contains no deviations, and none otherwise.
Hence the following formal result.

\begin{lem} \label{lm_compute_nego}
    Let $G$ be a mean-payoff game, let $\lambda$ be a requirement, let $i$ be a player, and let $v \in V_i$.
    Then, we have:
    $$\nego(\lambda)(v) = \sup_{\tau_\C \in \ML(\Conc_{\lambda i}(G))}~ \inf_{K \in \SConn\left(\Conc_{\lambda i}(G)_{\|(v, \{v\})}[\tau_\C]\right)}~ \opt(K).$$
    Moreover, mean-payoff games are games with steady negotiation.
\end{lem}

\begin{proof}
    By Lemma \ref{lm_concrete_memoryless}, in the game $\Conc_{\lambda i}(G)_{\|s}$, where $s = (v, \{v\})$, there exists a memoryless strategy $\tau_\C$ which is optimal for Challenger.
    Therefore, the best payoff that Prover can ensure against every strategy of Challenger is the best payoff she can ensure against $\tau_\C$.
    It follows from Theorem \ref{thm_concrete_game} that the highest value player $i$ can enforce against a hostile $\lambda$-rational environment is the minimal payoff of Challenger in a path in the graph $\Conc_{\lambda i}(G)_{\|s}[\tau_\C]$ starting from $s$.
    For any such path $\pi$, there exists a strongly connected component $K$ of $\Conc_{\lambda i}(G)_{\|s}[\tau_\C]$ such that after a finite number of steps, the path $\pi$ is a path in $K$.
    Let us now prove that the least payoff of Challenger in such a play is given by $\opt(K)$.
    Let us distinguish two cases.
    
    \begin{itemize}
    	\item \emph{If there is at least one deviation in $K$.}
    	
    	Then, for every play $\pi$ in $K$, it is possible to transform $\pi$ into a play $\pi'$ with $\mu(\dpi') = \mu(\dpi)$, which contains infinitely many deviations: it suffices to add round trips to a deviation, endlessly, but less and less often.
    	Therefore, the outcomes $\nu_\C(\pi)$ of plays in $K$ are exactly the mean-payoffs $\mu_i(\dpi)$ of plays in $K$ (plus possibly $+\infty$); and in particular, the lowest payoff Challenger can get in $K$ is the quantity:
    	$$\min_{c \in \SC(K)} \MP_i(\dc),$$
    	which is equal to $\opt(K)$ since $\Mem(K) = \emptyset$.

    	\item \emph{If there is no deviation in $K$.}
    	
    	By Lemma~\ref{lm_dseal}, the set of possible values of $\mu(\dpi)$ for all plays $\pi$ in $K$ is exactly the set:
    	$$X = \dseal \left( \underset{c \in \SC(K)}{\Conv} \hmu(c^\omega) \right).$$
    	
    	Since all the plays in $K$ contain finitely many deviations (actually none), for every path $\pi$ in $K$, we have $\nu_\C(\pi) = +\infty$ if and only if there exists $j \in \Pi$ and $u \in V_j \cap \Mem(K)$ such that $\mu_j(\dpi) < \lambda(u)$.
    	Then, the lowest outcome Prover can get in $K$ is:
    	$$\inf \left\{ x_i ~|~ \bx \in X, \forall j \in \Pi, \forall u \in V_j \cap \Mem(K), x_j \geq \lambda(u) \right\},$$
    	that is to say $\opt(K)$.
    \end{itemize}
    
    Theorem \ref{thm_concrete_game} enables to conclude to the desired formula.
    Moreover, let us notice that in all cases, Prover can choose one optimal play against each memoryless strategy of Challenger.
    By determinacy of Borel games, it comes that Prover has an optimal strategy, hence by Theorem~\ref{thm_concrete_game}, mean-payoff games are games with steady negotiation.
\end{proof}

We are now able to compute $\nego(\lambda)$ for a given $\lambda$.
However, because of Theorem~\ref{thm_not_stationary}, that is not sufficient to compute the least $\epsilon$-fixed point $\lambda^*$, and then to decide the SPE threshold problem.
Nevertheless, we will prove that $\lambda^*$ can also be extracted from the concrete negotiation game.

    \section{Analysis of the negotiation function} \label{sec_analysis}

\begin{defi}[Piecewise affine function]
    Let $D$ be a finite set of dimensions.
    A function $f: \bR^D \to \bR^D$ is \emph{piecewise affine} if for each $d \in D$, there exists a finite partition $\Phi_d$ of $\bR^D$, where every $P \in \Phi_d$ is a polyhedron, such that for each such $P$ there exists $\ba_P \in \R^D \setminus \{\bzero\}$ and $b_P \in \bR$ such that for every $\bx \in P$, the vector $\by = f(\bx)$ satisfies:
    $$y_d = \ba_P \cdot \bx + b_P,$$
    where $\cdot$ is the canonical scalar product.
\end{defi}

\begin{rem}
    The function $f$ is fully represented by the family $(\Phi_d, (\ba_P, b_P)_{P\in \Phi_d})_{d \in D}$.
    That representation is finite if each polyhedron $P \in \Phi_d$, for each dimension $d$, is defined by rational equations, and if each $\ba_P$ and each $b_P$ has rational or infinite values.
\end{rem}

\begin{thm} \label{thm_piecewise_affine}
    Let us assimilate every requirement $\lambda$ to the vector $\blambda = (\lambda(v))_{v \in V}$.
    Then, the negotiation function is piecewise affine, and a finite representation of it can be constructed in a time doubly exponential in the size of $G$.
\end{thm}

\begin{proof}
    Let $i \in \Pi$, let $v \in V_i$, and let $\lambda$ be a requirement.
    Let $\tau_\C$ be a memoryless strategy of Challenger in the game $\Conc_{\lambda i}(G)$, and let $K$ be a strongly connected component of the graph $\Conc_{\lambda i}(G)_{\|s}$, where $s = (v, \{v\})$.
    The result will follow from the fact that the quantity $\opt(K)$ is, itself, a piecewise affine function of $\blambda$.
    Note that the underlying graph of the game $\Conc_{\lambda i}(G)$ does not depend on $\lambda$.
    
    For a given $\lambda$, let us consider the polytope $Q = \Conv_{c \in \SC(K)} \MP(\dc) \subseteq \R^\Pi$, and the polyhedron $R_\lambda = \{\bx \in \R^\Pi ~|~ \forall i \in \Pi, \forall u \in \Mem(K), x_i \geq \lambda(u)\}$.
    Then, we have:
    \begin{align*}
        \opt(K) &= \inf \{ x_i ~|~ \bx \in \dseal Q \cap R_\lambda \} \\
        &= \inf \left\{ \min_{\bz \in Z} z_i ~\left|~ \begin{matrix}
            Z \text{ is a finite subset of } Q, \\
            \text{and } (\min_{\bz \in Z} z_j)_j \in R_\lambda
        \end{matrix} \right. \right\} \\
        &= \inf \left\{ \left. \min_{\bz \in Z} z_i ~\right|~ Z \text{ is a finite subset of } Q \cap R_\lambda \right\} \\
        &= \inf \left\{ x_i ~\left|~ \bx \in Q \cap R_\lambda \right. \right\}.
    \end{align*}
    From now, we can therefore drop the downward sealing.
    Moreover, let us note that the projection $\bx \mapsto x_i$, as an affine mapping over the polytope $Q \cap R_\lambda$, finds its minimum on a vertex of that polytope.
    Each vertex $\{\bx\}$ of $Q \cap R_\lambda$ is the intersection between a face $F$ of the polytope $Q$, and a face $F'$ of the polyhedron $R_\lambda$.
    Such a face $F$ is of the form $\Conv_{c \in C} \MP(\dc)$, where $C$ is a subset of $\SC(K)$; and such a face $F'$ is of the form $\bigcap H_{\lambda w}$, where $W$ is a subset of $\Mem(K)$, and where $H_{\lambda w}$ is the hyperplane $\{ \bx ~|~ x_j = \lambda(w)\}$ for each $j$ and $w \in V_j$.
    Thus, if we define the set:
    $$X = \left\{ \bx \in R_\lambda ~\left|~ \begin{matrix}
        \exists C \subseteq \SC(K), \exists W \subseteq \Mem(K), \\
        \Conv_{c \in C} \MP(\dc) \cap \bigcap_{w \in W} H_{\lambda w} = \{ \bx \}
    \end{matrix} \right. \right\},$$
    it is included in the polytope $Q \cap R_\lambda$, and contains all its vertices; hence the equality $\opt(K) = \inf\{ x_i ~|~ \bx \in X\}$.
    We can therefore write:
    $$\opt(K) = \inf_{C \subseteq \SC(K), W \subseteq \Mem(K)} f_{CW}(\blambda),$$
    where $f_{CW}$ is, for every $C$ and $W$, the function defined by:
    \begin{itemize}
        \item $f_{CW}(\blambda) = x_i$ if the intersection $I = \left(\Aff_{c \in C} \MP(\dc)\right) \cap \left(\bigcap_{w \in W} H_{\lambda w}\right)$, where $\Aff Z$ denotes the smallest affine space containing $Z$, is a singleton $\{\bx\}$ with $\bx \in \Conv_{c \in C} \MP(\dc) \cap R_\lambda$,
        
        \item and $f_{CW}(\lambda) = +\infty$ otherwise.
    \end{itemize}
    
    Let us now study each of those functions $f_{CW}$.
    Given $C, W$ and $\lambda$, we have $f_{CW}(\lambda) \neq +\infty$ if and only if the three following conditions are satisfied.
    
    \begin{itemize}
        \item First, the intersection $I$ is a singleton.
        The elements of $I$ are the points of the form $\bx = \left( \sum_{c \in C} \alpha_c \MP_j(\dc) \right)_j$ with $\balpha \in \R^C$, $\sum_c \alpha_c = 1$, and $x_j = \lambda(w)$ for each $j$ and $w \in W \cap V_j$.
        The set $I$ is therefore a singleton if and only if the matrix:
        $$A = \left(\left\{\begin{matrix}
            1 & \text{if } w = \Sum \\
            \MP_j(\dc) & \text{else, with } j \text{ s.t. } w \in V_j
        \end{matrix}\right. \right)_{w \in W \cup \{\Sum\}, c \in C} \in \R^{(W \cup \{\Sum\}) \times C}$$
        is such that there exists exactly one vector $\balpha \in \R^C$ satisfying:
        $$A \balpha = \left(\left\{\begin{matrix}
            1 & \text{if } j = \Sum \\
            \lambda(w) & \text{otherwise}
        \end{matrix}\right.\right)_{w \in W \cup \{\Sum\}}.$$
        
        That condition is satisfied if and only if $A$ is invertible, which can be decided in a time polynomial in the size of $A$, and does actually not depend on $\lambda$: either it is not satisfied, and the function $f_{CW}$ is constantly equal to $+\infty$, or it is, and only the following conditions must be considered.

        \item Second, the unique element of $I$ belongs to $\Conv_{c \in C} \MP(\dc)$.
        That is the case if and only if the vector:
        $$\balpha = A^{-1} \left(\left\{\begin{matrix}
            1 & \text{if } j = \Sum \\
            \lambda(w) & \text{otherwise}
        \end{matrix}\right.\right)_{w \in W \cup \{\Sum\}}$$
        has only non-negative coordinates.
        Here we should decompose:
        $$\left(\left\{\begin{matrix}
            1 & \text{if } j = \Sum \\
            \lambda(w) & \text{otherwise}
        \end{matrix}\right.\right)_{w \in W \cup \{\Sum\}} = B \blambda + \bbeta$$
        where $B = \left(\left\{\begin{matrix}
            1 & \mathrm{if~} w = v \\
            0 & \mathrm{otherwise}
        \end{matrix}\right.\right)_{w \in W \cup \{\Sum\}, v \in V}$
        and $\bbeta = \left(\left\{\begin{matrix}
            1 & \mathrm{if~} w = \Sum \\
            0 & \mathrm{otherwise}
        \end{matrix}\right.\right)_{w \in W \cup \{\Sum\}}$.
        
        Thus, the vector $\balpha = A^{-1} (B \blambda + \bbeta)$ has non-negative coordinates if and only if $\blambda$ belongs to the set:
        $$P_0 = \left(\blambda' \mapsto A^{-1} (B \blambda' + \bbeta)\right)^{-1}\left(\left[0, +\infty\right)^C\right)$$
        which, as a pre-image of a polyhedron by an affine function, is itself a polyhedron, which can be constructed in a time polynomial in the size of $A$, $B$ and $\beta$.

        \item Third, the vector:
        $$\bx = \left( \MP_j(\dc) \right)_{j \in \Pi, c \in C} \balpha$$
        is such that for each $j$ and each $w \in \Mem(K)$ (not only in $W$), we have $x_j \geq \lambda(w)$.
        The set $P_1$ of requirements $\blambda$ satisfying that condition can itself be written as the pre-image of a polyhedron by an affine function, and is therefore itself a polyhedron, which we can construct in a time polynomial in the size of $A$, $B$, $\beta$ and $\left( \MP_j(\dc) \right)_{j \in \Pi, c \in C}$.
    \end{itemize}
    
    Therefore, the function $f_{CW}$ is equal to $+\infty$ outside of the polyhedron $P_0 \cap P_1$, and satisfies:
    $$f_{CW}(\blambda) = \left( \MP_i(\dc) \right)_{c \in C} \cdot A^{-1} (B \blambda + \bbeta)$$
    inside it.
    It is therefore an affine function of which a representation can be constructed in a time polynomial in the size of $K$.
    Therefore, a representation of $\opt(K) = \inf_{C, W} f_{CW}(\blambda)$ as an affine function of $\blambda$ can be constructed in a time exponential in the size of $K$, and the negotiation function, expressed by:
    $$\nego: \blambda \mapsto \sup_{\tau_\C} \inf_K \inf_{C, W} f_{CW}(\blambda)$$
    can be constructed in a time doubly exponential in the size of $G$.
\end{proof}

\begin{exa}
    Let us consider the game of Example~\ref{ex_inf_spe}.
    If a requirement $\lambda$ is represented by the tuple $(\lambda(a), \lambda(b))$, the function $\nego: \R^2 \to \R^2$ can be depicted by Figure~\ref{fig_linear}, where in any one of the regions delimited by the dashed lines, we wrote a formula for the couple $(\nego(\lambda)(a), \nego(\lambda)(b))$.
    The orange area indicates the fixed points of the function, and the yellow area the other $\frac{1}{2}$-fixed points.
\end{exa}

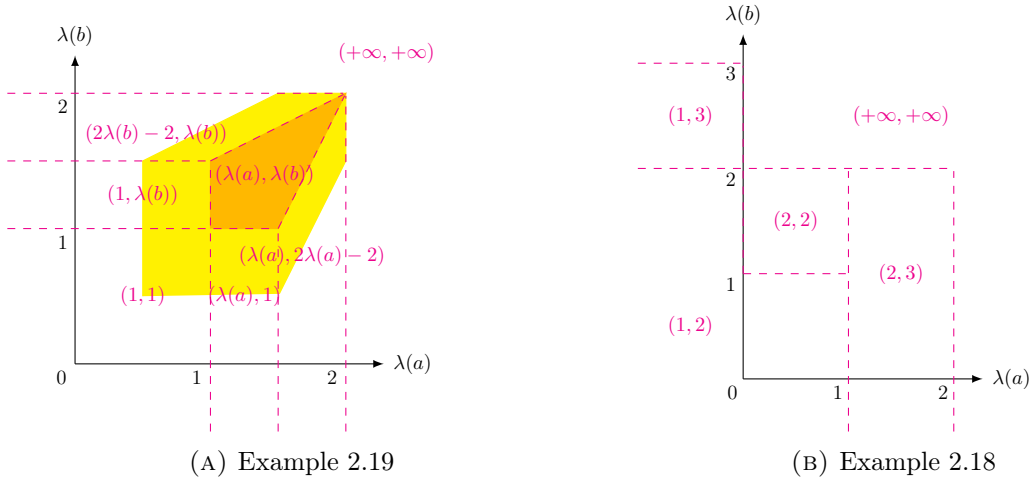
\begin{figure}
	\begin{center}
	\begin{subfigure}[b]{0.5\textwidth}
        	\begin{tikzpicture}[scale=1.8,>=latex,shorten >=1pt, every node/.style={scale=0.7}]
        	\filldraw[yellow] (1/2, 1/2) -- (1/2, 3/2) -- (3/2, 2) -- (2, 2) -- (2, 3/2) -- (3/2, 1/2);
        	\filldraw[orange, opacity=0.5] (1, 1) -- (1, 1.5) -- (2, 2) -- (1.5, 1) -- (1, 1);
        	
        	\draw[->] (0,0) -- (2.3,0);
        	\draw (2.3,0) node[right] {$\lambda(a)$};
        	\draw[->] (0,0) -- (0,2.3);
        	\draw (0,2.3) node[above] {$\lambda(b)$};
        	\foreach \x in {1,2} \draw(0,\x-0.1)node[left]{\x};
        	\foreach \x in {0,1,2} \draw(\x-0.1,0)node[below]{\x};
        	
        	\path[magenta, dashed] (1, -0.5) edge (1, 1.5);
        	\path[magenta, dashed] (2, -0.5) edge (2, 2);
        	\path[magenta, dashed] (-0.5, 1) edge (1.5, 1);
        	\path[magenta, dashed] (-0.5, 2) edge (2, 2);
        	\path[magenta, dashed] (1, 1.5) edge (2, 2);
        	\path[magenta, dashed] (1.5, 1) edge (2, 2);
        	\path[magenta, dashed] (-0.5, 1.5) edge (1, 1.5);
        	\path[magenta, dashed] (1.5, -0.5) edge (1.5, 1);
        	
        	\draw[magenta] (0.5, 0.5) node {$(1, 1)$};
        	\draw[magenta] (0.5, 1.25) node {$(1, \lambda(b))$};
        	\draw[magenta] (1.25, 0.5) node {$(\lambda(a), 1)$};
        	\draw[magenta] (0.6, 1.7) node {$(2\lambda(b)-2, \lambda(b))$};
        	\draw[magenta] (1.75, 0.8) node {$(\lambda(a), 2\lambda(a)-2)$};
        	\draw[magenta] (1.4, 1.4) node {$(\lambda(a), \lambda(b))$};
        	\draw[magenta] (2.3, 2.3) node {$(+\infty,  +\infty)$};
        	\end{tikzpicture}
        \caption{Example~\ref{ex_inf_spe}} \label{fig_linear}
        \end{subfigure}
        \hfill
        \begin{subfigure}[b]{0.45\textwidth}
		\begin{tikzpicture}[scale=1.4,>=latex,shorten >=1pt, every node/.style={scale=0.7}]
		
		\draw[->] (0,0) -- (2.3,0);
		\draw (2.3,0) node[right] {$\lambda(a)$};
		\draw[->] (0,0) -- (0,3.3);
		\draw (0,3.3) node[above] {$\lambda(b)$};
		\foreach \x in {1,2,3} \draw(0,\x-0.1)node[left]{\x};
		\foreach \x in {0,1,2} \draw(\x-0.1,0)node[below]{\x};
		
		\path[magenta, dashed] (0, 1) edge (0, 3);
		\path[magenta, dashed] (1, -0.5) edge (1, 2);
		\path[magenta, dashed] (2, -0.5) edge (2, 2);
		\path[magenta, dashed] (0, 1) edge (1, 1);
		\path[magenta, dashed] (-1, 2) edge (2, 2);
		\path[magenta, dashed] (-1, 3) edge (0, 3);
		
		\draw[magenta] (-0.5, 0.5) node {$(1, 2)$};
		\draw[magenta] (0.5, 1.5) node {$(2, 2)$};
		\draw[magenta] (1.5, 1) node {$(2, 3)$};
		\draw[magenta] (-0.5, 2.5) node {$(1, 3)$};
		\draw[magenta] (1.5, 2.5) node {$(+\infty,  +\infty)$};
		\end{tikzpicture}
		\caption{Example~\ref{ex_sans_spe}} \label{fig_linear2}
	    \end{subfigure}
	\end{center}
	\caption{The negotiation function on the games of Examples~\ref{ex_inf_spe} and~\ref{ex_sans_spe}}
    \label{fig_linear_total}
\end{figure}

\begin{exa}
	Now, let us consider the game of Example~\ref{ex_sans_spe}.
	Let us fix $\lambda(c) = 1$ and $\lambda(d) = 2$, and represent the requirements $\lambda$ by the tuples $(\lambda(a), \lambda(b))$, as in the previous example. Then, the negotiation function is depicted by Figure~\ref{fig_linear2}.
	One can check that there is no fixed point here, and even no $\frac{1}{2}$-fixed point --- except $(+\infty, +\infty)$.
\end{exa}

Consequently, the least $\epsilon$-fixed point of the negotiation function can itself be computed using the classical linear algebra tool box.

    \section{Conclusion: algorithm and complexity} \label{sec_conclusion}

Thanks to all the previous results, we are now able to state the decidability of the $\epsilon$-SPE threshold problem, and to bound its complexity.
Let us start with a lower bound.

        \subsection{Lower bound}

\begin{thm} \label{thm_np_hard}
    The $\epsilon$-SPE threshold problem is $\NP$-hard, even when $\epsilon$ is fixed equal to $0$.
\end{thm}

\begin{proof}
    We proceed by reduction from the $\NP$-complete problem SAT. This proof is liberally inspired from the proof of the $\NP$-hardness of the NE threshold problem in co-Büchi games by Michael Ummels, in \cite{DBLP:conf/fossacs/Ummels08}.
	
	Let $\phi = \bigwedge_{i = 1}^n \bigvee_{j = 1}^m L_{ij}$ be a formula from propositional logic, written in conjunctive normal form, over the finite variable set $X$. We construct a mean-payoff game $G^\phi_{\|v_0}$ that admits an SPE where the player $\S$ gets the payoff $1$, if and only if $\phi$ is satisfiable.
	
	First, we define the set of players $\Pi = \{\S\} \cup X$: every variable of $\phi$ is a player and there is an additional special player $\S$, called \emph{Solver}, who wants to prove that $\phi$ is satisfiable.
	
	Then, let us define the state space: for each clause $C_i$, with $i \in \Z/n\Z$, of $\phi$, we define a state $C_i$ that is controlled by Solver, and for each litteral $L_{ij}$ of $C_i$ we define a state $(C_i, L_{ij})$, that is controlled by the player $x$ such that $L_{ij} = x$ or $\neg x$.
	We add a transition from $C_i$ to $(C_i, L_{ij})$, and another one from $(C_i, L_{ij})$ to $C_{i+1}$.
	Moreover, we add a sink state $\bot$, with a transition from it to itself, and transitions from all the states of the form $(C, \neg x)$ to it.
	
	We define the reward function $r$ on this game as follows:
	\begin{itemize}
		\item $r_\S(\bot\bot) = 0$, and $r_\S(uv) = 1$ for any other transition $vw$;
		
		\item for each player $x$, we have $r_x(uv) = 0$ for every transition leading to a state of the form $v = (C, x)$, and $r_x(uv) = 1$ for any other transition.
	\end{itemize}
	
	Note that Solver can only get the payoffs $0$ (in a play that reaches $\bot$) or $1$ (in any other play).
	Another player $x$ gets the payoff $1$ in a play that never visits (or finitely often, or infinitely often but with negligible frequence) a vertex of the form $(C, x)$.
	Otherwise, he may get any payoff between $0.5$ and $1$, depending on the frequence with which such a state is visited.
	Finally, we initialize that game in $v_0 = C_1$.
	
	\begin{exa}
    The game $G^\phi$, when $\phi$ is the tautology $(x_1 \vee \neg x_1) \wedge \dots \wedge (x_6 \vee \neg x_6)$, is represented by Figure~\ref{fig_Gphi}.
    The rewards that are not written are equal to $1$.
    
    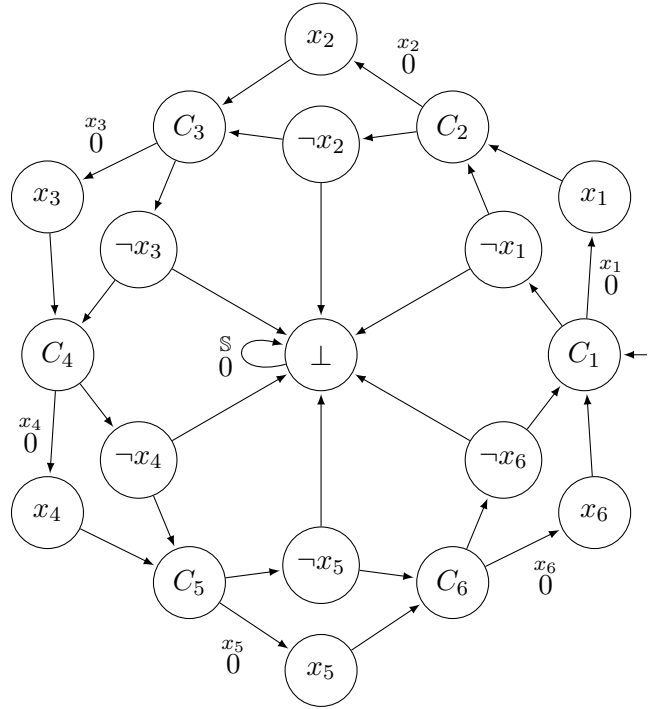
\begin{figure}
        \centering
        \begin{tikzpicture}[->,>=latex,shorten >=1pt, scale=0.7, initial text={}]
    		\node[state, initial right] (C1) at (0:5) {$C_1$};
    		\node[state] (C11) at (30:6) {$x_1$};
    		\node[state] (C12) at (30:4) {$\neg x_1$};
    		\node[state] (C2) at (60:5) {$C_2$};
    		\node[state] (C21) at (90:6) {$x_2$};
    		\node[state] (C22) at (90:4) {$\neg x_2$};
    		\node[state] (C3) at (120:5) {$C_3$};
    		\node[state] (C31) at (150:6) {$x_3$};
    		\node[state] (C32) at (150:4) {$\neg x_3$};
    		\node[state] (C4) at (180:5) {$C_4$};
    		\node[state] (C41) at (210:6) {$x_4$};
    		\node[state] (C42) at (210:4) {$\neg x_4$};
    		\node[state] (C5) at (240:5) {$C_5$};
    		\node[state] (C51) at (270:6) {$x_5$};
    		\node[state] (C52) at (270:4) {$\neg x_5$};
    		\node[state] (C6) at (300:5) {$C_6$};
    		\node[state] (C61) at (330:6) {$x_6$};
    		\node[state] (C62) at (330:4) {$\neg x_6$};
    		\node[state] (b) at (0,0) {$\bot$};
    		
    		\path[->] (C1) edge node[right] {$\stackrel{x_1}{0}$} (C11);
    		\path[->] (C1) edge (C12);
            \path[->] (C11) edge (C2);
            \path[->] (C12) edge (C2);
            \path[->] (C12) edge (b);
            \path[->] (C2) edge node[above right] {$\stackrel{x_2}{0}$} (C21);
    		\path[->] (C2) edge (C22);
            \path[->] (C21) edge (C3);
            \path[->] (C22) edge (C3);
            \path[->] (C22) edge (b);
            \path[->] (C3) edge node[above left] {$\stackrel{x_3}{0}$} (C31);
    		\path[->] (C3) edge (C32);
            \path[->] (C31) edge (C4);
            \path[->] (C32) edge (C4);
            \path[->] (C32) edge (b);
            \path[->] (C4) edge node[left] {$\stackrel{x_4}{0}$} (C41);
    		\path[->] (C4) edge (C42);
            \path[->] (C41) edge (C5);
            \path[->] (C42) edge (C5);
            \path[->] (C42) edge (b);
            \path[->] (C5) edge node[below left] {$\stackrel{x_5}{0}$} (C51);
    		\path[->] (C5) edge (C52);
            \path[->] (C51) edge (C6);
            \path[->] (C52) edge (C6);
            \path[->] (C52) edge (b);
            \path[->] (C6) edge node[below right] {$\stackrel{x_6}{0}$} (C61);
    		\path[->] (C6) edge (C62);
            \path[->] (C61) edge (C1);
            \path[->] (C62) edge (C1);
            \path[->] (C62) edge (b);
            \path (b) edge[loop left] node[left] {$\stackrel{\S}{0}$} (b);
        \end{tikzpicture}
        \caption{The game $G^\phi$}
        \label{fig_Gphi}
    \end{figure}
\end{exa}

	Now, let us prove that there is an SPE in $G^\phi_{\|v_0}$ in which Solver gets the payoff $1$, if and only if the formula $\phi$ is satisfiable, that is, if there exists a valuation $\nu: X \to \{0, 1\}$ that satisfies it.
	
	\begin{itemize}
		\item If such an SPE exists: let us write it $\bsigma$, and let $\rho = \< \bsigma \>$.
		Since $\mu_\S(\rho) = 1$, the sink state $\bot$ is never visited.
		Let us define a valuation $\nu$ on $X$ as follows: for each variable $x$, we have $\nu(x) = 1$ if and only if $\mu_x(\rho) < 1$.
		
		Now, let $C$ be a clause of $\phi$: since $C$, as a state, is necessarily visited infinitely often and with a fixed frequence in the play $\rho$ (because no player ever go to the sink state $\bot$), one of its successors, say $(C, L)$, is visited with a non-negligible frequence (more formally, the time between two occurrences of $(C, L)$ is bounded).
		If $L$ is a positive litteral, say $x$, then by definition of $\nu$, we have $\nu(x) = 1$ and the clause $C$ is satisfied.
		
		If $L$ has the form $\neg x$, then each time the state $(C, \neg x)$ is traversed, player $x$ has the possibility to deviate and to go to the sink state $\bot$, where he is sure to get the payoff $1$.
		Since $\bsigma$ is an SPE, it means that he already gets the payoff $1$ in the play $\rho$.
		By definition of $\nu$, we then have $\nu(x) = 0$, hence the litteral $\neg x$ is satisfied, hence so is the clause $C$.
		
		The valuation $\nu$ satisfies all the clauses of $\phi$, and therefore satisfies the formula $\phi$ itself.

		\item If $\phi$ is satisfied by some valuation $\nu$: let us define a strategy profile $\bsigma$ by:
		\begin{itemize}
			\item $\sigma_\S(hC) = (C, L)$ for each history $h C$ where $C$ is a clause of $\phi$, where $L$ is a litteral of $C$ that is satisfied in the valuation $\nu$;
			
			\item and $\sigma_x(h(C, \neg x)) = \bot$ if and only if $\nu(x) = 1$ for each history $h(C, \neg x)$ where $C$ is a clause of $\phi$ and $x$ is a variable.
		\end{itemize}
	
		Any other state has only one successor, hence we now have completely defined a strategy profile.
		Now, let us prove it is an SPE, in which Solver gets the payoff $1$.
		
		Let $hC$ be a history, where $C$ is a clause of $\phi$.
		We want to prove that $\bsigma_{\|hC}$ is a Nash equilibrium, in which Solver gets the payoff $1$.
		Let $\rho = \< \bsigma_{\|hC} \>$.
		If $\mu_\S(\rho) < 1$, i.e. if $\rho$ is of the form $h D (D, \neg x) \bot^\omega$, then by definition of $\bsigma$ we have $\nu(x) = 0$.
		But then, we cannot have $\sigma_\S(D) = (D, \neg x)$: contradiction.
		The play $\rho$ never reaches the state $\bot$, and Solver gets the payoff $1$, and as a consequence she does not have any profitable deviation.
		
		Now, if another player $x$ has a profitable deviation, it means that he does not get the payoff $1$ in $\rho$, and therefore that some state of the form $(D, x)$ is visited infinitely often.
		But then, if Solver choose to go to the state $(D, x)$, it means that the litteral $x$ is satisfied in $\nu$, i.e. that $\nu(x) = 1$.
		In that case, if some clause $D'$ contains the litteral $\neg x$, it is not a litteral satisfied by $\nu$, and therefore the strategy $\sigma_\S$, as we defined it, never chooses the transition to the state $(D', \neg x)$, where player $x$ could have the possibility to deviate from his strategy.
		Contradiction.

		Finally, after a history of the form $h(C, L)$, either:
		\begin{itemize}
			\item $L = \neg x$ with $\nu(x) = 1$, and in that case, we have $\< \bsigma_{\|h(C, L)} \> = (C, L) \bot^\omega$, player $x$ gets the payoff $1$, and no player has a profitable deviation;
			
			\item or $L$ is a positive litteral, and then there exists only one transition from the state $(C, L)$ to another clause $D$, and we go back to the previous case;
			
			\item or $L = \neg x$ with $\nu(x) = 0$, and in that case, we have $\sigma_x(C, L) = D$ where $D$ is the following clause, and by the first case the strategy profile $\bsigma_{\|h(C, \neg x)D}$ is a Nash equilibrium.
			Moreover, since the litteral $\neg x$ is not satisfied in $\nu$, the play $\< \bsigma_{\|h(C, \neg x)D} \>$ does never traverse again any state of the form $(D', \neg x)$, hence player $x$ wins, and therefore has no profitable deviation: the strategy profile $\bsigma_{\|h(C, \neg x)}$ is a Nash equilibrium.
		\end{itemize}
	\end{itemize}
	
	The SPE threshold problem, and therefore the $\epsilon$-SPE threshold problem, are $\NP$-hard in mean-payoff games.
\end{proof}

        \subsection{Decidability and upper bound}

We can now present an algorithm that decides the $\epsilon$-SPE threshold problem.

\begin{thm}\label{thm_decidable}
    The $\epsilon$-SPE threshold problem in mean-payoff games is decidable and\linebreak[4]$2\ExpTime$-easy.
\end{thm}

\begin{proof}
    Given $G_{\|v_0}, \bx, \by$ and $\epsilon$, by Theorem~\ref{thm_piecewise_affine}, an effective representation of the negotiation function can be computed in a time doubly exponential in the size of $G$; therefore, its least $\epsilon$-fixed point $\lambda^*$ can be computed in a time double exponential in the size of $G$ and of $\epsilon$, using the classical tools given by linear algebra.
    
    Then, by Theorem~\ref{thm_spe} and since mean-payoff games are games with steady negotiation by Lemma~\ref{lm_compute_nego}, the tuple $(G_{\|v_0}, \bx, \by, \epsilon)$ forms a positive instance of the $\epsilon$-SPE threshold problem if and only if there exists a $\lambda^*$-consistent play $\rho$ in $G_{\|v_0}$ with $\bx \leq \mu(\rho) \leq \by$.
    
    The existence of such a play can be decided in doubly exponential time, as follows: for each subset $W \subseteq V$, construct the game $G_W$ defined as the game $G$ in which all the vertices that do not belong to $W$ have been omitted.
    Then, for each connected component $K$ of the underlying graph of $G_W$ that is accessible from $v_0$, construct the polytope:
    $$P = \dseal \left(\underset{c \in \SC(K)}{\Conv} \MP(c)\right) \cap \prod_i [x_i, y_i] \cap \prod_i \left[ \max_{v \in W \cap V_i} \lambda^*(v), +\infty\right),$$
    with the representations and algorithm given in \cite{DBLP:conf/concur/ChatterjeeDEHR10}, which requires a time exponential in the size of $K$.
    If one of those exponentially many polytopes is nonempty, then using Lemma~\ref{lm_dseal}, there exists a play $\rho$ in $G_{\|v_0}$ with $\bx \leq \mu(\rho) \leq \by$ and $\Occ(\rho) \subseteq W$ and $\mu_i(\rho) \geq \lambda^*(v)$ for each $i \in \Pi$ and $v \in V_i \cap W$, i.e. that is $\lambda^*$-consistent.
    Conversely, if such a play exists, then for $W = \Occ(\rho)$ and $K = \Inf(\rho)$, the polytope $P$ will be nonempty.
\end{proof}

\bibliography{bibli}

\newcommand{\etalchar}[1]{$^{#1}$}
\begin{thebibliography}{BRvdB21}

\bibitem[BBG{\etalchar{+}}19]{DBLP:conf/concur/BrihayeBGRB19}
Thomas Brihaye, V{\'{e}}ronique Bruy{\`{e}}re, Aline Goeminne,
  Jean{-}Fran{\c{c}}ois Raskin, and Marie van~den Bogaard.
\newblock The complexity of subgame perfect equilibria in quantitative
  reachability games.
\newblock In {\em {CONCUR}}, volume 140 of {\em LIPIcs}, pages 13:1--13:16.
  Schloss Dagstuhl - Leibniz-Zentrum f{\"{u}}r Informatik, 2019.

\bibitem[BBMR15]{DBLP:conf/csl/BrihayeBMR15}
Thomas Brihaye, V{\'{e}}ronique Bruy{\`{e}}re, No{\'{e}}mie Meunier, and
  Jean{-}Fran{\c{c}}ois Raskin.
\newblock Weak subgame perfect equilibria and their application to quantitative
  reachability.
\newblock In {\em 24th {EACSL} Annual Conference on Computer Science Logic,
  {CSL} 2015, September 7-10, 2015, Berlin, Germany}, volume~41 of {\em
  LIPIcs}, pages 504--518. Schloss Dagstuhl - Leibniz-Zentrum f{\"{u}}r
  Informatik, 2015.

\bibitem[BCH{\etalchar{+}}16]{DBLP:conf/lata/BrenguierCHPRRS16}
Romain Brenguier, Lorenzo Clemente, Paul Hunter, Guillermo~A. P{\'{e}}rez,
  Mickael Randour, Jean{-}Fran{\c{c}}ois Raskin, Ocan Sankur, and Mathieu
  Sassolas.
\newblock Non-zero sum games for reactive synthesis.
\newblock In {\em Language and Automata Theory and Applications - 10th
  International Conference, {LATA} 2016, Prague, Czech Republic, March 14-18,
  2016, Proceedings}, volume 9618 of {\em Lecture Notes in Computer Science},
  pages 3--23. Springer, 2016.

\bibitem[BDS13]{DBLP:conf/lfcs/BrihayePS13}
Thomas Brihaye, Julie {De Pril}, and Sven Schewe.
\newblock Multiplayer cost games with simple {N}ash equilibria.
\newblock In {\em Logical Foundations of Computer Science, International
  Symposium, {LFCS} 2013, San Diego, CA, USA, January 6-8, 2013. Proceedings},
  volume 7734 of {\em Lecture Notes in Computer Science}, pages 59--73.
  Springer, 2013.

\bibitem[BMR14]{DBLP:conf/csl/BruyereMR14}
V{\'{e}}ronique Bruy{\`{e}}re, No{\'{e}}mie Meunier, and Jean{-}Fran{\c{c}}ois
  Raskin.
\newblock Secure equilibria in weighted games.
\newblock In {\em {CSL-LICS}}, pages 26:1--26:26. {ACM}, 2014.

\bibitem[BR15]{DBLP:conf/cav/BrenguierR15}
Romain Brenguier and Jean{-}Fran{\c{c}}ois Raskin.
\newblock Pareto curves of multidimensional mean-payoff games.
\newblock In Daniel Kroening and Corina~S. Pasareanu, editors, {\em Computer
  Aided Verification - 27th International Conference, {CAV} 2015, San
  Francisco, CA, USA, July 18-24, 2015, Proceedings, Part {II}}, volume 9207 of
  {\em Lecture Notes in Computer Science}, pages 251--267. Springer, 2015.
\newblock \href {https://doi.org/10.1007/978-3-319-21668-3\_15}
  {\path{doi:10.1007/978-3-319-21668-3\_15}}.

\bibitem[BRPR16]{DBLP:journals/corr/Bruyere0PR16}
V{\'{e}}ronique Bruy{\`{e}}re, St{\'{e}}phane~Le Roux, Arno Pauly, and
  Jean{-}Fran{\c{c}}ois Raskin.
\newblock On the existence of weak subgame perfect equilibria.
\newblock {\em CoRR}, abs/1612.01402, 2016.

\bibitem[BRPR17]{DBLP:conf/fossacs/Bruyere0PR17}
V{\'{e}}ronique Bruy{\`{e}}re, St{\'{e}}phane~Le Roux, Arno Pauly, and
  Jean{-}Fran{\c{c}}ois Raskin.
\newblock On the existence of weak subgame perfect equilibria.
\newblock In {\em Foundations of Software Science and Computation Structures -
  20th International Conference, {FOSSACS} 2017, Held as Part of the European
  Joint Conferences on Theory and Practice of Software, {ETAPS} 2017, Uppsala,
  Sweden, April 22-29, 2017, Proceedings}, volume 10203 of {\em Lecture Notes
  in Computer Science}, pages 145--161, 2017.

\bibitem[Bru17]{DBLP:conf/dlt/Bruyere17}
V{\'{e}}ronique Bruy{\`{e}}re.
\newblock Computer aided synthesis: {A} game-theoretic approach.
\newblock In {\em Developments in Language Theory - 21st International
  Conference, {DLT} 2017, Li{\`{e}}ge, Belgium, August 7-11, 2017,
  Proceedings}, volume 10396 of {\em Lecture Notes in Computer Science}, pages
  3--35. Springer, 2017.

\bibitem[BRvdB21]{Concur}
L{\'{e}}onard Brice, Jean{-}Fran{\c{c}}ois Raskin, and Marie van~den Bogaard.
\newblock Subgame-perfect equilibria in mean-payoff games.
\newblock {\em CoRR}, 2021.

\bibitem[CDE{\etalchar{+}}10]{DBLP:conf/concur/ChatterjeeDEHR10}
Krishnendu Chatterjee, Laurent Doyen, Herbert Edelsbrunner, Thomas~A.
  Henzinger, and Philippe Rannou.
\newblock Mean-payoff automaton expressions.
\newblock In Paul Gastin and Fran{\c{c}}ois Laroussinie, editors, {\em {CONCUR}
  2010 - Concurrency Theory, 21th International Conference, {CONCUR} 2010,
  Paris, France, August 31-September 3, 2010. Proceedings}, volume 6269 of {\em
  Lecture Notes in Computer Science}, pages 269--283. Springer, 2010.

\bibitem[CHP10]{DBLP:journals/iandc/ChatterjeeHP10}
Krishnendu Chatterjee, Thomas~A. Henzinger, and Nir Piterman.
\newblock Strategy logic.
\newblock {\em Inf. Comput.}, 208(6):677--693, 2010.
\newblock \href {https://doi.org/10.1016/j.ic.2009.07.004}
  {\path{doi:10.1016/j.ic.2009.07.004}}.

\bibitem[FKM{\etalchar{+}}10]{DBLP:journals/mor/FleschKMSSV10}
J{\'{a}}nos Flesch, Jeroen Kuipers, Ayala Mashiah{-}Yaakovi, Gijs Schoenmakers,
  Eilon Solan, and Koos Vrieze.
\newblock Perfect-information games with lower-semicontinuous payoffs.
\newblock {\em Math. Oper. Res.}, 35(4):742--755, 2010.

\bibitem[FP16]{DBLP:journals/ijgt/FleschP16}
J{\'{a}}nos Flesch and Arkadi Predtetchinski.
\newblock On refinements of subgame perfect
  {\textbackslash}({\textbackslash}epsilon {\textbackslash}) -equilibrium.
\newblock {\em Int. J. Game Theory}, 45(3):523--542, 2016.
\newblock \href {https://doi.org/10.1007/s00182-015-0468-8}
  {\path{doi:10.1007/s00182-015-0468-8}}.

\bibitem[FP17]{DBLP:journals/mor/FleschP17}
J{\'{a}}nos Flesch and Arkadi Predtetchinski.
\newblock A characterization of subgame-perfect equilibrium plays in {B}orel
  games of perfect information.
\newblock {\em Math. Oper. Res.}, 42(4):1162--1179, 2017.
\newblock \href {https://doi.org/10.1287/moor.2016.0843}
  {\path{doi:10.1287/moor.2016.0843}}.

\bibitem[Kop06]{DBLP:conf/icalp/Kopczynski06}
Eryk Kopczynski.
\newblock Half-positional determinacy of infinite games.
\newblock In Michele Bugliesi, Bart Preneel, Vladimiro Sassone, and Ingo
  Wegener, editors, {\em Automata, Languages and Programming, 33rd
  International Colloquium, {ICALP} 2006, Venice, Italy, July 10-14, 2006,
  Proceedings, Part {II}}, volume 4052 of {\em Lecture Notes in Computer
  Science}, pages 336--347. Springer, 2006.

\bibitem[KPV16]{DBLP:journals/amai/KupfermanPV16}
Orna Kupferman, Giuseppe Perelli, and Moshe~Y. Vardi.
\newblock Synthesis with rational environments.
\newblock {\em Ann. Math. Artif. Intell.}, 78(1):3--20, 2016.
\newblock \href {https://doi.org/10.1007/s10472-016-9508-8}
  {\path{doi:10.1007/s10472-016-9508-8}}.

\bibitem[Mar75]{BorelDeterminacy}
Donald~A. Martin.
\newblock Borel determinacy.
\newblock {\em Annals of Mathematics}, pages 363--371, 1975.

\bibitem[Meu16]{thesis_noemie}
No\'{e}mie Meunier.
\newblock {\em Multi-Player Quantitative Games: Equilibria and Algorithms}.
\newblock PhD thesis, Universit\'{e} de Mons, 2016.

\bibitem[Osb04]{Osborne}
{Martin J.} Osborne.
\newblock {\em An introduction to game theory}.
\newblock Oxford Univ. Press, 2004.

\bibitem[SV03]{solan2003deterministic}
Eilon Solan and Nicolas Vieille.
\newblock Deterministic multi-player {D}ynkin games.
\newblock {\em Journal of Mathematical Economics}, 39(8):911--929, 2003.

\bibitem[Umm06]{DBLP:conf/fsttcs/Ummels06}
Michael Ummels.
\newblock Rational behaviour and strategy construction in infinite multiplayer
  games.
\newblock In {\em {FSTTCS} 2006: Foundations of Software Technology and
  Theoretical Computer Science, 26th International Conference, Kolkata, India,
  December 13-15, 2006, Proceedings}, volume 4337 of {\em Lecture Notes in
  Computer Science}, pages 212--223. Springer, 2006.

\bibitem[Umm08]{DBLP:conf/fossacs/Ummels08}
Michael Ummels.
\newblock The complexity of {N}ash equilibria in infinite multiplayer games.
\newblock In Roberto~M. Amadio, editor, {\em Foundations of Software Science
  and Computational Structures, 11th International Conference, {FOSSACS} 2008,
  Held as Part of the Joint European Conferences on Theory and Practice of
  Software, {ETAPS} 2008, Budapest, Hungary, March 29 - April 6, 2008.
  Proceedings}, volume 4962 of {\em Lecture Notes in Computer Science}, pages
  20--34. Springer, 2008.
\newblock \href {https://doi.org/10.1007/978-3-540-78499-9\_3}
  {\path{doi:10.1007/978-3-540-78499-9\_3}}.

\bibitem[VCD{\etalchar{+}}15]{DBLP:journals/iandc/VelnerC0HRR15}
Yaron Velner, Krishnendu Chatterjee, Laurent Doyen, Thomas~A. Henzinger,
  Alexander~Moshe Rabinovich, and Jean{-}Fran{\c{c}}ois Raskin.
\newblock The complexity of multi-mean-payoff and multi-energy games.
\newblock {\em Inf. Comput.}, 241:177--196, 2015.

\bibitem[VS03]{vieille:hal-00464953}
Nicolas Vieille and Eilon Solan.
\newblock {Deterministic multi-player Dynkin games}.
\newblock {\em {Journal of Mathematical Economics}}, Vol.39,num. 8:pp.911--929,
  November 2003.
\newblock \href {https://doi.org/10.1016/S0304-4068(03)00021-1}
  {\path{doi:10.1016/S0304-4068(03)00021-1}}.

\end{thebibliography}

\newpage


\appendix

		\section{Abstract negotiation game} \label{app_abstract}

\begin{defi}[Abstract negotiation game]
	Let $G_{\|v_0}$ be a game, let $i \in \Pi$, and let $\lambda$ be a requirement on $G$. The \emph{abstract negotiation game} of $G_{\|v_0}$ for player $i$ with requirement $\lambda$ is the two-player zero-sum game:
	
	$$\Abs_{\lambda i}(G)_{\|[v_0]} = \left( \{\P, \C\}, S, (S_{\P}, S_{\C}), \Delta, \nu\right)_{\|[v_0]},$$
	
	where:
	
	\begin{itemize}
		\item $\P$ denotes the player \emph{Prover} and $\C$ the player \emph{Challenger};
		
		\item the states of $S_\C$ are written $[\rho]$, where $\rho$ is a $\lambda$-consistent play in $G$;
		
		\item the states of $S_\P$ are written $[hv]$, where $hv$ is a history in $G$, with $h \in \Hist_i(G)$, or $[v]$ with $v \in V$, plus two additional states $\top$ and $\bot$;
		
		\item the set $\Delta$ contains the transitions of the form:
		
		\begin{itemize}
			\item $[hv][v\rho]$, where $[hv] \in S_\P$ and $[v\rho] \in S_{\C}$ (Prover proposes a play);
			
			\item $[\rho] [\rho_0...\rho_n v]$, where $[\rho] \in S_{\C}, n \in \N,	\rho_n \in V_i$, and $v \neq \rho_{n+1}$ (Challenger makes player $i$ deviate);
			
			\item $[\rho] \top$, where $[\rho] \in S_{\C}$ (Challenger accepts the proposed play);
			
			\item $\top \top$ (the game is over);
			
			\item $[hv] \bot$ (Prover has no more play to propose);
			
			\item $\bot \bot$ (the game is over).
		\end{itemize}

		\item $\nu$ is the payoff function defined by, for all $\rho^{(0)}, \rho^{(1)}, \dots, h^{(1)} v_1, h^{(2)} v_2, \dots, k, H$:
	
		$$\begin{matrix}
			& \nu_\C \left( [v_0] \left[\rho^{(0)}\right] \left[h^{(1)}v_1\right] \left[\rho^{(1)}\right] \dots \left[h^{(k)} v_k\right] \left[\rho^{(k)}\right] \top^\omega \right) \\[1mm]
			= & \mu_i\left(h^{(1)} \dots h^{(k)} \rho^{(k)}\right), \\[2mm]
			
			& \nu_\C \left( [v_0] \left[\rho^{(0)}\right] \left[h^{(1)}v_1\right] \left[\rho^{(1)}\right] \dots \left[h^{(n)} v_n\right] \left[\rho^{(n)}\right] \dots \right) \\[1mm]
			= & \mu_i\left(h^{(1)} h^{(2)} \dots\right), \\[2mm]
			
			& \nu_\C\left(H \bot^\omega \right) = +\infty,
		\end{matrix}$$
		
		and by $\nu_\P = -\nu_\C$.
	\end{itemize}
\end{defi}

\begin{rem}
If the game $G$ is Borel, then so is the game $\Abs_{\lambda i}(G)$.
\end{rem}

\begin{prop}
	Let $G_{\|v_0}$ be a Borel game, let $\lambda$ be a requirement on $G$ and let $i \in \Pi$ be such that $v_0 \in V_i$.
	Then, we have: 
	$$\val_\C\left(\Abs_{\lambda i}(G)_{\|[v_0]}\right) = \nego(\lambda)(v_0).$$
\end{prop}

\begin{proof}
Let $\alpha \in \R$, and let us prove that the following statements are equivalent:

\begin{enumerate}
	\item \label{case1_pf_abs} there exists a strategy $\tau_\P$ such that for every strategy $\tau_\C$, $\nu_\C(\< \btau \>) < \alpha$;
	
	\item \label{case2_pf_abs} there exists a $\lambda$-rational strategy profile $\bsigma_{-i}$ in the game $G_{\|v_0}$ such that for every strategy $\sigma_i$, we have $\mu_i\left( \< \bsigma \> \right) < \alpha$.
\end{enumerate}

\begin{itemize}
	\item \emph{(\ref{case1_pf_abs}) implies (\ref{case2_pf_abs}).}
	
	Let $\tau_\P$ be such that for every strategy $\tau_\C$, $\nu_\C(\< \btau \>) < \alpha$.
	
	In what follows, any history $h$ compatible with an already defined strategy profile $\bsigma_{-i}$ in $G_{\|v_0}$ will be decomposed in:
	$$h = v_0 h^{(0)} v_1 h^{(1)} \dots h^{(n-1)} v_n h^{(n)},$$
	so that there exist plays $\rho^{(0)}, \dots, \rho^{(n-1)}, \eta$ and a history:
	$$[v_0] \left[\rho^{(0)}\right] \left[v_1 h^{(1)} v_2\right] \dots \left[v_{n-1}h^{(n-1)}v_n\right] \left[v_nh^{(n)}\eta\right]$$
	in the game $\Abs_{\lambda i}(G)$ compatible with $\tau_{\P}$: the existence and the unicity of that decomposition can be proved by induction. Intuitively, the history $h$ is cut in histories which are prefixes of plays that can be proposed by Prover.
	
	Then, let us define inductively the strategy profile $\bsigma_{-i}$ by, for every $h$ such that $\bsigma_{-i}$ has been defined on the prefixes of $h$, and such that the last state of $h$ is not controlled by player $i$, $\bsigma_{-i}(h) = \eta_0$ with $\eta$ defined from $h$ as higher. Let us prove that $\bsigma_{-i}$ is the desired strategy profile.
	
	\begin{itemize}
		\item \emph{The strategy profile $\bsigma_{-i}$ is $\lambda$-rational.}
		
		Let us define $\sigma_i$ so that for every history $hv$ compatible with $\bsigma_{-i}$, the play $\< \bsigma_{\|hv} \>$ is $\lambda$-consistent.
		
		For any history:
		$$h = v_0 h^{(0)} v_1 h^{(1)} \dots h^{(n-1)} v_n h^{(n)}$$
		compatible with $\bsigma_{-i}$ and ending in $V_i$, let $\sigma_i(h) = \eta_0$ with $\eta$ corresponding to the decomposition of $h$, so that by induction:
		$$\< \bsigma_{\|v_0 h^{(0)} v_1 h^{(1)} \dots h^{(n-1)} v_n} \> = v_n h^{(n)} \eta.$$
		
		Let now $hv$ be a history in $G_{\|v_0}$, and let us show that the play $\< \bsigma_{\|hv} \>$ is $\lambda$-consistent. If we decompose:
		$$hv = v_0 h^{(0)} v_1 h^{(1)} \dots h^{(n-1)} v_n h^{(n)}$$
		with the same definition of $\eta$ (note that the vertex $v$ is now included in the decomposition), then $\< \bsigma_{\|hv} \> = v\eta$, and by definition of the abstract negotiation game, $v_n h^{(n)} \eta$ is a $\lambda$-consistent play, and therefore so is $v \eta$.

		\item \emph{The strategy profile $\bsigma_{-i}$ keeps player $i$'s payoff under the value $\alpha$.}
		
		Let $\sigma_i$ be a strategy for player $i$, and let $\rho = \< \bsigma \>$. We want to prove that $\mu_i(\rho) < \alpha$.
		
		Let us define two finite or infinite sequences $\left( \rho^{(k)} \right)_{k \in K}$ and $\left( h^{(k)} v_k \right)_{k \in K}$, where $K = \{1, \dots, n\}$ or $K = \N \setminus \{0\}$, by for every $k \in K$:
		
		$$\left[ \rho^{(k)} \right] = \tau_{\P} \left( [v_0] \left[ \rho^{(0)} \right] \dots \left[ \rho^{(k-1)} \right] \left[ h^{(k)} v_k \right] \right)$$
		and so that for every $k$, the history $h^{(k)} v_k$ is the shortest prefix of $\rho$ that is not a prefix of $h^{(1)} \dots h^{(k-1)} \rho^{(k-1)}$ (or equivalently, the history $h^{(k)}$ is the longest common prefix of $\rho$ and $h^{(1)} \dots h^{(k-1)} \rho^{(k-1)}$).
		
		Then, the length of the longest common prefix of $h^{(1)} \dots h^{(k-1)} \rho^{(k)}$ and $\rho$ increases with $k$, and the set $K$ is finite if and only if there exists $n$ such that $h^{(1)} \dots h^{(n-1)} \rho^{(n)} = \rho$.
		
		In the infinite case, let:
		$$\chi = [v_0] \left[ \rho^{(0)} \right] \left[ h^{(1)} v_1 \right] \dots \left[ \rho^{(k)} \right] \left[ h^{(k)} v_k \right] \dots.$$
		The play $\chi$ is compatible with $\tau_{\P}$, hence $\nu_\C(\chi) < \alpha$, that is to say:
		$$\mu_i\left(h^{(1)} h^{(2)} \dots\right) < \alpha,$$
		ie. $\mu_i(\rho) < \alpha$.
		
		In the finite case, let:
		$$\chi = [v_0] \left[ \rho^{(0)} \right] \left[ h^{(1)} v_1 \right] \dots \left[ \rho^{(n)} \right] \top^\omega.$$
		For the same reason, $\nu_\C(\chi) < \alpha$, that is to say $\mu_i \left( h^{(1)} \dots h^{(n)} \rho^{(n)} \right) = \mu_i(\rho) < \alpha$.
	\end{itemize}
	
	\item \emph{(\ref{case2_pf_abs}) implies (\ref{case1_pf_abs}).}
	
	Let $\bsigma_{-i}$ be a strategy profile keeping player $i$'s payoff below $\alpha$, $\lambda$-rational assuming a strategy $\sigma_i$.
	Let us define a strategy $\tau_{\P}$ for Prover in the abstract negotiation game.
	
	Let $H = [v_0] \left[\rho^{(0)}\right] \left[h^{(1)}v_1\right] \left[\rho^{(1)}\right] \dots \left[h^{(n)} v_n\right]$ be a history in the abstract game, ending in $S_\P$. Then, we define:
	$$\tau_{\P}(H) = \left[ \< \bsigma_{\|h^{(1)} \dots h^{(n)}v_n} \> \right].$$
	
	If $H$ is a history ending in $\top$, then $\tau_\P(H) = \top$, and in the same way if $H$ ends in $\bot$, then $\tau_\P(H) = \bot$.
	
	Let us show that $\tau_{\P}$ is the strategy we were looking for. Let $\chi$ be a play compatible with $\tau_{\P}$, and let us note that the state $\bot$ does not appear in $\chi$. Then, the play $\chi$ can only have two forms:
	
	\begin{itemize}
		\item If $\chi = [v_0] \left[\rho^{(0)}\right] \left[h^{(1)}v_1\right]  \dots \left[\rho^{(n)}\right] \top^\omega$, then we have:
		$$\rho^{(n)} = \< \bsigma_{\|h^{(1)} \dots h^{(n)} v_n} \>,$$
		and the history $h^{(1)} \dots h^{(n)} v_n$ in the game $G_{\|v_0}$ is compatible with $\bsigma_{-i}$. By hypothesis, we have:
		$$\mu_i\left(h^{(1)} \dots h^{(n)} \rho^{(n)}\right) < \alpha,$$
		hence $\nu_\C(\chi) < \alpha$.
		
		\item If $\chi = [v_0] \left[\rho^{(0)}\right]  \dots \left[h^{(n)} v_n\right] \left[\rho^{(n)}\right] \dots$, then the play $\rho = h^{(1)} h^{(2)} \dots$ is compatible with $\bsigma_{-i}$, and by hypothesis $\mu_i(\rho) < \alpha$, hence $\nu_\C(\chi) < \alpha$. \qedhere
	\end{itemize}
\end{itemize}
\end{proof}

	\section{Some examples of negotiation sequences} \label{app_ex}

We gather in this section some examples that could be interesting for the reader who would want to get a full overall view on the behaviour of the negotiation function on the mean-payoff games.

\begin{exa} \label{ex_sans_spe_sequence}
Let us take again the game of Example~\ref{ex_sans_spe}: let us give (in red) the values of $\lambda_1 = \nego(\lambda_0)$, which correspond to the antagonistic values.

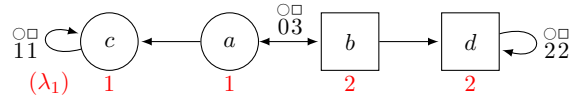
\begin{center}
	\begin{tikzpicture}[->,>=latex,shorten >=1pt, initial text={}, scale=0.8, every node/.style={scale=0.8}]
	\node[state] (a) at (0, 0) {$a$};
	\node[state] (c) at (-2, 0) {$c$};
	\node[state, rectangle] (b) at (2, 0) {$b$};
	\node[state, rectangle] (d) at (4, 0) {$d$};
	\path (a) edge (c);
	\path[<->] (a) edge node[above] {$\stackrel{\playcircle}{0} \stackrel{\Box}{3}$} (b);
	\path (b) edge (d);
	\path (d) edge [loop right] node {$\stackrel{\playcircle}{2} \stackrel{\Box}{2}$} (d);
	\path (c) edge [loop left] node {$\stackrel{\playcircle}{1} \stackrel{\Box}{1}$} (c);
	
	\node[red] (l) at (-3, -0.7) {$(\lambda_1)$};
	\node[red] (la) at (0, -0.7) {$1$};
	\node[red] (lb) at (2, -0.7) {$2$};
	\node[red] (lc) at (-2, -0.7) {$1$};
	\node[red] (ld) at (4, -0.7) {$2$};
	\end{tikzpicture}
\end{center}

At the second step, let us execute the abstract game on the state $a$, with the requirement $\lambda_1$: whatever Prover proposes at first, Challenger has the possibility to deviate and to reach the state $b$. Then, Prover has to propose a $\lambda_1$-consistent play from the state $b$, i.e. a play in which player $\Circle$ gets at least the payoff $2$: such a play necessarily ends in the state $d$, and gives player $\Box$ the payoff $2$.

The other states keep the same values.

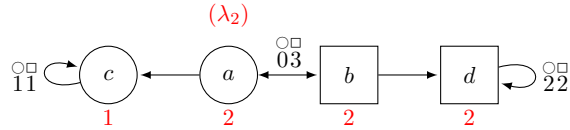
\begin{center}
	\begin{tikzpicture}[<->,>=latex,shorten >=1pt, , scale=0.8, every node/.style={scale=0.8}]
	\node[state] (a) at (0, 0) {$a$};
	\node[state] (c) at (-2, 0) {$c$};
	\node[state, rectangle] (b) at (2, 0) {$b$};
	\node[state, rectangle] (d) at (4, 0) {$d$};
	\path[->] (a) edge (c);
	\path[<->] (a) edge node[above] {$\stackrel{\playcircle}{0} \stackrel{\Box}{3}$} (b);
	\path[->] (b) edge (d);
	\path (d) edge [loop right] node {$\stackrel{\playcircle}{2} \stackrel{\Box}{2}$} (d);
	\path (c) edge [loop left] node {$\stackrel{\playcircle}{1} \stackrel{\Box}{1}$} (c);
	
	\node[red] (l) at (0, 1) {$(\lambda_2)$};
	\node[red] (la) at (0, -0.7) {$2$};
	\node[red] (lb) at (2, -0.7) {$2$};
	\node[red] (lc) at (-2, -0.7) {$1$};
	\node[red] (ld) at (4, -0.7) {$2$};
	\end{tikzpicture}
\end{center}

But then, at the third step, from the state $b$: whatever Prover proposes at first, Challenger can deviate to reach the state $a$. Then, Prover has to propose a $\lambda_2$-consistent play from $a$, i.e. a play in which player $\Circle$ gets at least the payoff $2$: such a play necessarily end in the state $d$, i.e. after possibly some prefix, Prover proposes the play $abd^\omega$. But then, Challenger can always deviate to go back to the state $a$; and the play which is thus created is $(ab)^\omega$ which gives player $\Box$ the payoff $3$.

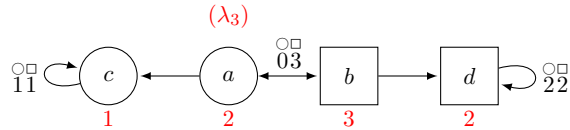
\begin{center}
	\begin{tikzpicture}[<->,>=latex,shorten >=1pt, , scale=0.8, every node/.style={scale=0.8}]
	\node[state] (a) at (0, 0) {$a$};
	\node[state] (c) at (-2, 0) {$c$};
	\node[state, rectangle] (b) at (2, 0) {$b$};
	\node[state, rectangle] (d) at (4, 0) {$d$};
	\path[->] (a) edge (c);
	\path[<->] (a) edge node[above] {$\stackrel{\playcircle}{0} \stackrel{\Box}{3}$} (b);
	\path[->] (b) edge (d);
	\path (d) edge [loop right] node {$\stackrel{\playcircle}{2} \stackrel{\Box}{2}$} (d);
	\path (c) edge [loop left] node {$\stackrel{\playcircle}{1} \stackrel{\Box}{1}$} (c);
	
	\node[red] (l) at (0, 1) {$(\lambda_3)$};
	\node[red] (la) at (0, -0.7) {$2$};
	\node[red] (lb) at (2, -0.7) {$3$};
	\node[red] (lc) at (-2, -0.7) {$1$};
	\node[red] (ld) at (4, -0.7) {$2$};
	\end{tikzpicture}
\end{center}

Finally, from the states $a$ and $b$, there exists no $\lambda_3$-consistent play, and therefore no $\lambda$-rational strategy profile.

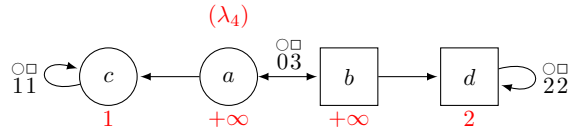
\begin{center}
	\begin{tikzpicture}[<->,>=latex,shorten >=1pt, , scale=0.8, every node/.style={scale=0.8}]
	\node[state] (a) at (0, 0) {$a$};
	\node[state] (c) at (-2, 0) {$c$};
	\node[state, rectangle] (b) at (2, 0) {$b$};
	\node[state, rectangle] (d) at (4, 0) {$d$};
	\path[->] (a) edge (c);
	\path[<->] (a) edge node[above] {$\stackrel{\playcircle}{0} \stackrel{\Box}{3}$} (b);
	\path[->] (b) edge (d);
	\path (d) edge [loop right] node {$\stackrel{\playcircle}{2} \stackrel{\Box}{2}$} (d);
	\path (c) edge [loop left] node {$\stackrel{\playcircle}{1} \stackrel{\Box}{1}$} (c);
	
	\node[red] (l) at (0, 1) {$(\lambda_4)$};
	\node[red] (la) at (0, -0.7) {$+\infty$};
	\node[red] (lb) at (2, -0.7) {$+\infty$};
	\node[red] (lc) at (-2, -0.7) {$1$};
	\node[red] (ld) at (4, -0.7) {$2$};
	\end{tikzpicture}
\end{center}
and for all $n \geq 4$, $\lambda_n = \lambda_4$.
\end{exa}

\begin{exa} \label{ex_sans_spe_sc}
In all the previous examples, all the games whose underlying graphs were strongly connected contained SPEs.
Here is an example of game with a strongly connected underlying graph that does not contain SPEs.
This game is similar to the game of Example~\ref{ex_sans_spe}, hence we do not give the details of the computation of the negotiation sequence.

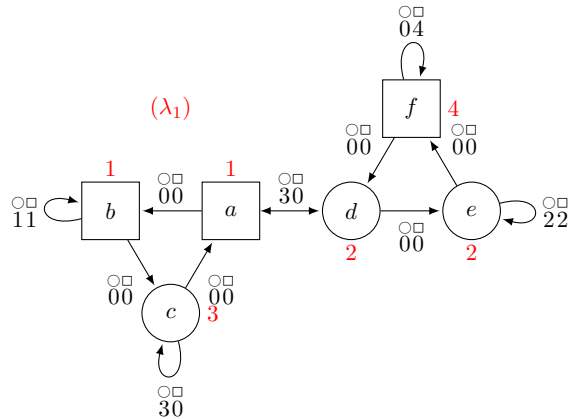
\begin{center}
	\begin{tikzpicture}[<->,>=latex,shorten >=1pt, scale=0.8, every node/.style={scale=0.8}]
	\node[state, rectangle] (b) at (0, 0) {$b$};
	\node[state] (c) at (1, -1.7) {$c$};
	\node[state, rectangle] (a) at (2, 0) {$a$};
	\node[state] (d) at (4, 0) {$d$};
	\node[state] (e) at (6, 0) {$e$};
	\node[state, rectangle] (f) at (5, 1.7) {$f$};
	
	\path[<->] (a) edge node[above] {$\stackrel{\playcircle}{3} \stackrel{\Box}{0}$} (d);
	\path[->] (a) edge node[above] {$\stackrel{\playcircle}{0} \stackrel{\Box}{0}$} (b);
	\path[->] (b) edge node[below left] {$\stackrel{\playcircle}{0} \stackrel{\Box}{0}$} (c);
	\path[->] (c) edge node[below right] {$\stackrel{\playcircle}{0} \stackrel{\Box}{0}$} (a);
	\path (b) edge [loop left] node {$\stackrel{\playcircle}{1} \stackrel{\Box}{1}$} (b);
	\path (c) edge [loop below] node {$\stackrel{\playcircle}{3} \stackrel{\Box}{0}$} (c);
	\path[->] (d) edge node[below] {$\stackrel{\playcircle}{0} \stackrel{\Box}{0}$} (e);
	\path[->] (e) edge node[above right] {$\stackrel{\playcircle}{0} \stackrel{\Box}{0}$} (f);
	\path[->] (f) edge node[above left] {$\stackrel{\playcircle}{0} \stackrel{\Box}{0}$} (d);
	\path (e) edge [loop right] node {$\stackrel{\playcircle}{2} \stackrel{\Box}{2}$} (e);
	\path (f) edge [loop above] node {$\stackrel{\playcircle}{0} \stackrel{\Box}{4}$} (f);
	
	\node[red] (l) at (1, 1.7) {$(\lambda_1)$};
	\node[red] (la) at (2, 0.7) {$1$};
	\node[red] (lb) at (0, 0.7) {$1$};
	\node[red] (lc) at (1.7, -1.7) {$3$};
	\node[red] (ld) at (4, -0.7) {$2$};
	\node[red] (le) at (6, -0.7) {$2$};
	\node[red] (lf) at (5.7, 1.7) {$4$};
	\end{tikzpicture}
\end{center}

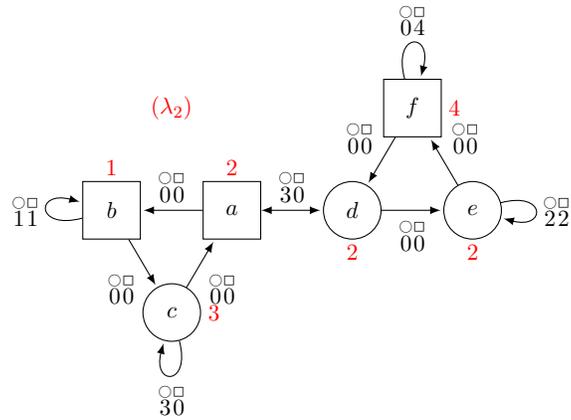
\begin{center}
	\begin{tikzpicture}[<->,>=latex,shorten >=1pt, scale=0.8, every node/.style={scale=0.8}]
	\node[state, rectangle] (b) at (0, 0) {$b$};
	\node[state] (c) at (1, -1.7) {$c$};
	\node[state, rectangle] (a) at (2, 0) {$a$};
	\node[state] (d) at (4, 0) {$d$};
	\node[state] (e) at (6, 0) {$e$};
	\node[state, rectangle] (f) at (5, 1.7) {$f$};
	
	\path[<->] (a) edge node[above] {$\stackrel{\playcircle}{3} \stackrel{\Box}{0}$} (d);
	\path[->] (a) edge node[above] {$\stackrel{\playcircle}{0} \stackrel{\Box}{0}$} (b);
	\path[->] (b) edge node[below left] {$\stackrel{\playcircle}{0} \stackrel{\Box}{0}$} (c);
	\path[->] (c) edge node[below right] {$\stackrel{\playcircle}{0} \stackrel{\Box}{0}$} (a);
	\path (b) edge [loop left] node {$\stackrel{\playcircle}{1} \stackrel{\Box}{1}$} (b);
	\path (c) edge [loop below] node {$\stackrel{\playcircle}{3} \stackrel{\Box}{0}$} (c);
	\path[->] (d) edge node[below] {$\stackrel{\playcircle}{0} \stackrel{\Box}{0}$} (e);
	\path[->] (e) edge node[above right] {$\stackrel{\playcircle}{0} \stackrel{\Box}{0}$} (f);
	\path[->] (f) edge node[above left] {$\stackrel{\playcircle}{0} \stackrel{\Box}{0}$} (d);
	\path (e) edge [loop right] node {$\stackrel{\playcircle}{2} \stackrel{\Box}{2}$} (e);
	\path (f) edge [loop above] node {$\stackrel{\playcircle}{0} \stackrel{\Box}{4}$} (f);
	
	\node[red] (l) at (1, 1.7) {$(\lambda_2)$};
	\node[red] (la) at (2, 0.7) {$2$};
	\node[red] (lb) at (0, 0.7) {$1$};
	\node[red] (lc) at (1.7, -1.7) {$3$};
	\node[red] (ld) at (4, -0.7) {$2$};
	\node[red] (le) at (6, -0.7) {$2$};
	\node[red] (lf) at (5.7, 1.7) {$4$};
	\end{tikzpicture}
\end{center}

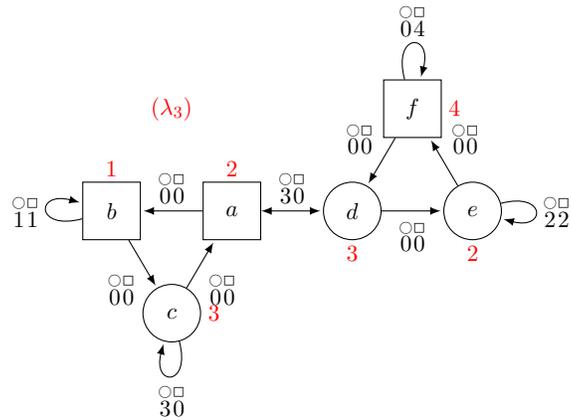
\begin{center}
	\begin{tikzpicture}[<->,>=latex,shorten >=1pt, scale=0.8, every node/.style={scale=0.8}]
	\node[state, rectangle] (b) at (0, 0) {$b$};
	\node[state] (c) at (1, -1.7) {$c$};
	\node[state, rectangle] (a) at (2, 0) {$a$};
	\node[state] (d) at (4, 0) {$d$};
	\node[state] (e) at (6, 0) {$e$};
	\node[state, rectangle] (f) at (5, 1.7) {$f$};
	
	\path[<->] (a) edge node[above] {$\stackrel{\playcircle}{3} \stackrel{\Box}{0}$} (d);
	\path[->] (a) edge node[above] {$\stackrel{\playcircle}{0} \stackrel{\Box}{0}$} (b);
	\path[->] (b) edge node[below left] {$\stackrel{\playcircle}{0} \stackrel{\Box}{0}$} (c);
	\path[->] (c) edge node[below right] {$\stackrel{\playcircle}{0} \stackrel{\Box}{0}$} (a);
	\path (b) edge [loop left] node {$\stackrel{\playcircle}{1} \stackrel{\Box}{1}$} (b);
	\path (c) edge [loop below] node {$\stackrel{\playcircle}{3} \stackrel{\Box}{0}$} (c);
	\path[->] (d) edge node[below] {$\stackrel{\playcircle}{0} \stackrel{\Box}{0}$} (e);
	\path[->] (e) edge node[above right] {$\stackrel{\playcircle}{0} \stackrel{\Box}{0}$} (f);
	\path[->] (f) edge node[above left] {$\stackrel{\playcircle}{0} \stackrel{\Box}{0}$} (d);
	\path (e) edge [loop right] node {$\stackrel{\playcircle}{2} \stackrel{\Box}{2}$} (e);
	\path (f) edge [loop above] node {$\stackrel{\playcircle}{0} \stackrel{\Box}{4}$} (f);
	
	\node[red] (l) at (1, 1.7) {$(\lambda_3)$};
	\node[red] (la) at (2, 0.7) {$2$};
	\node[red] (lb) at (0, 0.7) {$1$};
	\node[red] (lc) at (1.7, -1.7) {$3$};
	\node[red] (ld) at (4, -0.7) {$3$};
	\node[red] (le) at (6, -0.7) {$2$};
	\node[red] (lf) at (5.7, 1.7) {$4$};
	\end{tikzpicture}
\end{center}

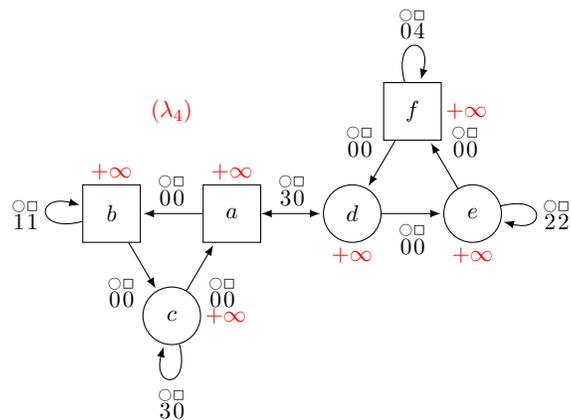
\begin{center}
	\begin{tikzpicture}[<->,>=latex,shorten >=1pt, scale=0.8, every node/.style={scale=0.8}]
	\node[state, rectangle] (b) at (0, 0) {$b$};
	\node[state] (c) at (1, -1.7) {$c$};
	\node[state, rectangle] (a) at (2, 0) {$a$};
	\node[state] (d) at (4, 0) {$d$};
	\node[state] (e) at (6, 0) {$e$};
	\node[state, rectangle] (f) at (5, 1.7) {$f$};
	
	\path[<->] (a) edge node[above] {$\stackrel{\playcircle}{3} \stackrel{\Box}{0}$} (d);
	\path[->] (a) edge node[above] {$\stackrel{\playcircle}{0} \stackrel{\Box}{0}$} (b);
	\path[->] (b) edge node[below left] {$\stackrel{\playcircle}{0} \stackrel{\Box}{0}$} (c);
	\path[->] (c) edge node[below right] {$\stackrel{\playcircle}{0} \stackrel{\Box}{0}$} (a);
	\path (b) edge [loop left] node {$\stackrel{\playcircle}{1} \stackrel{\Box}{1}$} (b);
	\path (c) edge [loop below] node {$\stackrel{\playcircle}{3} \stackrel{\Box}{0}$} (c);
	\path[->] (d) edge node[below] {$\stackrel{\playcircle}{0} \stackrel{\Box}{0}$} (e);
	\path[->] (e) edge node[above right] {$\stackrel{\playcircle}{0} \stackrel{\Box}{0}$} (f);
	\path[->] (f) edge node[above left] {$\stackrel{\playcircle}{0} \stackrel{\Box}{0}$} (d);
	\path (e) edge [loop right] node {$\stackrel{\playcircle}{2} \stackrel{\Box}{2}$} (e);
	\path (f) edge [loop above] node {$\stackrel{\playcircle}{0} \stackrel{\Box}{4}$} (f);
	
	\node[red] (l) at (1, 1.7) {$(\lambda_4)$};
	\node[red] (la) at (2, 0.7) {$+\infty$};
	\node[red] (lb) at (0, 0.7) {$+\infty$};
	\node[red] (lc) at (1.9, -1.7) {$+\infty$};
	\node[red] (ld) at (4, -0.7) {$+\infty$};
	\node[red] (le) at (6, -0.7) {$+\infty$};
	\node[red] (lf) at (5.9, 1.7) {$+\infty$};
	\end{tikzpicture}
\end{center}
\end{exa}

\begin{exa} \label{ex_big}
This example shows how a new requirement can emerge from the combination of several cycles.

Let $G$ be the following game:

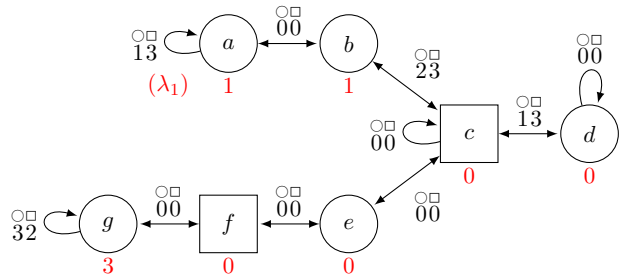
\begin{center}
	\begin{tikzpicture}[->,>=latex,shorten >=1pt, scale=0.8, every node/.style={scale=0.8}]
	\node[state] (a) at (2, 1.5) {$a$};
	\node[state] (b) at (4, 1.5) {$b$};
	\node[state, rectangle] (c) at (6, 0) {$c$};
	\node[state] (d) at (8, 0) {$d$};
	\node[state] (e) at (4, -1.5) {$e$};
	\node[state, rectangle] (f) at (2, -1.5) {$f$};
	\node[state] (g) at (0, -1.5) {$g$};
	
	\node[red] (l) at (1, 0.8) {$(\lambda_1)$};
	\node[red] (a') at (2, 0.8) {$1$};
	\node[red] (b') at (4, 0.8) {$1$};
	\node[red] (c') at (6, -0.7) {$0$};
	\node[red] (d') at (8, -0.7) {$0$};
	\node[red] (e') at (4, -2.2) {$0$};
	\node[red] (f') at (2, -2.2) {$0$};
	\node[red] (g') at (0, -2.2) {$3$};
	
	\path (a) edge[loop left] node[left] {$\stackrel{\playcircle}{1} \stackrel{\Box}{3}$} (a);
	\path (c) edge[loop left] node[left] {$\stackrel{\playcircle}{0} \stackrel{\Box}{0}$} (c);
	\path (d) edge[loop above] node[above] {$\stackrel{\playcircle}{0} \stackrel{\Box}{0}$} (d);
	\path (g) edge[loop left] node[left] {$\stackrel{\playcircle}{3} \stackrel{\Box}{2}$} (g);
	\path[<->] (a) edge node[above] {$\stackrel{\playcircle}{0} \stackrel{\Box}{0}$} (b);
	\path[<->] (b) edge node[above right] {$\stackrel{\playcircle}{2} \stackrel{\Box}{3}$} (c);
	\path[<->] (c) edge node[above] {$\stackrel{\playcircle}{1} \stackrel{\Box}{3}$} (d);
	\path[<->] (c) edge node[below right] {$\stackrel{\playcircle}{0} \stackrel{\Box}{0}$} (e);
	\path[<->] (e) edge node[above] {$\stackrel{\playcircle}{0} \stackrel{\Box}{0}$} (f);
	\path[<->] (f) edge node[above] {$\stackrel{\playcircle}{0} \stackrel{\Box}{0}$} (g);
	\end{tikzpicture}
\end{center}

At the first step, the requirement $\lambda_1$ captures the antagonistic values.

Then, from the state $c$, if player $\Box$ forces the access to the state $b$, then player $\Circle$ must get at least $1$: the worst play that can be proposed to player $\Box$ is then $(babc)^\omega$, which gives player $\Box$ the payoff $\frac{3}{2}$.

From the state $f$, if player $\Box$ forces the access to the state $g$, then the worst play that can be proposed to them is $g^\omega$.

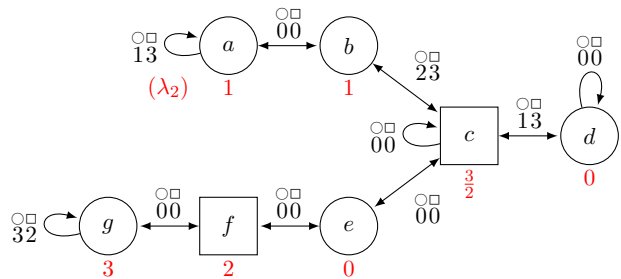
\begin{center}
	\begin{tikzpicture}[->,>=latex,shorten >=1pt, scale=0.8, every node/.style={scale=0.8}]
	\node[state] (a) at (2, 1.5) {$a$};
	\node[state] (b) at (4, 1.5) {$b$};
	\node[state, rectangle] (c) at (6, 0) {$c$};
	\node[state] (d) at (8, 0) {$d$};
	\node[state] (e) at (4, -1.5) {$e$};
	\node[state, rectangle] (f) at (2, -1.5) {$f$};
	\node[state] (g) at (0, -1.5) {$g$};
	
	\node[red] (l) at (1, 0.8) {$(\lambda_2)$};
	\node[red] (a') at (2, 0.8) {$1$};
	\node[red] (b') at (4, 0.8) {$1$};
	\node[red] (c') at (6, -0.8) {$\frac{3}{2}$};
	\node[red] (d') at (8, -0.7) {$0$};
	\node[red] (e') at (4, -2.2) {$0$};
	\node[red] (f') at (2, -2.2) {$2$};
	\node[red] (g') at (0, -2.2) {$3$};
	
	\path (a) edge[loop left] node[left] {$\stackrel{\playcircle}{1} \stackrel{\Box}{3}$} (a);
	\path (c) edge[loop left] node[left] {$\stackrel{\playcircle}{0} \stackrel{\Box}{0}$} (c);
	\path (d) edge[loop above] node[above] {$\stackrel{\playcircle}{0} \stackrel{\Box}{0}$} (d);
	\path (g) edge[loop left] node[left] {$\stackrel{\playcircle}{3} \stackrel{\Box}{2}$} (g);
	\path[<->] (a) edge node[above] {$\stackrel{\playcircle}{0} \stackrel{\Box}{0}$} (b);
	\path[<->] (b) edge node[above right] {$\stackrel{\playcircle}{2} \stackrel{\Box}{3}$} (c);
	\path[<->] (c) edge node[above] {$\stackrel{\playcircle}{1} \stackrel{\Box}{3}$} (d);
	\path[<->] (c) edge node[below right] {$\stackrel{\playcircle}{0} \stackrel{\Box}{0}$} (e);
	\path[<->] (e) edge node[above] {$\stackrel{\playcircle}{0} \stackrel{\Box}{0}$} (f);
	\path[<->] (f) edge node[above] {$\stackrel{\playcircle}{0} \stackrel{\Box}{0}$} (g);
	\end{tikzpicture}
\end{center}

Then, from the state $d$, if player $\Circle$ forces the access to the state $c$, then player $\Box$ must get at least $\frac{3}{2}$: the worst play that can be proposed to player $\Circle$ is then $(cccd)^\omega$, which gives player $\Circle$ the payoff $\frac{1}{2}$.

At the same time, from the state $e$, player $\Circle$ can now force the acces to the state $f$: then, the worst play that can be proposed to them is $fg^\omega$.

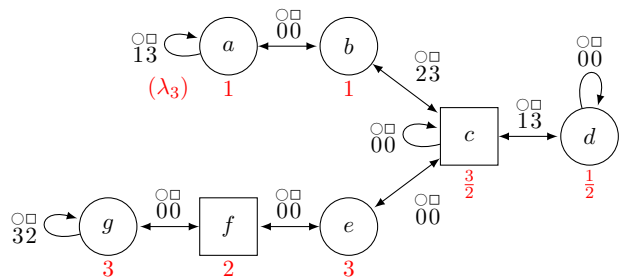
\begin{center}
	\begin{tikzpicture}[->,>=latex,shorten >=1pt, scale=0.8, every node/.style={scale=0.8}]
	\node[state] (a) at (2, 1.5) {$a$};
	\node[state] (b) at (4, 1.5) {$b$};
	\node[state, rectangle] (c) at (6, 0) {$c$};
	\node[state] (d) at (8, 0) {$d$};
	\node[state] (e) at (4, -1.5) {$e$};
	\node[state, rectangle] (f) at (2, -1.5) {$f$};
	\node[state] (g) at (0, -1.5) {$g$};
	
	\node[red] (l) at (1, 0.8) {$(\lambda_3)$};
	\node[red] (a') at (2, 0.8) {$1$};
	\node[red] (b') at (4, 0.8) {$1$};
	\node[red] (c') at (6, -0.8) {$\frac{3}{2}$};
	\node[red] (d') at (8, -0.8) {$\frac{1}{2}$};
	\node[red] (e') at (4, -2.2) {$3$};
	\node[red] (f') at (2, -2.2) {$2$};
	\node[red] (g') at (0, -2.2) {$3$};
	
	\path (a) edge[loop left] node[left] {$\stackrel{\playcircle}{1} \stackrel{\Box}{3}$} (a);
	\path (c) edge[loop left] node[left] {$\stackrel{\playcircle}{0} \stackrel{\Box}{0}$} (c);
	\path (d) edge[loop above] node[above] {$\stackrel{\playcircle}{0} \stackrel{\Box}{0}$} (d);
	\path (g) edge[loop left] node[left] {$\stackrel{\playcircle}{3} \stackrel{\Box}{2}$} (g);
	\path[<->] (a) edge node[above] {$\stackrel{\playcircle}{0} \stackrel{\Box}{0}$} (b);
	\path[<->] (b) edge node[above right] {$\stackrel{\playcircle}{2} \stackrel{\Box}{3}$} (c);
	\path[<->] (c) edge node[above] {$\stackrel{\playcircle}{1} \stackrel{\Box}{3}$} (d);
	\path[<->] (c) edge node[below right] {$\stackrel{\playcircle}{0} \stackrel{\Box}{0}$} (e);
	\path[<->] (e) edge node[above] {$\stackrel{\playcircle}{0} \stackrel{\Box}{0}$} (f);
	\path[<->] (f) edge node[above] {$\stackrel{\playcircle}{0} \stackrel{\Box}{0}$} (g);
	\end{tikzpicture}
\end{center}

But then, from the state $c$, player $\Box$ can now force the access to the state $e$: then, the worst play that can be proposed to them is $efg^\omega$.

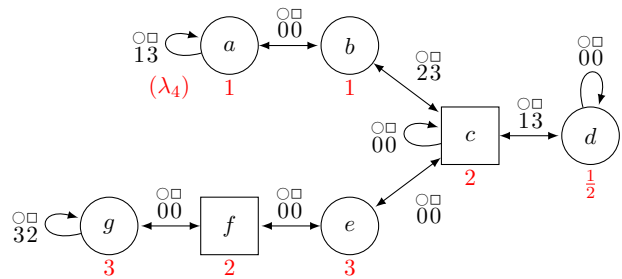
\begin{center}
	\begin{tikzpicture}[->,>=latex,shorten >=1pt, scale=0.8, every node/.style={scale=0.8}]
	\node[state] (a) at (2, 1.5) {$a$};
	\node[state] (b) at (4, 1.5) {$b$};
	\node[state, rectangle] (c) at (6, 0) {$c$};
	\node[state] (d) at (8, 0) {$d$};
	\node[state] (e) at (4, -1.5) {$e$};
	\node[state, rectangle] (f) at (2, -1.5) {$f$};
	\node[state] (g) at (0, -1.5) {$g$};
	
	\node[red] (l) at (1, 0.8) {$(\lambda_4)$};
	\node[red] (a') at (2, 0.8) {$1$};
	\node[red] (b') at (4, 0.8) {$1$};
	\node[red] (c') at (6, -0.7) {$2$};
	\node[red] (d') at (8, -0.8) {$\frac{1}{2}$};
	\node[red] (e') at (4, -2.2) {$3$};
	\node[red] (f') at (2, -2.2) {$2$};
	\node[red] (g') at (0, -2.2) {$3$};
	
	\path (a) edge[loop left] node[left] {$\stackrel{\playcircle}{1} \stackrel{\Box}{3}$} (a);
	\path (c) edge[loop left] node[left] {$\stackrel{\playcircle}{0} \stackrel{\Box}{0}$} (c);
	\path (d) edge[loop above] node[above] {$\stackrel{\playcircle}{0} \stackrel{\Box}{0}$} (d);
	\path (g) edge[loop left] node[left] {$\stackrel{\playcircle}{3} \stackrel{\Box}{2}$} (g);
	\path[<->] (a) edge node[above] {$\stackrel{\playcircle}{0} \stackrel{\Box}{0}$} (b);
	\path[<->] (b) edge node[above right] {$\stackrel{\playcircle}{2} \stackrel{\Box}{3}$} (c);
	\path[<->] (c) edge node[above] {$\stackrel{\playcircle}{1} \stackrel{\Box}{3}$} (d);
	\path[<->] (c) edge node[below right] {$\stackrel{\playcircle}{0} \stackrel{\Box}{0}$} (e);
	\path[<->] (e) edge node[above] {$\stackrel{\playcircle}{0} \stackrel{\Box}{0}$} (f);
	\path[<->] (f) edge node[above] {$\stackrel{\playcircle}{0} \stackrel{\Box}{0}$} (g);
	\end{tikzpicture}
\end{center}

And finally, from that point, if from the state $d$ player $\Circle$ forces the access to the state $c$, then player $\Box$ must have at least the payof $2$; and therefore, the worst play that can be proposed to player $\Circle$ is now $(ccd)^\omega$, which gives her the payoff $\frac{2}{3}$.

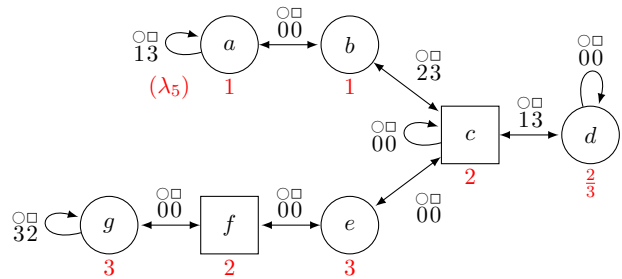
\begin{center}
	\begin{tikzpicture}[->,>=latex,shorten >=1pt, scale=0.8, every node/.style={scale=0.8}]
	\node[state] (a) at (2, 1.5) {$a$};
	\node[state] (b) at (4, 1.5) {$b$};
	\node[state, rectangle] (c) at (6, 0) {$c$};
	\node[state] (d) at (8, 0) {$d$};
	\node[state] (e) at (4, -1.5) {$e$};
	\node[state, rectangle] (f) at (2, -1.5) {$f$};
	\node[state] (g) at (0, -1.5) {$g$};
	
	\node[red] (l) at (1, 0.8) {$(\lambda_5)$};
	\node[red] (a') at (2, 0.8) {$1$};
	\node[red] (b') at (4, 0.8) {$1$};
	\node[red] (c') at (6, -0.7) {$2$};
	\node[red] (d') at (8, -0.8) {$\frac{2}{3}$};
	\node[red] (e') at (4, -2.2) {$3$};
	\node[red] (f') at (2, -2.2) {$2$};
	\node[red] (g') at (0, -2.2) {$3$};
	
	\path (a) edge[loop left] node[left] {$\stackrel{\playcircle}{1} \stackrel{\Box}{3}$} (a);
	\path (c) edge[loop left] node[left] {$\stackrel{\playcircle}{0} \stackrel{\Box}{0}$} (c);
	\path (d) edge[loop above] node[above] {$\stackrel{\playcircle}{0} \stackrel{\Box}{0}$} (d);
	\path (g) edge[loop left] node[left] {$\stackrel{\playcircle}{3} \stackrel{\Box}{2}$} (g);
	\path[<->] (a) edge node[above] {$\stackrel{\playcircle}{0} \stackrel{\Box}{0}$} (b);
	\path[<->] (b) edge node[above right] {$\stackrel{\playcircle}{2} \stackrel{\Box}{3}$} (c);
	\path[<->] (c) edge node[above] {$\stackrel{\playcircle}{1} \stackrel{\Box}{3}$} (d);
	\path[<->] (c) edge node[below right] {$\stackrel{\playcircle}{0} \stackrel{\Box}{0}$} (e);
	\path[<->] (e) edge node[above] {$\stackrel{\playcircle}{0} \stackrel{\Box}{0}$} (f);
	\path[<->] (f) edge node[above] {$\stackrel{\playcircle}{0} \stackrel{\Box}{0}$} (g);
	\end{tikzpicture}
\end{center}

The requirement $\lambda_5$ is a fixed point of the negotiation function.
\end{exa}
\end{document}